%% file: main.tex
\pdfoutput=1
\newif\ifarxiv
\arxivtrue

  \documentclass{lmcs} %
\pdfoutput=1
\usepackage[utf8]{inputenc}

\usepackage{lastpage}
\lmcsdoi{22}{4}{6}
\lmcsheading{}{\pageref{LastPage}}{}{}%
{Jan.~22,~2025}{Oct.~20,~2026}{}

\usepackage{hyperref}
\usepackage{cleveref}
\usepackage{dutchcal} %
\usepackage{xspace} %
\usepackage{bbding} %
\usepackage{stmaryrd,amssymb}
\usepackage{algorithm2e}
\AtBeginDocument{%
  \providecommand\BibTeX{{%
    \normalfont B\kern-0.5em{\scshape i\kern-0.25em b}\kern-0.8em\TeX}}}

\input{kl-config}
\input{kl-complexity}

\input{kl-knowledges}

\input{macros}
\renewcommand{\phi}{\varphi}
\allowdisplaybreaks

\newtheorem{open}[thm]{Open question}
\newtheorem{claim}[thm]{Claim}
\NewCommandCopy{\proofqedsymbol}{\qedsymbol}%
\newcommand{\exampleqedsymbol}{{$\triangle$}}%
\AtBeginEnvironment{proof}{\renewcommand{\qedsymbol}{\proofqedsymbol}}
\AtBeginEnvironment{example}{%
  \pushQED{\qed}\renewcommand{\qedsymbol}{\exampleqedsymbol}%
}
\AtEndEnvironment{example}{\popQED\endexample}
\AtBeginEnvironment{exa}{%
  \pushQED{\qed}\renewcommand{\qedsymbol}{\exampleqedsymbol}%
}
\AtEndEnvironment{exa}{\popQED\endexample}

\crefname{thm}{Theorem}{Theorems}
\crefname{thmC}{Theorem}{Theorems}

\crefname{defi}{Definition}{Definitions}

\let\endexample\endexa
\crefname{exa}{Example}{Examples}

\crefname{rem}{Remark}{Remarks}

\crefname{obs}{Observation}{Observations}

\crefname{cor}{Corollary}{Corollaries}

\crefname{lemma}{Lemma}{Lemmata}
\crefname{lem}{Lemma}{Lemmata}

\crefname{prop}{Proposition}{Propositions}

\crefname{claim}{Claim}{Claims}
\crefname{clm}{Claim}{Claims}
\crefname{fact}{Fact}{Facts}
\crefname{nota}{Notation}{Notations}
\crefname{qu}{Question}{Questions}

\newcommand{\tup}[1]{\langle #1 \rangle}

\begin{document}
\title{}
\title[A Common Ancestor of PDL, Conjunctive Queries, and UNFO]{A Common Ancestor of PDL, Conjunctive Queries, \texorpdfstring{\\}{}and Unary Negation First-order Logic}

\author[D.~Figueira]{Diego Figueira\lmcsorcid{0000-0003-0114-2257}}[a]
\author[S.~Figueira]{Santiago Figueira\lmcsorcid{0000-0002-8055-397X}}[b]
\address{Univ. Bordeaux, CNRS, Bordeaux INP, LaBRI, UMR5800, F-33400 Talence, France}
\address{Universidad de Buenos Aires, FCEN, DC \& CONICET-UBA, ICC, Argentina}

  \begin{abstract}
  \input{abstract}
  \end{abstract}
  \maketitle

\bigskip
\input{knowledge-notice}
\bigskip
\input{conf-notice}
\newpage
\tableofcontents

\input{intro}

\input{prelim}

\input{pdl-intro}

\input{translation}

\input{characterization}

\input{separation}

\input{satisfiability-altbis.tex}
\input{UNTC-intro}

\input{UNTCequivUCPDL}
\input{UNTCsat}
\input{conclusions}

\section*{Acknowledgements}
This work was partially funded by the French-Argentinian IRP
\href{http://www.irp-sinfin.org/}{SINFIN}. Diego Figueira is partially funded by ANR EXPAND, grant ANR-25-CE23-1215. Santiago Figueira is partially funded by UBACyT 20020250300015BA, PICT-2021-I-A-00838, and PIP 2025-2027 GI 11220250100228CO.

  \bibliographystyle{alphaurl}
  \bibliography{long,references}

\end{document}

%% file: kl-config.tex
\usepackage[xcolor, hyperref, cleveref, notion, quotation, electronic]{knowledge}
\usepackage{mathcommand}
\knowledgeconfigure{quotation, protect quotation={tikzcd}}
\knowledgeconfigure{diagnose line=true, diagnose bar=true}

\definecolor{Dark Ruby Red}{HTML}{7c1b1e}
\definecolor{Dark Blue Sapphire}{HTML}{004452} %
\definecolor{Dark Gamboge}{HTML}{be7c00}

\IfKnowledgePaperModeTF{
}{
    \knowledgestyle{intro notion}{color={Dark Ruby Red}, emphasize}
    \knowledgestyle{notion}{color={Dark Blue Sapphire}}
    \hypersetup{
        colorlinks=true,
        breaklinks=true,
        linkcolor={Dark Blue Sapphire}, %
        citecolor={Dark Blue Sapphire}, %
        filecolor={Dark Blue Sapphire}, %
        urlcolor={Dark Blue Sapphire},
    }
    \IfKnowledgeElectronicModeTF{
    }{
        \knowledgeconfigure{anchor point color={Dark Ruby Red}, anchor point shape=corner}
        \knowledgestyle{intro unknown}{color={Dark Gamboge}, emphasize}
        \knowledgestyle{intro unknown cont}{color={Dark Gamboge}, emphasize}
        \knowledgestyle{kl unknown}{color={Dark Gamboge}}
        \knowledgestyle{kl unknown cont}{color={Dark Gamboge}}
    }
}

%% file: kl-complexity.tex
\knowledge{wrap=\textsf}
  | NL

\knowledge{notion, text={paraNL}, wrap=\textsf}
  | para-NL
  | paraNL

\knowledge{notion, wrap=\textsf}
  | FPT
  | fixed-parameter tractable

\knowledge{text={ExpSpace}, wrap=\textsf}
  | EXPSPACE
  | ExpSpace

\knowledge{text={2ExpTime}, wrap=\textsf}
  | 2ExpTime

\knowledge{text={3ExpTime}, wrap=\textsf}
  | 3ExpTime

\knowledge{text={ExpTime}, wrap=\textsf}
  | ExpTime

\knowledge{text={2ExpSpace}, wrap=\textsf}
  | 2EXPSPACE
  | 2ExpSpace

\knowledge{text={PSpace}, wrap=\textsf}
  | PSpace
  | PSPACE

\knowledge{text={NP}, wrap=\textsf}
  | NP

\knowledge{text={\ensuremath{\Pi^p_2}}, wrap=\textsf}
  | PiP2
  | Pi2
  | pi2
  | Pitwo
  | PiTwo
  | pitwo

  \knowledge{text={\ensuremath{\textup{P}^{\textup{NP}[O(\log^2 n)]}}}, wrap=\textsf}
  | PNPlog2

\knowledge{url={https://en.wikipedia.org/wiki/Savitch\%27s_theorem}}
  | Savitch's Theorem

%% file: kl-knowledges.tex
\knowledge{text={i.e.}, italic}
  | ie

\knowledge{text={s.t.}, italic}
  | st

\knowledge{text={e.g.}, italic}
  | eg

\knowledge{text={vs.}, italic}
  | vs

\knowledge{text={w.r.t.}, italic}
  | wrt

\knowledge{text={a.k.a.}, italic}
  | aka

\knowledge{text={w.l.o.g.}, italic}
  | wlog

\knowledge{text={W.l.o.g.}, italic}
  | Wlog

\knowledge{text={cf.}, italic}
  | cf

\knowledge{text={iff}, italic}
  | iff

\knowledge{url={https://complexityzoo.net/Complexity_Zoo:L\#logcfl}}
| LOGCFL

\knowledge{notion}
 | full version

\knowledge{notion}
 | notion@notice
 | definition@notice

\knowledge{notion}
 | $\omega$-regular tree satisfiability
 
 \knowledge{notion}
 | successful

 \knowledge{notion}
 | tree
 | trees

\knowledge{notion}
 | $\SigmaN$-labeled $\SigmaE$-tree
 | $\SigmaN$-labeled $\SigmaE$-trees
 | $\pset{P}$-labeled $A$-trees
 | $\SigmaN$-labeled tree
 | $\SigmaN$-labeled trees

\knowledge{notion}
 | two-way alternating parity tree automaton
 | TWAPTA
 | TWAPTAs

\knowledge{notion}
 | satisfiable with respect to@twapta

 \knowledge{notion}
 | index

 \knowledge{notion}
 | size@untc

 \knowledge{notion}
 | size@pdl

 \knowledge{notion}
 | star-depth

 \knowledge{notion}
 | state size

\knowledge{notion}
 | transition size

 \knowledge{notion}
 | non-deterministic finite automaton
 | NFA
 | NFAs
\knowledge{notion}
 | initial@NFA
 \knowledge{notion}
 | final@NFA
\knowledge{notion}
 | graph
 | graphs

\knowledge{notion}
 | vertex
 | vertices

\knowledge{notion}
 | edge
 | edges

\knowledge{notion}
 | minor
 | minors

 \knowledge{notion}
 | underlying graph@C

\knowledge{notion}
 | underlying graph
 | underlying graphs

\knowledge{notion}
 | bi-pointed multigraph
 | bi-pointed multigraphs
 | 2-graph
 | 2-graphs

\knowledge{notion}
 | label

\knowledge{notion}
 | variable
 | variables

\knowledge{notion}
 | formula
 | formulas
 | Formulas

 \knowledge{notion}
 | program
 | programs

 \knowledge{notion}
 | relation names
 | relation name
 | relation

 \knowledge{notion}
 | expression
 | expressions
 | Expressions

\knowledge{notion}
 | positive

\knowledge{notion}
 | $C$-satisfying assignment
 | $C'$-satisfying assignment

\knowledge{notion}
 | atomic propositions
 | atomic proposition

\knowledge{notion}
 | atomic programs
 | atomic program

\knowledge{notion}
 | program intersection
 | program intersections

\knowledge{notion}
 | $\ICPDL$ program over $\allStates$
 | $\ICPDL$ programs over $\allStates$
 | $\ICPDL$ programs over
 | $\ICPDL$ program over

\knowledge{notion}
 | $a$-successor

\knowledge{notion}
 | $a$-predecessor

\knowledge{notion}
 | intersection width

 \knowledge{notion}
 | conjunctive width

 \knowledge{notion}
 | NF-conjunctive width

 \knowledge{notion}
 | conjunctive programs
 | conjunctive program

 \knowledge{notion}
 | arity

\knowledge{notion}
 | $\Rels$-structure
 | $\Rels$-structures
 | structure
 | structures

\knowledge{notion}
 | domain

\knowledge{notion}
 | Gaifman graph

\knowledge{notion}
 | domain elements
 | domain element

 \knowledge{notion}
 | Kripke structure
 | Kripke structures
 | structure@kripke
 | structures@kripke

 \knowledge{notion}
 | finite degree

 \knowledge{notion}
 | distance

 \knowledge{notion}
 | dimension

\knowledge{notion}
 | world
 | worlds

\knowledge{notion}
 | $\Tw$-model property
 | $\Tw[1]$-model property

 \knowledge{notion}
 | countable model property

 \knowledge{notion}
 | $\+G$-CPDL

\knowledge{notion}
 | CPDL

\knowledge{notion}
 | ICPDL

\knowledge{notion}
 |loop-CPDL

\knowledge{notion}
 | $k$-ary simulation relation
 | $k$-simulation
 | $k$-simulates

 \knowledge{notion}
 | $k$-ary bisimulation relation
 | $k$-bisimulation
 | $k$-bisimilar

 \knowledge{notion}
 | \quasibisim

\knowledge{notion}
 | graph database
 | graph databases
 | Graph databases
 | database
 | databases

\knowledge{notion}
 | evaluation
 | evaluated
 | evaluations
 | evaluates
 | evaluate

\knowledge{notion}
 | model checking

\knowledge{notion}
 | atom
 | atoms

 \knowledge{notion}
 | p-atom
 | p-atoms

\knowledge{notion}
 | r-atom
 | r-atoms

\knowledge{notion}
 | output variables

\knowledge{notion}
 | Boolean

 \knowledge{notion}
 | satisfiability problem
 | satisfiability

\knowledge{notion}
 | satisfies
 | satisfy
 | satisfying

\knowledge{notion}
 | conjunctive query
 | conjunctive queries
 | Conjunctive queries
 | CQ
 | CQs

 \knowledge{notion}
 | CQPDL
 \knowledge{notion}
 | UCQPDL

 \knowledge{notion}
 | Regular Queries
 | Regular queries
 | regular queries

\knowledge{notion}
 | unions of CQs
 | union of CQs
 | UCQs

\knowledge{notion}
 | CRPQ
 | CRPQs
 | conjunctive regular path queries

\knowledge{notion}
 | C2RPQ
 | C2RPQs
 | conjunctive two-way regular path queries
 | conjunctive two-way regular path query
 | CRPQ with two-way navigation
 | two-way navigation

\knowledge{notion}
 | UC2RPQ
 | (U)C2RPQ
 | UC2RPQs
 | unions of C2RPQs
 | union of C2RPQs

\knowledge{notion}
 | UCRPQ
 | (U)CRPQ
 | UCRPQs
 | unions of CRPQs
 | union of CRPQs

\knowledge{notion}
 | connected

\knowledge{notion}
 | containment
 | contained
 | contains
 | contain

\knowledge{notion}
 | containment problem

\knowledge{notion}
 | equi-expressive
 
\knowledge{notion}
 | equivalent
 | equivalence
 | semantical equivalence
 | semantically equivalent

\knowledge{notion}
 | expansion
 | expansions

\knowledge{notion}
 | partial homomorphism
 | partial homomorphisms

 \knowledge{notion}
 | homomorphism
 | homomorphisms
 | Homomorphisms
 | homomorphic
 | C2RPQs homomorphisms
 
\knowledge{notion}
 | strong onto homomorphism
 | strong onto homomorphisms
 | strong onto

\knowledge{notion}
 | sublanguage
 | sublanguages

\knowledge{notion}
 | closed under sublanguages
 | closure under sublanguages

\knowledge{notion}
 | atom refinement
 | atom refinements

\knowledge{notion}
 | atom $m$-refinement
 | atom $m$-refinements
 | $m$-refinements@atom

\knowledge{notion}
 | equality atom
 | equality atoms

\knowledge{notion}
 | contraction
 | contractions
 | contract

\knowledge{notion}
 | refinements
 | refinement

 \knowledge{notion}
 | $m$-refinements
 | $m$-refinement

\knowledge{notion}
 | refinement length

\knowledge{notion}
 | tree decomposition
 | tree decompositions

 \knowledge{notion}
 | good

\knowledge{notion}
 | A@treedec
\knowledge{notion}
 | B@treedec
\knowledge{notion}
 | C@treedec

\knowledge{notion}
 | full bag
 | full
 | full bags

\knowledge{notion}
 | path decomposition
 | path decompositions

\knowledge{notion}
 | bag
 | bags

\knowledge{notion}
 | contains@tw
 | contained@tw
 | containing@tw

\knowledge{notion}
 | width

\knowledge{notion}
 | tree-width
 | Tree-width

 \knowledge{notion}
 | tree-width@structure

 \knowledge{notion}
 | tree-width@pdl

\knowledge{notion}
 | nice tree decomposition
 | niceness
 | nice
 | nice tree decompositions

\knowledge{notion}
 | tagged tree decomposition
 | tagged tree decompositions

\knowledge{notion}
 | tagged
 | tagging
 | tag

\knowledge{notion}
 | nice tagged tree decomposition
 | nice tagged tree-decom\-position

\knowledge{notion}
 | cyclic@path
 | acyclic@path
 | cyclic path
 | acyclic path
 | acyclic paths

\knowledge{notion}
 | locally acyclic
 | local acyclicity

 \knowledge{notion}
 | simple path

\knowledge{notion}
 | non-branching path
 | non-branching paths

\knowledge{notion}
 | Path induced
 | path induced
 | induced path
 | induces
 | induce
 | induced

\knowledge{notion}
 | leaves
 | leaving

\knowledge{notion}
 | semantic tree-width $k$ problem
 | semantic tree-width $1$ problem

\knowledge{notion}
 | semantic tree-width

\knowledge{notion}
 | maximal under-approximation by infinitary unions
 | maximal under-approximation
 | maximal under-approximations

\knowledge{notion}
 | approximations

\knowledge{notion}
 | infinitary unions
 | infinitary union
 | Infinitary unions

\knowledge{notion}
 | trio

\knowledge{notion}
 | key lemma
 | Key lemma

\knowledge{notion}
 | atomic
 | non-atomic
 | non-atomic bags

\knowledge{notion}
 | profile

\knowledge{notion}
 | type
 | types

 \knowledge{notion}
 | $k$-summary query
 | $k$-summary queries
 | summary query
 | summary queries

 \knowledge{notion}
 | simple regular expression
 | simple regular expressions
 | SRE
 | SREs

 \knowledge{notion}
 | path-$l$ approximation
 | path-$l$ approximations

\knowledge{notion}
 | recursive atom
 | recursive atoms
 | recursive@rat

\knowledge{notion}
 | non-recursive atom
 | non-recursive atoms
 | non-recursive@nrat

\knowledge{notion}
 | simple normal form

\knowledge{notion}
 | simple 

\knowledge{notion}
 | complete

 \knowledge{notion}
 | edge labels

 \knowledge{notion}
 | node labels
 | node label

 \knowledge{notion}
 | $k$-valid

\knowledge{notion}
 | labeling

\knowledge{notion}
 | at least as expressive as

\knowledge{notion}
 | $(k+1)$-clique formula
 | $3$-clique formula

\knowledge{notion}
 | extensible clique
 | extensible cliques

\knowledge{notion}
 | ML-simulation
 | ML-simulates

\knowledge{notion}
 | subexpression
 | subexpressions

\knowledge{notion}
 | is a clique in

\knowledge{notion}
 | composition
 
\knowledge{notion}
 | witnessing path

\knowledge{notion}
 | switching node
 | switching nodes

\knowledge{notion}
 | switching edge
 | switching edges

\knowledge{notion}
 | jumping

\knowledge{notion}
 | $(s,u)$-run

\knowledge{notion}
 | Unary negation first-order logic

\knowledge{notion}
 | Guarded negation first-order logic

\knowledge{notion}
 | Guarded quantification first-order logic

\knowledge{notion}
 | Unary negation fragment of first-order logic extended with regular path expressions

\knowledge{notion}
 | NF-atom
 | NF-atoms

\knowledge{notion}
 | NF-conjunction
 | NF-conjunctions

\knowledge{notion}
 | decomposition
 | decompositions

\knowledge{notion}
 | atom width

\knowledge{notion}
 | normal form

\knowledge{notion}
 | strong normal form

\knowledge{notion}
 | underlying graph@U

\knowledge{notion}
 | connected@U

\knowledge{notion}
 | $\TC$-depth

\knowledge{notion}
 | $*$-depth

\knowledge{notion}
 | negation-depth

\knowledge{notion}
 | $k$-$\Univ$-simulation

\knowledge{notion}
 | $k$-$\Univ$-bisimulation
 | $k$-$\Univ$-bisimilar

\knowledge{notion}
 | $k$-$\Univ$-simulates

\knowledge{notion}
 | \quasibisimE

\knowledge{notion}
 | universal program
 | universal programs

\knowledge{notion}
 | test program
 | test programs

 \knowledge{notion}
 | unnested form

\knowledge{notion}
 | nesting depth

\knowledge{notion}
 | processed
 | unprocessed

\knowledge{notion}
  | unprocessed leaf

\knowledge{notion}
 | test labels
 | test label

\knowledge{notion}
 | equi-satisfiable

%% file: macros.tex
\knowledgenewrobustcmd{\expfun}{\cmdkl{\textit{exp}}}
\knowledgenewrobustcmd{\polyfun}{\cmdkl{\textit{poly}}}

\knowledgenewrobustcmd{\uGraph}[1]{\cmdkl{\ensuremath{\textup{\bf G}_{#1}}}}
\knowledgenewrobustcmd{\uGraphU}[1]{\cmdkl{\ensuremath{\textup{\bf G}_{#1}}}}
\knowledgenewrobustcmd{\uGraphC}[1]{\cmdkl{\ensuremath{\textup{\bf G}_{#1}}}}%
\knowledgenewrobustcmd{\SigmaE}[1][]{\cmdkl{\Sigma}^{#1}_{\cmdkl{E}}}
\knowledgenewrobustcmd{\barSigmaE}[1][]{\cmdkl{\overline{\Sigma}}^{#1}_{\cmdkl{E}}}
\knowledgenewrobustcmd{\SigmaN}{\cmdkl{\Sigma_{N}}}
\knowledgenewrobustcmd{\domT}{\cmdkl{\textit{dom}}}
\knowledgenewrobustcmd{\tree}{\cmdkl{\textit{tree}}}
\knowledgenewrobustcmd{\tcdepth}{\cmdkl{\textit{\TC-depth}}}
\knowledgenewrobustcmd{\posB}{\cmdkl{B^+}}
\knowledgenewrobustcmd{\translatedC}{\cmdkl{\hat \pi_{T_C,f}}}%
\knowledgenewrobustcmd{\negdepth}{\cmdkl{\textit{neg-depth}}}
\knowledgenewrobustcmd{\ICPDLp}{\ensuremath{\cmdkl{\textup{ICPDL}^{\!+}}}\xspace}
\knowledgenewrobustcmd{\IUCPDLp}{\ensuremath{\cmdkl{\textup{IUCPDL}^{\!+}}}\xspace}
\knowledgenewrobustcmd{\allStates}{\cmdkl{\+Q}}

\knowledgenewrobustcmd{\kSimGame}[1][k]{\cmdkl{\mathbf{G}[\pebblesim{#1}\!]}}%
\knowledgenewrobustcmd{\kBisimGame}[1][k]{\cmdkl{\mathbf{G}[\pebblebisim{#1}\!]}}%
\knowledgenewrobustcmd{\kSimGameE}[1][k]{\cmdkl{\mathbf{G}^{\Univ}[\pebblesim{#1}\!]}}%
\knowledgenewrobustcmd{\kBisimGameE}[1][k]{\cmdkl{\mathbf{G}^{\Univ}[\pebblebisim{#1}\!]}}%

\renewcommand{\epsilon}{\varepsilon}
\newif\ifproofappendix

\newrobustcmd\labelwithproof[1]{%
\label{#1}%
\ifproofappendix%
\marginnote{\footnotesize{%
  \textnormal{First stated in page~\pageref{#1}.}%
}}
\else%
\marginnote{\footnotesize{%
  \textnormal{See the proof of \Cref{#1} in page~\pageref{proof-#1}.}%
}}%
\fi%
}

\newrobustcmd\introinrestatable[1]{%
\ifproofappendix%
\kl{#1}%
\else%
\intro{#1}%
\fi%
}

\newrobustcmd\introinrestatableopt[1]{%
\ifproofappendix%
\kl[#1]{#1}%
\else%
\intro[#1]{#1}%
\fi%
}

\newrobustcmd\recall[1]{
  \proofappendixtrue%
    #1*
  \proofappendixfalse%
}

\usepackage[backgroundcolor=orange!20, textcolor={Dark Ruby Red}, textsize=tiny,disable]{todonotes}
\definecolor{diegochange}{RGB}{0,0,0}
\definecolor{santichange}{RGB}{0,0,0}

\definecolor{green}{RGB}{0,120,0}
\definecolor{hlyellow}{RGB}{250, 250, 190}
\definecolor{diegoeditcolor}{RGB}{210,210,255}
\definecolor{remieditcolor}{RGB}{210,255,210}
\definecolor{edwineditcolor}{RGB}{150,200,55}
\newcommand{\sidediego}[1]{\todo[backgroundcolor=diegoeditcolor, size=\tiny]{{\bf D:} #1}}
\newcommand{\sidesanti}[1]{\todo[backgroundcolor=remieditcolor, size=\tiny]{{\bf S:} #1}}

\newcommand{\santi}[1]{\todo[inline,color=remieditcolor, size=\footnotesize]{{\bf S:} #1}}

\newcommand{\diego}[1]{\todo[inline,color=diegoeditcolor, size=\footnotesize]{{\bf D:} #1}}
\definecolor{light-gray}{gray}{0.9}

\newcommand{\proofcase}[1]{\noindent\colorbox{light-gray}{#1}~~}
\newcommand{\proofsubcase}[1]{\noindent$\rhd$~\underline{~#1~}~}

\newrobustcmd{\wrote}{\color{wrote}\scriptsize\text{wrote}}
\newrobustcmd{\advised}{\color{advised}\scriptsize\text{advised}}

\renewcommand{\phi}{\varphi}

\newcommand{\set}[1]{\{#1\}}

\newrobustcmd{\Nat}{\mathbb{N}}
\newrobustcmd{\Rat}{\mathbb{Q}}
\newcommand{\dcup}{\mathop{\dot\cup}} %
\knowledgenewrobustcmd\pset[1]{\cmdkl{\wp}(#1)} %

\newcommand{\resp}[1]{(resp.~#1)}
\newcommand{\ie}{\textit{i.e.}}
\newcommand{\aka}{\textit{a.k.a.}\ }

\knowledgenewrobustcmd{\A}{\mathbb{A}} %
\knowledgenewrobustcmd{\Aext}{\cmdkl{\mathbb{A}^\pm}} %
\knowledgenewrobustcmd\vertex[1]{\cmdkl{V}(#1)}
\knowledgenewrobustcmd\edges[1]{\cmdkl{E}(#1)}

\knowledgenewrobustcmd{\Graphloop}{\cmdkl{\+G_\circlearrowleft}}
\knowledgenewrobustcmd{\Graphcap}{\cmdkl{\+G_\cap}}

\knowledgenewrobustcmd\worlds[1]{\cmdkl{W}\!(#1)} %
\knowledgenewrobustcmd\dom[1]{\cmdkl{dom}(#1)} %
\knowledgenewrobustcmd\Kripke{\cmdkl{\mathbb{K}}} %

\knowledgenewrobustcmd\qvar{\footnotesize\bullet} %

\newcommand{\eqdef}{\mathrel{{\mathop:}}=}
\newcommand{\eqqdef}{\mathrel{{\mathop:}{\mathop:}}=}

\knowledgenewrobustcmd{\kHoms}[1][k]{\cmdkl{\textit{Hom}_{#1}}}%
\knowledgenewrobustcmd{\Prop}{\cmdkl{\mathbb{P}}}
\knowledgenewrobustcmd{\Rels}{\cmdkl{\mathbf{\sigma}}}
\knowledgenewrobustcmd{\arity}{\cmdkl{\textup{arity}}}
\knowledgenewrobustcmd{\Prog}{\cmdkl{\mathbb{A}}}
\knowledgenewrobustcmd\Vars{\cmdkl{\mathbb{V}}}
\knowledgenewrobustcmd{\ICPDLg}[1]{\ensuremath{\cmdkl{\textup{ICPDL}^{\!+}}\!}\cmdkl{(}#1\cmdkl{)}}
\newcommand{\twticpdl}{\ICPDLg{\Tw[2]}}
\knowledgenewrobustcmd{\subexpr}{\ensuremath{\cmdkl{\textup{sub}}}}
\knowledgenewrobustcmd{\dbracket}[1]{\cmdkl{\llbracket} #1 \cmdkl{\rrbracket}}

\knowledgenewrobustcmd{\CPDL}{\ensuremath{\cmdkl{\textup{CPDL\xspace}}}}
\knowledgenewrobustcmd{\PDL}{\ensuremath{\cmdkl{\textup{PDL\xspace}}}}
\knowledgenewrobustcmd{\ICPDL}{\ensuremath{\cmdkl{\textup{ICPDL}}}\xspace}
\knowledgenewrobustcmd{\loopCPDL}{\ensuremath{\cmdkl{\textup{loop-CPDL}}}\xspace}
\knowledgenewrobustcmd{\CPDLp}{\cmdkl{\ensuremath{\textup{CPDL}^{\!+}}}}
\knowledgenewrobustcmd{\UCPDLp}{\cmdkl{\ensuremath{\textup{UCPDL}^{\!+}}}}
\knowledgenewrobustcmd{\CPDLg}[1]{\ensuremath{\cmdkl{\textup{CPDL}^{\!+}}\!}\cmdkl{(}#1\cmdkl{)}}
\knowledgenewrobustcmd{\UCPDLg}[1]{\ensuremath{\cmdkl{\textup{UCPDL}^{\!+}}\!}\cmdkl{(}#1\cmdkl{)}}
\knowledgenewrobustcmd{\Univ}{\ensuremath{\cmdkl{\textsf{U}}}}

\knowledgenewrobustcmd{\CPDLpp}{\cmdkl{\ensuremath{\textup{CPDL}^{\!++}}}}
\knowledgenewrobustcmd{\UNFO}{\cmdkl{\ensuremath{\textup{UNFO}}}}
\knowledgenewrobustcmd{\UNFOreg}{\cmdkl{\ensuremath{\textup{UNFO}^{{\rm reg}}}}}
\knowledgenewrobustcmd{\UNFOregOne}{\cmdkl{\ensuremath{\textup{UNFO}^{{\rm reg}}_1}}}
\knowledgenewrobustcmd{\UNTC}{\cmdkl{\ensuremath{\textup{UNTC}}}}
\knowledgenewrobustcmd{\GNTC}{\cmdkl{\ensuremath{\textup{GNTC}}}}
\knowledgenewrobustcmd{\GNFO}{\cmdkl{\ensuremath{\textup{GNFO}}}}
\knowledgenewrobustcmd{\GNFP}{\cmdkl{\ensuremath{\textup{GNFP}}}}
\knowledgenewrobustcmd{\GNFPUP}{\cmdkl{\ensuremath{\textup{GNFP-UP}}}}
\knowledgenewrobustcmd{\GFO}{\cmdkl{\ensuremath{\textup{GFO}}}}

\knowledgenewrobustcmd{\CQPDL}{\cmdkl{\ensuremath{\textup{CQPDL}}}}
\knowledgenewrobustcmd{\UCQPDL}{\cmdkl{\ensuremath{\textup{UCQPDL}}}}

\knowledgenewrobustcmd{\iwidth}{\cmdkl{\textup{iw}}}
\knowledgenewrobustcmd{\Iwidth}{\cmdkl{\textup{IW}}}

\knowledgenewrobustcmd{\cqsize}[1]{\cmdkl{\textup{cw}}(#1)}
\knowledgenewrobustcmd{\Cqsize}[1]{\cmdkl{\textup{CW}}(#1)}
\knowledgenewrobustcmd{\Cqsizealt}{\cmdkl{\textup{CW}}}

\knowledgenewrobustcmd{\untcsize}[1]{\cmdkl{\|}#1\cmdkl{\|}}

\knowledgenewrobustcmd{\cqsizenf}[1]{\cmdkl{\textup{cw}_{\rm NF}}(#1)}
\knowledgenewrobustcmd{\Cqsizenf}[1]{\cmdkl{\textup{CW}_{\rm NF}}(#1)}

\knowledgenewrobustcmd{\dbracketaut}[1]{\cmdkl{\llbracket} #1 \cmdkl{\rrbracket}}
\knowledgenewrobustcmd{\dbracketnfa}[1]{\cmdkl{\llbracket} #1 \cmdkl{\rrbracket}}
\knowledgenewrobustcmd{\dbrackett}[1]{\cmdkl{\llbracket} #1 \cmdkl{\rrbracket}}
\knowledgenewrobustcmd{\indexAut}[1]{\cmdkl{\textbf{\textit{i}}(}#1\cmdkl{)}}
\knowledgenewrobustcmd{\sizeStates}[1]{\cmdkl{|}#1\cmdkl{|_{\textup{st}}}}
\knowledgenewrobustcmd{\sizeTrans}[2][]{\cmdkl{|}#2\cmdkl{|}^{#1}_{\cmdkl{\textup{tr}}}}
\knowledgenewrobustcmd{\sizeStatesNFA}[1]{\cmdkl{|}#1\cmdkl{|_{\textup{st}}}}

\knowledgenewrobustcmd{\pebblesimE}[1]{\mathrel{\cmdkl{\rightharpoonup^{\Univ}_{#1}}}}
\knowledgenewrobustcmd{\notpebblesimE}[1]{\mathrel{\cmdkl{{\not\rightharpoonup^{\Univ}}_{#1}}}}
\knowledgenewrobustcmd{\pebblebisimE}[1]{\mathrel{\cmdkl{\rightleftharpoons^{\Univ}_{#1}}}}
\knowledgenewrobustcmd{\pebblequasibisimE}[1]{\mathrel{\cmdkl{\rightharpoondown^{\!\!\leftrightarrow\Univ}_{#1}}}}

\knowledgenewrobustcmd{\pebblesim}[1]{\mathrel{\cmdkl{\rightharpoonup_{#1}}}}
\knowledgenewrobustcmd{\notpebblesim}[1]{\mathrel{\cmdkl{{\not\rightharpoonup}_{#1}}}}
\knowledgenewrobustcmd{\pebblebisim}[1]{\mathrel{\cmdkl{\rightleftharpoons_{#1}}}}
\knowledgenewrobustcmd{\pebblequasibisim}[1]{\mathrel{\cmdkl{\rightharpoondown^{\!\!\leftrightarrow}_{#1}}}}
\knowledgenewrobustcmd{\notpebblequasibisim}[1]{\mathrel{\cmdkl{\not\rightharpoondown^{\!\!\leftrightarrow}_{#1}}}}
\knowledgenewrobustcmd{\notpebblebisim}[1]{\mathrel{\cmdkl{{\not\rightleftharpoons}_{#1}}}}

\knowledgenewrobustcmd{\pebblesimneg}[1]{\mathrel{\cmdkl{\rightharpoonup^\lnot_{#1}}}}

\knowledgenewrobustcmd{\mlsim}{\mathrel{\cmdkl{\rightharpoonup}}}
\knowledgenewrobustcmd{\mlbisim}{\mathrel{\cmdkl{\rightleftharpoons}}}

\knowledgenewrobustcmd{\pathl}{\cmdkl{\mathbf{P}_{\!l}}} %

\knowledgenewrobustcmd\subaut[3]{#1\cmdkl{[#2,#3]}}

\knowledgenewrobustcmd\bagmap{\cmdkl{\mathbf{v}}}
\knowledgenewrobustcmd\tagmap{\cmdkl{\mathbf{t}}}
\knowledgenewrobustcmd\tagmappath[1]{\cmdkl{\mathbf{t}[#1]}}
\newrobustcmd\tagmappathprime[1]{%
  \withkl{\kl[\tagmappath]}{%
    \cmdkl{\mathbf{t}'[#1]}%
  }%
}

\knowledgenewrobustcmd{\atom}[1]{\,\xrightarrow{\smash{#1}}\,}
\knowledgenewrobustcmd{\coatom}[1]{\,\xleftarrow{\smash{#1}}\,}
\knowledgenewrobustcmd{\atoms}[1]{\cmdkl{\textnormal{Atoms}}(#1)}
\knowledgenewrobustcmd{\contained}{\mathrel{\cmdkl{\subseteqq}}}
\newrobustcmd{\strcontained}{
  \mathrel{\withkl{\kl[\contained]}{\cmdkl{%
    \subsetneqq
  }}}
}
\disablecommand\equiv
\suggestcommand\equiv{Use instead \semequiv for semantical equivalence.}
\knowledgenewrobustcmd{\semequiv}{\mathrel{\cmdkl{\LaTeXequiv}}} %
\knowledgenewrobustcmd{\langsemequiv}[1][]{\mathrel{\cmdkl{\LaTeXequiv_{#1}}}} %
\knowledgenewrobustcmd{\lleq}[1][]{\mathrel{\cmdkl{\leqq_{#1}}}}
\knowledgenewrobustcmd{\notlleq}[1][]{\mathrel{\cmdkl{\nleqq_{#1}}}}
\knowledgenewrobustcmd{\lleqs}[1][]{\mathrel{\cmdkl{\lneqq_{#1}}}}
\knowledgenewrobustcmd{\vars}{\cmdkl{\textit{vars}}} %

\knowledgenewrobustcmd{\CRPQ}{\cmdkl{\textnormal{CRPQ}}}

\knowledgenewrobustcmd{\UCtwoRPQ}{\cmdkl{\textnormal{UC2RPQ}}}
\newrobustcmd{\CtwoRPQ}{%
  \withkl{\kl[\UCtwoRPQ]}{\cmdkl{%
    \textnormal{C2RPQ}
  }}
}

\knowledgenewrobustcmd{\UCRPQSRE}{\ensuremath{\cmdkl{\textup{UCRPQ}(\textup{SRE})}}}
\newrobustcmd{\CRPQSRE}{%
  \withkl{\kl[\UCRPQSRE]}{\cmdkl{%
    \textup{CRPQ}(\textup{SRE})
  }}
}

\newcommand{\xrightarrowdbl}[2][]{%
  \xrightarrow[#1]{#2}\mathrel{\mkern-14mu}\rightarrow
}
\knowledgenewrobustcmd\surj{%
    \mathrel{\cmdkl{%
      \xrightarrowdbl{\textit{\tiny hom}}
    }}
}
\knowledgenewrobustcmd{\fun}{f}
\knowledgenewrobustcmd{\homto}{\mathrel{\cmdkl{\xrightarrow{\textit{\tiny hom}}}}} %

\knowledgenewrobustcmd{\class}{\mathcal{C}}
\knowledgenewrobustcmd{\Tw}[1][k]{\cmdkl{\textup{TW\!}_{#1\!}}}

\knowledgenewrobustcmd{\Refin}[1][]{\cmdkl{\textnormal{Ref}^{\smash{#1}}}}
\knowledgenewrobustcmd{\MUA}[2]{\cmdkl{\ensuremath{\textnormal{App}_{#2}(#1)}}}
\knowledgenewrobustcmd{\MUAHom}[2]{\cmdkl{\ensuremath{\textnormal{App}_{#2}^{\smash{\star}}(#1)}}}
\knowledgenewrobustcmd{\MUAHomBounded}[3]{\cmdkl{\ensuremath{\textnormal{App}_{#2}^{\smash{\star,#3}}(#1)}}}
\knowledgenewrobustcmd{\type}{\cmdkl{\textnormal{type}}}
\knowledgenewrobustcmd{\Qapp}{\cmdkl{\ensuremath{\textnormal{App}_{\Tw}^{\textup{zip}}(\gamma)}}}
\knowledgenewrobustcmd{\contract}[1]{\cmdkl{[}#1\cmdkl{]}}

\newcommand{\complexityclass}[1]{\textup{\textsf{#1}}\xspace}

\newcommand{\ptime}{\complexityclass{PTime}}

\newcommand{\wone}{\complexityclass{W[1]}}
\newcommand{\fpt}{\complexityclass{FPT}}

\newrobustcmd\pitwo{\ensuremath{\Pi^p_2}}

\newrobustcmd\sigmatwo{\ensuremath{\Sigma^p_2}}

\knowledgenewrobustcmd{\tw}{\ensuremath{\cmdkl{\text{tw}}}}

\knowledgenewrobustcmd{\lo}{\cmdkl{\textnormal{loop}}\xspace}

\knowledgenewrobustcmd{\mapcoord}{\cmdkl{\mapsto}}
\knowledgenewrobustcmd{\dimtup}{\cmdkl{\dim}}
\knowledgenewrobustcmd{\transUNTCECPDL}{\ensuremath{\cmdkl{\tr_4}}}
\knowledgenewrobustcmd{\transECPDLUNTC}{\ensuremath{\cmdkl{\tr_5}}}
\knowledgenewrobustcmd{\transATOM}{\ensuremath{\cmdkl{\rm atom}}}

\newcommand{\logicOp}[1]{\textup{\textsf{#1}}}
\knowledgenewrobustcmd{\TC}{\cmdkl{\logicOp{TC}}}
\knowledgenewrobustcmd{\NFARightarrow}[2][]{\mathrel{\cmdkl{\Rightarrow}^{#1}_{\!#2}}}
\knowledgenewrobustcmd{\ShapesC}[1][C[x_s,x_t]]{\cmdkl{\+S}_{#1}}%
\knowledgenewrobustcmd{\translatedAtomsC}{\cmdkl{C_{T_C,f}}}
\knowledgenewrobustcmd{\langTWAPTA}[1]{\cmdkl{L(}#1\cmdkl{)}}

\knowledgenewrobustcmd{\pdlsize}[1]{\cmdkl{\|}#1\cmdkl{\|}}

\knowledgenewrobustcmd{\ksplit}[1][k]{\cmdkl{#1\textit{-split}}}
\knowledgenewrobustcmd{\vconcat}{\mathrel{\cmdkl{{+}\!{+}}}}
\knowledgenewrobustcmd{\vmap}{\mathrel{\cmdkl{\odot}}}

\newenvironment{nestedproof}[1][Proof]{%
  \begin{proof}[#1]%
  \renewcommand{\qedsymbol}{$\lhd$}%
}{%
  \end{proof}%
}

%% file: abstract.tex
We introduce and study $\UCPDLp$, a family of expressive logics rooted in Propositional Dynamic Logic ($\PDL$) with converse (\underline{C}PDL) and universal modality (\underline{U}CPDL).
In terms of expressive power, $\UCPDLp$ strictly contains $\PDL$ extended with intersection and converse (\aka~$\ICPDL$), as well as Conjunctive Queries (CQ), Conjunctive Regular Path Queries (CRPQ), or some known extensions thereof (Regular Queries and CQPDL). 
Further, it is equivalent to the extension of the unary-negation fragment of first-order logic ($\UNFO$) with unary transitive closure, denoted by $\UNTC$, which in turn strictly contains a previously studied extension of $\UNFO$ with regular expressions known as $\UNFOreg$.

We investigate the expressive power, indistinguishability via bisimulations and satisfiability for $\UCPDLp$ and $\CPDLp$.
We argue that natural subclasses of $\CPDLp$ can be defined in terms of the "tree-width" of the underlying graphs of the formulas.
We show that the class of $\CPDLp$ formulas of "tree-width" 2 is equivalent to $\ICPDL$, and that it also coincides with $\CPDLp$ formulas of "tree-width" 1. However, beyond "tree-width" 2, incrementing the "tree-width" strictly increases the expressive power. We characterize the expressive power for every class of fixed "tree-width" formulas in terms of a bisimulation game with pebbles. Based on this characterization, we show that $\CPDLp$ has a tree-like model property.
We prove that the satisfiability problem for $\UCPDLp$ is decidable in "2ExpTime", coinciding with the complexity of $\ICPDL$. As a consequence, the satisfiability problem for $\UNTC$ is shown to be "2ExpTime"-complete as well.
We also exhibit classes for which satisfiability is reduced to "ExpTime".

\color{black}

%% file: knowledge-notice.tex
\noindent
\raisebox{-.4ex}{\HandRight}\ \ This pdf contains internal links: clicking on a "notion@@notice" leads to its \AP ""definition@@notice"".%

%% file: conf-notice.tex
\noindent
\raisebox{-.4ex}{\HandRight}\ \ 
This paper is based on results from the LICS'23 paper \cite{thispaper}, see \Cref{ssec:confdelta} for an overview of the added material.

%% file: intro.tex
\section{Introduction}

Some fundamental formalisms featuring simple recursion, studied in the areas of modal logics and graph databases, revolve around PDL (Propositional Dynamic Logics) and CRPQ (Conjunctive Regular Path Queries), respectively. Both have been extensively studied in their respective fields and are considered fundamental yardsticks for expressive power and complexity, from which many extensions and variants have been derived.

Although they have evolved through different paths motivated by separate applications, they share several common features: 
    (i) the models of `Kripke structures' (for PDL) and `graph databases' (for CRPQ) are essentially the same, "ie", labeled directed graphs; 
    (ii) the recursive features considered in each formalism are also similar, based on building regular expressions over simpler expressions (`RPQ atoms' in the case of CRPQ, and `programs' in the case of PDL).
In this work we study the question: 
\begin{center}
    \it Is it possible to amalgamate these expressive logics on "Kripke structures" and query languages on "graph databases" into one, well-behaved, framework? 
\end{center}
As we shall see, our results suggest that the answer is `yes', and that the resulting logic is equivalent to the unary-negation fragment of first-order logic, extended with transitive closure.
\color{black}

Our point of departure is $\ICPDL$, which is $\PDL$ extended with intersection and converse. This is a well-studied logic, and amongst the most expressive decidable logics on "Kripke structures". On the other hand, we want to be able to test for "conjunctive queries", or more generally "CRPQs", which are the basic building blocks for query languages for "graph databases". 
The outcome is an arguably natural and highly expressive logic, which inherits all good computational and model-theoretical behaviors of $\PDL$ and "CRPQs", which has the flavor of $\PDL$ but allows for richer `conjunctive' tests.
We name this logic $\reintro*\CPDLp$.

\diego{a paragraph on crpq?}

$\PDL$ was originally conceived as a logic for reasoning about programs \cite{DBLP:journals/jcss/FischerL79}. However, variants of $\PDL$ are nowadays used in various areas of computer science, in particular in description logics, epistemic logics, program verification, or for querying datasets (see \cite{DBLP:journals/japll/Lange06,DBLP:conf/csl/GollerL06} for more applications).
One of the most studied extensions of $\PDL$ is the addition of converse navigation, "program intersection" and program complement. 
In particular, adding converse and "program intersection" operators results in a well-behaved logic, known as $\ICPDL$ (`I' for "intersection@program intersection", `C' for converse). $\ICPDL$ has decidable "satisfiability@satisfiability problem", polynomial time "model checking", and enjoys a `tree-like' model property \cite{DBLP:conf/csl/GollerL06}.
A simpler version of "program intersection" studied before is $\CPDL$ extended with "program looping@\loopCPDL", which can state that a program starts and ends at the same point, known as $\loopCPDL$, and it is strictly less expressive than $\ICPDL$.
However, adding the complement of programs results in a logic with a highly undecidable satisfiability problem \cite{harel2001dynamic,DBLP:conf/csl/GollerL06}, although its model checking is still polynomial time \cite{DBLP:journals/japll/Lange06}.

In this work we explore a family of logics which in particular generalizes "program intersection". 
The intuition is that the intersection of two $\PDL$ programs $\pi$ and $\pi'$ could be viewed as the conjunction of two atoms $R_\pi(x,y) \land R_{\pi'}(x,y)$ over the binary relations denoted by the programs $\pi$ and $\pi'$ (\ie, essentially a "conjunctive query"). Pursuing this idea, $\CPDLp$ is the extension of $\CPDL$ allowing expressions which test for any arbitrary number of atoms. For example, a formula can test for an $n$-clique of $\pi$-related elements with $\bigwedge_{1 \leq i < j \leq n} R_\pi(x_i,x_j)$. Further, these tests can be nested, composed, or iterated just like conventional $\PDL$ programs. 
The resulting logic seems appealing from an expressiveness point of view: it captures not only $\ICPDL$ but also several "graph database" query languages studied lately in the quest for finding well-behaved expressive query languages. These include "C2RPQ", "Regular Queries" and $\CQPDL$. 
\color{santichange}
We also investigate the addition of the "universal program" ($\Univ$), allowing to quantify over all worlds of the model, leading to the logic that we called $\reintro*\UCPDLp$. %
\color{black}
While this is a rather standard addition -- easy to handle algorithmically and innocuous in terms of complexity --, it allows us to capture the expressive power of FO logic based languages.
\color{santichange}
See \Cref{fig:expressive-power} for a general idea of where the resulting logics $\CPDLp$ and $\UCPDLp$ sit.
\color{black}
\santi{cambie esto porque hablabamos de CPDL+ y de repenten cambiamabamos a UCPDL+}

Further, subclasses of $\CPDLp$ can be naturally defined by restricting the allowed underlying graphs (\aka Gaifman graph) of these new kind of tests. Thus, $\CPDLg{\+G}$ is the restriction to tests whose underlying graphs are in the class of graphs $\+G$. In particular, for suitable (and very simple) classes we find $\ICPDL$ and $\loopCPDL$.

\begin{figure}
    \centering
        \includegraphics[scale =.65]{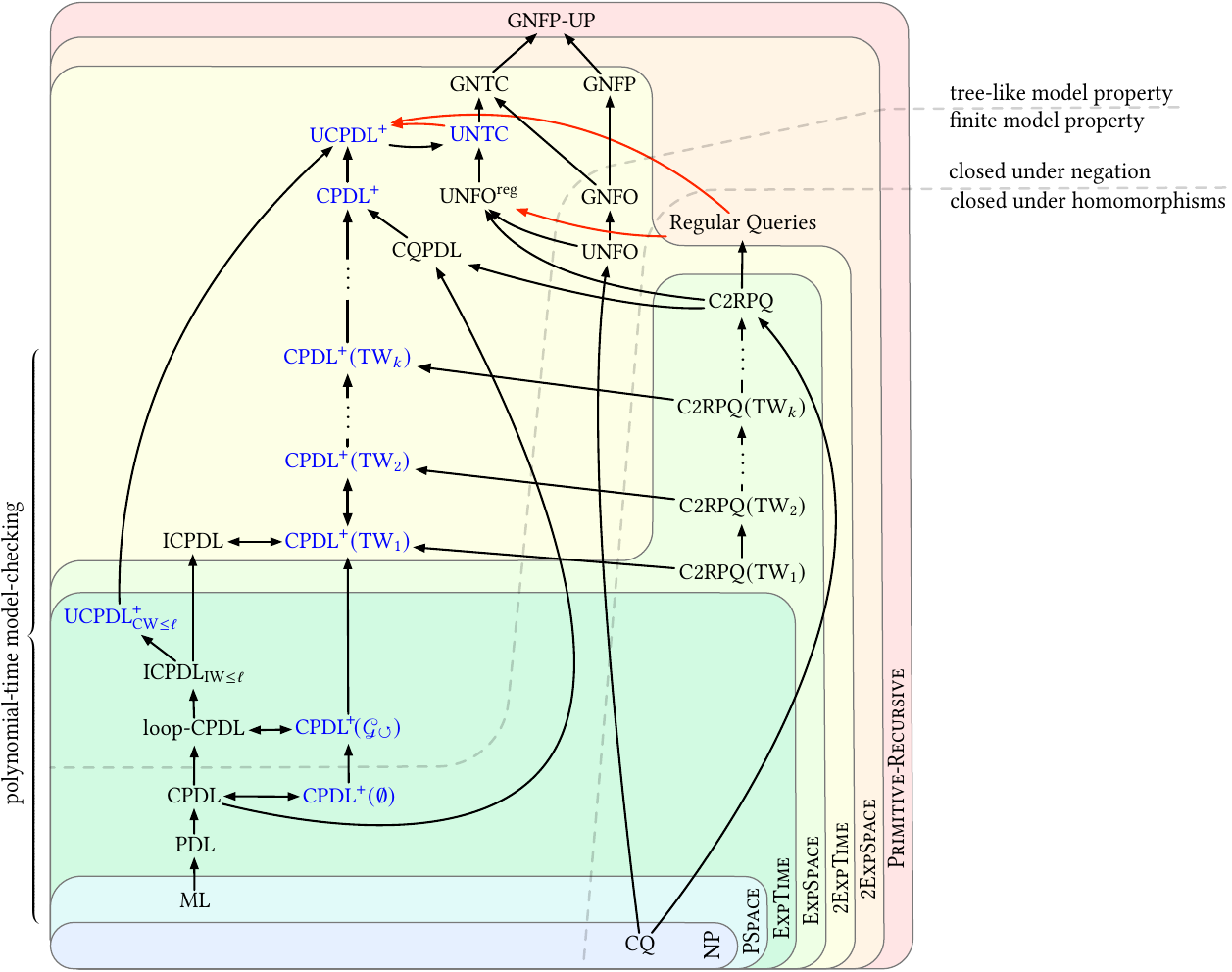}
    
    \caption{{\bf The landscape of expressive power and complexity.} 
    { The arrow goes in the direction of the more expressive language, and all arrows are witnessed via polynomial time translations, except the \color{red}red \color{black} ones, which are exponential. For the query languages "CQ", $\CtwoRPQ$, "Regular Queries", and $\CQPDL$ we restrict to unary queries in order to be able to compare the expressive power. 
    By $\CtwoRPQ(\Tw)$ we denote 
    $\CtwoRPQ$s whose underlying graph has "tree-width"~$\leq k$.
    $\ICPDL_{\Iwidth\leq\ell}$ for any $\ell$, is the set of all $\ICPDL$ "formulas" $\phi$ s.t.\ $\Iwidth(\phi) \leq \ell$, introduced by Göller et al.~\cite{DBLP:journals/jsyml/GollerLL09} and defined in \Cref{sec:solving-omega-reg-sat}.
    $\UCPDLp_{\Cqsizealt \leq \ell}$ for any $\ell$, is the set of all $\UCPDLp$ "formulas" $\phi$ s.t.\ $\Cqsize{\phi} \leq \ell$, as defined in \Cref{sec:sat}. 
    (We assume any fixed $\ell\geq 2$ and $k \geq 2$ for the interpretation of expressiveness arrows.) \color{black} 
    The complexities correspond to the basic reasoning problem for each formalism: satisfiability for logics closed under negation, and containment for query languages of the "CQ"/$\CtwoRPQ$ family. In \color{blue}blue \color{black} we highlight the family of logics introduced in the present work.}
    }
    \label{fig:expressive-power}
\end{figure}
\input{related}

\input{results}

\input{ssec-confdelta}

\input{organization}

%% file: related.tex
\subsection{Related Logics}

As mentioned earlier, $\CPDLp$ can be seen as an ``umbrella logic'' capturing both expressive extensions of $\PDL$, such as $\ICPDL$, and  expressive extensions of "conjunctive queries" studied in the context of "graph databases". We next provide some details on some query languages related to $\ICPDLp$.

\AP
A ""graph database"" is often abstracted as a finite, edge-labeled graph, where labels come from some finite alphabet of relations~(see "eg",~\cite{DBLP:conf/pods/Baeza13} for more details), but via trivial adaptations one can use "graph@graph database" query languages on "Kripke structures", or vice-versa.

\AP
""Conjunctive queries"", or CQs, are conjunctions of relational atoms, closed under projection. For example, the "CQ" 
\color{diegochange}
``$q(x,y) = R(x,y) \land S(y,z) \land T(z,x)$'' returns, on an edge-labeled graph, all pairs of vertices $(u,v)$ such that $u,v$ participate in an $R$-$S$-$T$-cycle.
\color{black}

\AP
In analogy to "conjunctive queries", the class of ""conjunctive regular path queries"", or 
$\intro* \CRPQ$s,
allow for atoms of the form $x \xrightarrow{L} y$, where $L$ is a regular expression over the alphabet of the relations, with the semantics that there exists a path from $x$ to $y$ reading a word of relation names in $L$ \cite{DBLP:conf/kr/CalvaneseGLV00}. 
This language is often considered as the backbone of  any reasonable graph query language \cite{DBLP:conf/pods/Baeza13} and it is embedded in the basic navigational mechanism of the new ISO standard Graph Query Language (GQL) \cite{isoGQL} and the SQL extension for querying graph-structured data SQL/PGQ \cite{isoPGQ}.
The extension of 
\AP
""conjunctive two-way regular path queries"", or $\intro* \CtwoRPQ$, 
allows $L$ to be built over the alphabet augmented with the converse of the relations.

\AP
""Regular queries"" \cite{DBLP:journals/mst/ReutterRV17}, also known as \emph{Nested Positive }2RPQ~\cite{DBLP:conf/dlog/BourhisKR14}, can be seen as the closure under Kleene star, union and converse of $\CtwoRPQ$s. Its containment problem  has been shown to be decidable in elementary time, contrary to some more ambitious generalizations of $\CtwoRPQ$  \cite{DBLP:conf/pods/RudolphK13,DBLP:conf/dlog/BourhisKR14}.

\AP
Perhaps the language closest in spirit to this work is $\intro*\CQPDL$ 
\cite{DBLP:conf/lics/BenediktBB16}, denoting "conjunctive queries" whose atom relations can be defined by $\CPDL$ "programs". It subsumes $\CtwoRPQ$s in expressive power, but it is incomparable with "Regular queries". 

\AP ""Unary negation first-order logic""  ($\intro*\UNFO$) is the fragment of first-order logic over relational structures where negation only occurs in front of formulas with at most one free variable (see \cite{segoufin2013unary} for a formal definition).

\AP
The addition of regular expressions over binary relations to $\UNFO$, called $\intro*\UNFOreg$, was studied by Jung et al.\ \cite{jung2018querying}. It allows for expressions of the form $E(x,y)$, where $E$ is a $\CPDL$-program including tests of the form $\phi(x)?$ 
(see Section~\ref{sec:separation} for a formal definition).

\AP
As we will show, the logic we introduce subsumes all previous formalisms in terms of expressive power. We can further characterize its expressive power in terms of an extension of $\UNFOreg$, corresponding to $\UNFO$ with monadic transitive closure, which we will denote by $\UNTC$.

\color{santichange}
\AP
Beyond the formalisms discussed above, there is also closely related work on path expressions in description logics and guarded fragments. A particularly close development is the recently introduced Regular Guarded Fragment (RGF) of Bednarczyk and Kiero\'nski~\cite{BednarczykKieronski-kr25}, which is also motivated in part by connections with path logics such as $\ICPDL$, though it is incomparable in terms of expressive power with our logic. On the description-logic side, one should also mention the $Z$ family and its tame fragments~\cite{bednarczyk2024database}, as well as automata-theoretic work on positive two-way regular path queries over expressive DL knowledge bases~\cite{calvanese2009regular}. Another related reference is the guarded fixed-point logic with unguarded parameters (\intro*$\GNFPUP$) of Benedikt, Bourhis, and Vanden Boom~\cite{DBLP:conf/lics/BenediktBB16}, which is more expressive than  our proposal and shows that one can add a controlled form of transitive-closure-like expressivity to guarded fixed-point logic while preserving decidability: in particular, $\GNFPUP$ can express the transitive closure of a binary relation and more generally path properties, and it comes with automata-based decidability and boundedness results. These works are clearly close in spirit to the present paper, since they all investigate expressive but decidable formalisms for path navigation and recursion. 
\color{black}

\color{diegochange}
\AP
In a joint work with Nakamura \cite{DBLP:conf/lics/FigueiraFN26} --developed after the first appearance of the present manuscript-- we investigate guarded-negation first-order logic with guarded transitive closure, denoted by $\intro*\GNTC$.
That work gives a polynomial-time reduction from $\GNTC$ satisfiability to $\UNTC$ satisfiability, preserving also finite satisfiability. Combined with \Cref{thm:UNTC-sat}, this already yields the "2ExpTime" upper bound for $\GNTC$ satisfiability. It additionally gives a substantially more direct automata-theoretic proof of the upper bound for $\UNTC$: building on the bounded-treewidth model property established here, it provides a direct exponential-time reduction from $\UNTC$ satisfiability to the non-emptiness problem for 2-way alternating parity tree automata.
Since $\UCPDLp$ admits a polynomial-time translation into $\UNTC$ (\Cref{prop:ECPDL-UNTC}), this also yields an alternative proof of the "2ExpTime" upper bound for $\UCPDLp$ satisfiability (\Cref{thm:sat-cpdlp}-(1)). The argument of \cite{DBLP:conf/lics/FigueiraFN26} does not directly recover the "ExpTime" upper bound established here for classes of $\UCPDLp$ with bounded "conjunctive width" (\Cref{thm:sat-cpdlp}-(2)).
The same work shows that, in combined complexity, the model-checking problems for $\GNTC$ and $\UNTC$ are complete for the class "PNPlog2", matching the model-checking complexity of $\GNFO$. As a consequence, the same bound also holds for $\UNFOreg$.

\color{black}

%% file: results.tex
\subsection{Contributions}

We give next an overview of the main results of this work.

\paragraph{Expressivity}
As already mentioned, $\CPDLp$ contains some previously studied graph database query languages (see \Cref{fig:expressive-power}). Further, depending on the class $\+G$ of graphs, $\CPDLg{\+G}$ can be identified with different previous formalisms. In particular, $\CPDL$ corresponds to $\CPDLg{\emptyset}$, namely $\CPDLp$ on the empty class of graphs.

In \Cref{sec:relation-to-loop-icpdl} we show that, on the other hand, the logic $\loopCPDL$ can also be expressed as $\CPDLg{\Graphloop}$, where $\Graphloop$ is the class with just one "graph" having a self-looping vertex (\Cref{prop:loopCPDL-eq-CPDLgLoop}).
\color{diegochange}
We will also see that $\ICPDL$ and  $\CPDLg{\Graphcap}$ are trivially "equi-expressive", where $\Graphcap$ is the class having just one "graph"  containing two vertices and one "edge" between them.
\color{black}
What is more, we also show that $\ICPDL$ is "equi-expressive" to $\CPDLg{\Tw[2]}$. Hence, "formulas" and "programs" of $\CPDLg{\Tw[1]}$, $\CPDLg{\Tw[2]}$, $\CPDLg{\Graphcap}$, and $\ICPDL$ define the same unary and binary relations, respectively (\Cref{thm:ICPDL_equals_TW1_equals_TW2}).

However, as we show in \Cref{sec:separation}, beyond $\Tw[2]$ the expressive power starts increasing in an infinite hierarchy, where $\CPDLg{\Tw[k]}$ is strictly contained in $\CPDLg{\Tw[k+1]}$ in terms of expressive power, for every $k \geq 2$ (\Cref{thm:separation}).

On the other hand, we show that $\UCPDLp$ has the same expressive power as the unary-negation fragment of first-order logic extended with transitive closure, which we denote by $\UNTC$ (\Cref{cor:UCPDLp-equiv-UNTC}). However, the type of transitive closure is very restrictive: it can be applied only on binary formulas with no parameters.

\paragraph{Indistinguishability}
We characterize model-indis\-tin\-gui\-shability for $\CPDLg{\Tw[k]}$ in terms of a $k$-pebble local consistency game. That is, we study the conditions under which a pair of models cannot be distinguished by a "formula" or "program" from $\CPDLg{\Tw[k]}$.
In Section~\ref{sec:caracterization} we characterize both the indistinguishability via "positive" "formulas" (\Cref{thm:sim-char}) or arbitrary "formulas"  (\Cref{them:bisim-charac}).
These notions correspond, in some sense, to the lifting of the simulation relation from Modal Logics to $\CPDLg{\Tw[k]}$. 
We show that for $k\geq 2$, $\CPDLg{\Tw}$ enjoys the "$\Tw$-model property": if a "formula" of $\CPDLg{\Tw}$ is satisfiable, it is satisfied in a "model@Kripke structure" of "tree-width" $k$ (\Cref{cor:treewidth-k-model-property}).

\paragraph{Satisfiability}
In \Cref{sec:sat} we study the "satisfiability problem", that is, whether for a given formula there exists a "model@Kripke structure" 
which satisfies it. 
We show that the satisfiability for $\UCPDLp$ is "2ExpTime"-complete (\Cref{thm:sat-cpdlp}), \ie, the same complexity as $\ICPDL$ \cite{DBLP:journals/jsyml/GollerLL09}. 
We also identify a hierarchy, based on what we call `"conjunctive width"' instead of "tree-width", so that any class of bounded "conjunctive width" is decidable in "ExpTime", \ie, the complexity of $\PDL$ or $\loopCPDL$.
These results exploit the "bounded tree-width model property@$\Tw$-model property" of $\CPDLp$ above.

We show that the "satisfiability problem" for $\UNTC$ is "2ExpTime"-complete (\Cref{thm:UNTC-sat}), via a reduction to $\UCPDLp$-"satisfiability@satisfiability problem".
This result is obtained via a translation to $\UCPDLp$ which, despite being exponential, produces formulas of polynomial "conjunctive width" (\Cref{prop:UNTC-ECPDL}).
\color{black}

%% file: ssec-confdelta.tex
\subsection{Conference Version of this Work}
\label{ssec:confdelta}
Some of the results of this work were presented in LICS'23 \cite{thispaper}. This is the extended version with detailed proofs as well as:
\begin{enumerate}
    \item the extension of the logic with "universal programs" {\color{santichange} and $n$-ary relations};
    \item {\color{santichange} the extension of the pebble games and notions of bisimulation for dealing with such "universal programs"};
    \item the proof for the "satisfiability problem" has been extended (to account for the new features above) and  simplified by minimizing the repetition of developments in prior work;
    \item the upper bound for satisfiability has been improved from "3ExpTime" to "2ExpTime"(-complete);
    \item the inclusion of $\UNTC$ in our study, showing the equivalence of $\UCPDLp$ with $\UNTC$ in terms of expressive power and a tight bound for $\UNTC$ satisfiability;
    \item several inexpressibility results "wrt" related logics such as $\GNFO$ and $\GNFP$;
    \color{diegochange}
    \item we include further comparisons with relevant prior work as well as with subsequent developments, most notably our joint work with Nakamura on $\GNTC$ \cite{DBLP:conf/lics/FigueiraFN26}.
    \color{black}
\end{enumerate}
\color{black}

%% file: organization.tex
\subsection{Organization}

In \Cref{sec:prelim} we fix the basic notation used throughout the manuscript.
In \Cref{sec:pdl-defn} we define the family of $\PDL$ logics and we introduce the logic of $\UCPDLp$ and its semantics.
\Cref{sec:relation-to-loop-icpdl} shows the equivalence, in terms of expressive power, of $\CPDLg{\Tw[1]}$, $\CPDLg{\Tw[2]}$ and $\ICPDL$ on the one hand, and of $\loopCPDL$ with $\CPDLg{\Graphloop}$ on the other.
In \Cref{sec:caracterization} we define (bi)simulation relations for the fragments $\CPDLg{\Tw}$ of $\UCPDLp$ and show that they characterize model-indistinguishability.
\Cref{sec:separation} shows that $\CPDLg{\Tw[k+1]}$ is strictly more expressive than $\CPDLg{\Tw[k]}$ for every $k \geq 2$, as well as the incomparability of $\CPDLp$ with respect to $\UNFOreg$ and $\UNFO$.
In \Cref{sec:sat} we solve the satisfiability problem for $\UCPDLp$.
In \Cref{sec:UNTC} we introduce the extension of $\UNFO$ with monadic transitive closure and we show that it is equivalent to $\UCPDLp$ in terms of expressive power via computable translations, and we study its satisfiability problem.
Finally, in \Cref{sec:conclusions} we give some concluding remarks and also state a  model-checking result.

%% file: prelim.tex
\section{Preliminaries}\label{sec:prelim}
We use the standard notation $\bar u$ to denote a tuple of elements from some set $X$ and $\bar u[i]$ to denote the element in its $i$-th component. 
\AP
If $\bar q=(q_1, \dotsc, q_k)$ is a tuple and $1\leq i\leq k$, by $\bar q[i\intro*\mapcoord r]$ we denote the tuple $(q_1, \dotsc, q_{i-1}, r, q_{i+1}, \dotsc, q_k)$, and by $\intro*\dimtup(\bar q)$ we denote its ""dimension"" $k$.
For relations $R,S\subseteq X\times X$ we define the ""composition"" $R\circ S:=\set{(u,v)\mid \exists w. uRw, wSv}$.
For a function $f\colon X \to Y$, we often write $f(\bar u)$ to denote $(f(\bar u[1]), \dotsc, f(\bar u[k]))$, where $k$ is the dimension of $\bar u$. 
\AP
By $\intro*\pset{X}$ we denote the set of all subsets of $X$.

\AP
A (simple, undirected) ""graph"" is a tuple $G=(V,E)$ where $V$ is a set of ""vertices"" and $E$ is a set of ""edges"", where each "edge" is a set of vertices of size at most two. We refer to $V$ and $E$ by $V(G)$ and $E(G)$, respectively. 
We henceforth assume that "graphs" have a finite set of "vertices" unless otherwise stated. 
\AP
A ""tree"" is a connected "graph" with no cycles (in particular no singleton edges).

\AP
A ""tree decomposition"" of a "graph" $G$ is a
pair $(T, \intro*\bagmap)$ where $T$ is a "tree" and $\bagmap\colon V(T) \to \pset{V(G)}$ is a function that associates to each node of $T$, called ""bag"",
a set of vertices of $G$. When $x \in \bagmap(b)$ we shall say that the bag $b \in V(T)$
""contains@@tw"" the vertex $x$. 
\AP
Further, it must satisfy the following three properties:
    ""A@@treedec"". each vertex $x$ of $G$ is "contained@@tw" in at least one "bag" of $T$;
    ""B@@treedec"". if $G$ has an edge between $x$ and $y$, there is at least one "bag" of $T$
        that "contains@@tw" both $x$ and $y$; and 
    ""C@@treedec"".
    for each vertex $x$ of $G$, the set of bags of $T$ "containing@@tw" $x$ is a 
        connected subtree of $V(T)$.

The "tree decomposition" of an infinite "graph" is an infinite tree and a mapping satisfying the conditions "A@@treedec", "B@@treedec", "C@@treedec" above.
\AP
The ""width"" of $(T, \bagmap)$ is the maximum (or supremum if $T$ is infinite) of $|\bagmap(b)|-1$ when $b$ ranges over
$V(T)$. The ""tree-width"" of $G$, denoted $\intro*\tw(G)$, is the minimum of the "width" of all "tree decompositions" of $G$.
\AP 
We denote by $\intro*\Tw$ the set of all (finite) "graphs" of "tree-width" at most $k$. See Figure \ref{fig:example_ugraph} (B and C) for "graphs" of different "tree-width" and Figure \ref{fig:example_treedecomp} for an example of a "tree decomposition". 
\AP A "tree decomposition" $(T, \bagmap)$ is ""good"" if $T$ is binary and for every pair of nodes $u,v$ of $T$ such that $v$ is a child of $u$ we have $\bagmap(u) \subseteq \bagmap(v)$ or $\bagmap(v) \subseteq \bagmap(u)$.

\AP
Let $\intro*\Rels$ and $\intro*\Vars$ be countably infinite, pairwise disjoint, sets of ""relation names"" and ""variables"", respectively, and let $\intro*\arity\colon \sigma \to \Nat \setminus \set 0$ be a function assigning an ""arity"" to each "relation name".
We will denote by $\intro*\Prop \eqdef \set{R \in \Rels : \arity(R)=1}$ the set of unary relations, which we call ""atomic propositions""; 
and by $\intro*\Prog \eqdef \set{R \in \Rels : \arity(R)=2}$ the set of binary relations, which we call ""atomic programs"".
\AP
A ""$\Rels$-structure"" (or simply a "structure") $K$ is a pair $(X,\iota)$ where $X$ is the ""domain"" and $\iota$ is a function so that for each relation $R \in \Rels$ of "arity" $n$ we have that $\iota(R) \subseteq X^n$. We will henceforth write $\intro*\dom K$ and $R^K$ to denote $X$ and $\iota(R)$, respectively. 
\AP
We call each element of $X$ a ""domain element"" or ""world"", borrowing the jargon from modal logics.
\AP
A ""Kripke structure"" is a "structure" $K$ such that $R^K=\emptyset$ for every $R \in \Rels$ with $\arity(R)>2$. We denote by $\intro*\Kripke$ the class of all "Kripke structures".

\AP
The ""Gaifman graph"" of a "structure" $K$ is the "graph" $(V,E)$ where $V = \dom K$ and $E$ contains a pair $\set{u,v}$ if there is some $R \in \Rels$ and some tuple $\bar w \in R^K$ containing both $u$ and $v$.
\AP
The ""distance"" between two "domain elements" of $K$ is simply the length of the shortest path in the "Gaifman graph" of $K$.
\AP
We say that $K$ is of ""finite degree"" if every "world" of $K$ has a finite number of `neighbors' at "distance" $1$.
\AP
A ""homomorphism"" from a "structure" $K$ to another "structure" $K'$ is a function $f \colon \dom K \to \dom{K'}$ such that for every "relation" $R \in \Rels$ of "arity" $n$ and 
$\bar w \in (\dom K)^n$ we have that $\bar w \in R^K$ implies $f(\bar w) \in R^{K'}$.\footnote{We use the standard notation where $f((u_1,\dotsc, u_n)) \eqdef (f(u_1),\dotsc, f(u_n))$.}
\AP
A function $f \colon \hat X \to \dom{K'}$ with $\hat X \subseteq \dom K$ is a ""partial homomorphism"" if it is a "homomorphism" from $\hat K$ to $K'$, where $\hat K$ is the substructure of $K$ induced by $\hat X$.
\AP
Let $\intro*\kHoms(K,K')$ be the set of all $(\bar u, \bar v) \in \dom{K}^k \times \dom{K'}^k$ such that $\set{\bar u[i] \mapsto \bar v[i] \mid 1 \leq i \leq k}$ is a "partial homomorphism" from $K$ to $K'$. 
When $\bar u'$ and $\bar v'$ are tuples of dimension $k' < k$, we will often abuse notation and write $(\bar u', \bar v') \in \kHoms(K,K')$ to denote $(\bar u', \bar v') \in \kHoms[k']$.

\AP
The ""tree-width@@structure"" of a "structure" is the "tree-width" of its "Gaifman graph".

%% file: pdl-intro.tex
\section{PDL and its Extensions}
\label{sec:pdl-defn}

\subsection{Known Extensions of \texorpdfstring{$\PDL$}{PDL}: \texorpdfstring{$\ICPDL$}{ICPDL} and \texorpdfstring{$\loopCPDL$}{loop-CPDL}}
We first define $\AP\intro*\CPDL$, "ie", $\PDL$ with converse.
\AP""Expressions"" of $\CPDL$ can be either
""formulas"" $\phi$ or ""programs"" $\pi$, defined by the following grammar, where $p$ ranges over $\Prop$ and $a$ over $\Prog$:
\begin{align*}
    \phi &\eqqdef p \mid \lnot \phi \mid \phi\land\phi \mid \tup{\pi} \\
    \pi &\eqqdef \epsilon \mid a \mid \bar a \mid \pi \cup \pi \mid \pi \circ \pi \mid \pi^* \mid \phi? 
\end{align*}
where $\phi?$ "programs" are often called \AP""test programs"".
\AP
We define the semantics of "programs" $\intro*\dbracket{\pi}_K$ and of "formulas" $\dbracket{\phi}_K$ in a "Kripke structure" $K=(X,\iota)$, where  $\dbracket{\pi}_K\subseteq X\times X$ and $\dbracket{\phi}_K\subseteq X$, as follows:
\begin{align*}
    \dbracket{p}_K \eqdef{}& p^K &\text{ for $p \in \Prop$}, \\
    \dbracket{\lnot\phi}_K \eqdef{}&  X\setminus \dbracket{\phi}_K,\\
    \dbracket{\phi_1\land\phi_2}_K \eqdef{}&  \dbracket{\phi_1}_K\cap\dbracket{\phi_2}_K,\\
    \dbracket{\tup{\pi}}_K \eqdef{}&  \{u\in X\mid \exists v\in X.(u,v)\in\dbracket{\pi}_K\},\\
    \dbracket{\epsilon}_K \eqdef{}&  \{(u,u) \mid u\in X\}, \\
    \dbracket{a}_K \eqdef{}& a^K  \text{ for $a \in \Prog$,}\\
    \dbracket{\bar a}_K \eqdef{}&  \{(v,u)\in X^2\mid (u,v) \in a^K\} &\text{ for $a \in \Prog,$}\\
    \dbracket{\pi_1\star\pi_2}_K \eqdef{}&  \dbracket{\pi_1}_K\star\dbracket{\pi_2}_K &\text{ for $\star \in \set{\cup,\circ}$},\\
    \dbracket{\pi^*}_K \eqdef{}&  \mbox{the reflexive transitive closure of $\dbracket{\pi}_K$},\\
    \dbracket{\phi?}_K \eqdef{}&  \{(u,u) \mid u\in X,u\in\dbracket{\phi}_K\}.
\end{align*}
We write $K,u\models\phi$ for $u\in\dbracket{\phi}_K$ and $K,u,v\models\pi$ for $(u,v)\in\dbracket{\pi}_K$.
\AP
We write $\phi_1\intro*\semequiv\phi_2$ \resp{$\pi_1\reintro*\semequiv\pi_2$} 
if $\dbracket{\phi_1}_K=\dbracket{\phi_2}_K$ \resp{$\dbracket{\pi_1}_K=\dbracket{\pi_2}_K$} 
\AP
for every "structure" $K$ (or, equivalently, every "Kripke structure" $K$), in which case we say that $\phi_1,\phi_2$ \resp{$\pi_1,\pi_2$} are ""equivalent"".
\AP
$\intro*\PDL$ is the fragment of $\CPDL$ "expressions" which do not use "atomic programs" of the form $\bar a$.

$\AP\intro*\ICPDL$ is defined as the extension of $\CPDL$ with ""program intersection""
$\pi \eqqdef \pi \cap \pi$
with semantics defined by
$    \dbracket{\pi_1\cap\pi_2}_K \eqdef \dbracket{\pi_1}_K\cap\dbracket{\pi_2}_K.
$
\knowledgenewrobustcmd{\reverseof}[1]{#1^{\cmdkl{\textit{rev}}}}
Observe that $\ICPDL$ "programs" are closed under converse, that is, for every "program" $\pi$ there is a `reverse' program $\reverseof{\pi}$ such that $\dbracket{\pi}_K$ contains $(u,v)$ if{f} $\dbracket{\reverseof{\pi}}_K$ contains $(v,u)$, where $\AP\intro*\reverseof{\pi}$ is defined as follows:
\begin{samepage}
\begingroup\allowdisplaybreaks[0]
\begin{align*}
\reverseof{\star}&\eqdef{}\star&\mbox{for $\star\in\{\epsilon,\phi?\}$,}\\
\reverseof{a}& \eqdef{} \overline a &\mbox{for $a\in\A$,}\\
\reverseof{\overline a}& \eqdef{}a  &\mbox{for $a\in\A$,}\\
\reverseof{(\pi_1\star\pi_2)}&\eqdef{}\reverseof{\pi_1}\star \reverseof{\pi_2}&\mbox{for $\star\in\{\cap,\cup\}$,}\\
\reverseof{(\pi_1 \circ \pi_2)}&\eqdef{}\reverseof{\pi_2} \circ \reverseof{\pi_1},\\
\reverseof{(\pi^*)}&\eqdef{}(\reverseof{\pi})^*.
\end{align*}
\endgroup
\end{samepage}
\begin{remark}Notice that 
    $\pi$ is an $\ICPDL$-"program" iff $\reverseof{\pi}$ is an $\ICPDL$-"program", and $(u,v)\in\dbracket{\pi}_K$ iff $(v,u)\in\dbracket{\reverseof{\pi}}_K$.
\end{remark}

\noindent
Finally, $\AP\intro*\loopCPDL$ is defined as the extension of $\CPDL$ resulting from adding
$\AP\intro*\lo(\pi)$
\begin{align*}
    \phi & \eqqdef  \lo(\pi)
\end{align*}
to its grammar, with semantics defined by
    $\dbracket{\lo(\pi)}_K  \eqdef \{x \mid (x,x)\in \dbracket{\pi}_K \}.$

\subsection{A New Logic: \texorpdfstring{$\UCPDLp$}{UCPDL⁺}}
We will now define a new kind of "program" which enables testing for conjunction of first-order atoms on arbitrary "$\Rels$-structures". Such atoms may in turn be "programs".
\AP
Concretely, we define an `""atom""' as either (a) an expression of the form $\pi(x,x')$, where $\pi$ is a {\CPDLp} "program" and $x,x' \in \Vars$, or (b) an expression $R(x_1, \dotsc, x_n)$, for $R\in\Rels$, where $\arity(R)=n$, $n>2$, and $x_1, \dotsc, x_n \in \Vars$.
We will denote by ""p-atom"" (for `\underline{p}rogram') and ""r-atom"" (for `\underline{r}elation') the "atoms" of the first (a) and second (b) kind, respectively.
\color{black}
\AP
For an "atom" $\tau(x_1, \dotsc, x_n)$ we define $\intro*\vars(\tau(x_1, \dotsc, x_n)) \eqdef \set{x_1, \dotsc, x_n}$, and for a set of "atoms" $C$ we define $\reintro*\vars(C) \eqdef \bigcup_{A \in C} \vars(A)$.

We define $\intro*\CPDLp$ as an extension of $\CPDL$ allowing also "programs" of the form 
    \[\pi \eqqdef C[x_s,x_t]\]
where:
\AP
(1) $C$ is a finite set of "atoms";  
(2) $x_s,x_t\in\vars(C)$;\footnote{Note that $x_s$ and $x_t$ may be equal or distinct "variables".}
and
(3) the ""underlying graph@@C"" 
$\intro*\uGraphC{C}$ of $C$ is connected, where $\uGraphC{C}$ is defined as $V(\uGraphC{C})=\vars(C)$, and $E(\uGraphC{C}) = \set{\set{u,v} \mid \set{u,v} \subseteq \vars(A), A \in C}$.\footnote{The fact that we restrict our attention to ``connected'' "conjunctive programs" is unessential for any of our results, but it will allow us to easily relate to previously defined logics, such as $\ICPDL$ and $\loopCPDL$.}
\AP
We call these "programs" ""conjunctive programs"" (as they generalize "conjunctive queries").
Observe that $\set{x_s,x_t}\subseteq{\vars(C)}$, and hence  we also define $\reintro*\vars(C[x_s,x_t]) \eqdef {\vars(C)}$.

\AP
Fix a "$\Rels$-structure" $K$. A function $f \colon \vars(C) \to \dom{K}$ is a ""$C$-satisfying assignment"" if $(f(x),f(x')) \in \dbracket{\pi'}_K$ for every "p-atom" $\pi'(x,x')\in C$, and 
$f(\bar x) \in R^K$ for every "r-atom" $R(\bar x)\in C$.
We define the semantics of "conjunctive programs" on arbitrary "$\Rels$-structures".
The semantics $\dbracket{C[x_s,x_t]}_K$ of a "conjunctive program" on a "$\Rels$-structure" $K$ is the set of all pairs $(w_s,w_t) \in \dom{K} \times \dom{K}$ such that $f(x_s) = w_s$ and $f(x_t) = w_t$ for some "$C$-satisfying assignment" $f$.

\AP
For a class $\+G$ of "graphs" we define $\intro*\CPDLg{\+G}$ as the fragment of $\CPDLp$ whose "programs" have one of the shapes allowed in $\+G$. 
Formally, for any "conjunctive program" $\pi = C[x_s,x_t]$ we consider the ""underlying graph"" $\intro*\uGraph{C[x_s,x_t]}$ of $C[x_s,x_t]$ as having $V(\uGraph{C[x_s,x_t]}) = V(\uGraphC{C})$ and $E(\uGraph{C[x_s,x_t]}) = \set{\set{x_s,x_t}} \cup E(\uGraphC{C})$. 
Observe that $x_s$ and $x_t$ are always connected via an "edge" in $\uGraph{C[x_s,x_t]}$ but not necessarily in $\uGraphC{C}$. $\CPDLg{\+G}$ is the fragment of $\CPDLp$ whose only allowed "conjunctive programs" are of the form $C[x_s,x_t]$ where $\uGraph{C[x_s,x_t]}\in\+G$. See Figure~\ref{fig:example_ugraph} for an example of "underlying graphs".

\begin{figure}
\ifarxiv
    \includegraphics[scale=0.25]{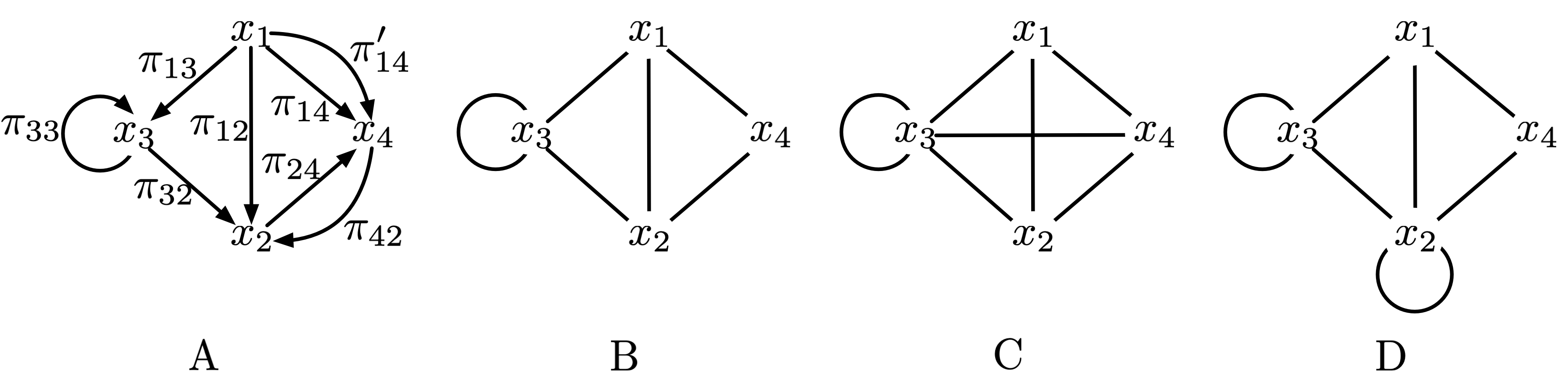}
\else
    \includegraphics[width=.47\textwidth]{img/ugraphs.png}
\fi
\caption{A. graphical representation of the "atoms" contained in the set 
$C=\{
\pi_{13}(x_1,x_3),$ $
\pi_{12}(x_1,x_2),$ $
\pi_{14}(x_1,x_4),$ $
\pi'_{14}(x_1,x_4),$ $
\pi_{33}(x_3,x_3),$ $
\pi_{32}(x_3,x_2),$ $
\pi_{24}(x_2,x_4),$ $
\pi_{42}(x_4,x_2)
\}$; B. $\uGraph{C[x_1,x_2]}$; C. $\uGraph{C[x_3,x_4]}$; D. $\uGraph{C[x_2,x_2]}$. Observe that $\uGraph{C[x_1,x_2]},\uGraph{C[x_2,x_2]}\in\Tw[2]$ but $\uGraph{C[x_3,x_4]}\in\Tw[3]\setminus\Tw[2]$.
}
\label{fig:example_ugraph}
\end{figure}

\AP
We will consider $\CPDLp$ extended with a ""universal program"" $\intro* \Univ$ whose semantics over arbitrary "$\Rels$-structures" $K$ is $\dbracket{\Univ}_K = \dom{K}\times\dom{K}$. Unlike $\CPDLp$ -- which, like $\PDL$, only allows movements within a connected component of a "structure" -- the operator $\Univ$ allows us to ``jump'' between components: $\tup{\Univ\circ\phi?}$ is true at a "world" of $K$ if and only if there is some "world" in $K$ satisfying $\phi$. We denote by $\intro*\UCPDLp$ and $\intro*\UCPDLg{\+G}$ the previously defined logics extended with $\Univ$. Observe that condition (3) on the "underlying graph@@C" $\uGraphC{C}$ of $C$ being connected is trivial in $\UCPDLp$, since if needed one can add "p-atoms" $\Univ(x,x')$ to $C$ to get the connectedness condition.

\AP
We say that an "expression" of $\UCPDLp$ is ""positive"" if it does not contain any subformula of the form $\lnot \psi$.

\color{santichange}
When restricting our attention to "Kripke structures", we shall always assume that "conjunctive programs" do not contain "r-atoms". Indeed, by definition, "Kripke structures" only interpret relation symbols of arity at most $2$, and hence there are no relations of arity greater than $2$ available to form such atoms. Thus, over "Kripke structures", we can assume without any loss of generality that "conjunctive programs" are built only from "p-atoms".
\color{black}

\AP
Sometimes we will need to use $\CPDLp$ or $\UCPDLp$ extended with "program intersection" ($\cap$), that we denote by $\intro*\ICPDLp$ and $\intro*\IUCPDLp$. The semantics of $\pi_1\cap\pi_2$ is that of $\ICPDL$. Observe that any expression of $\ICPDL$ is also an expression of $\ICPDLp$ (with no "conjunctive programs") with equal semantics; hence $\ICPDLp$ extends $\ICPDL$.
\AP
Analogously, $\intro*\ICPDLg{\+G}$ is defined as $\CPDLg{\+G}$ extended with "program intersection". 
\begin{remark}\label{rk:intersecion_does_not_add_expressive_power}
 $\ICPDLp$ and $\CPDLp$ \resp{$\IUCPDLp$ and $\UCPDLp$} are equivalent in terms of expressive power. Furthermore, for any class $\+G$ containing the graph which consists of an edge between two distinct nodes we have that $\ICPDLg{\+G}$ and $\CPDLg{\+G}$ are also equivalent in terms of expressive power.
\end{remark}
\begin{proof}
$\pi_1\cap\pi_2$ is "equivalent" to $\set{\pi_1(x,y),\pi_2(x,y)}[x,y]$.
\end{proof}

\noindent
\AP
For an $\ICPDLp$ "expression" $e$, let $\reintro*\subexpr(e)$ be the set of all ""subexpressions"" of $e$, defined in the usual way. That is, $\AP\intro*\subexpr(e)$ is the smallest set satisfying the following:
\begin{itemize}
    \item $e \in \subexpr(e)$,
    \item if $\psi \land \psi' \in \subexpr(e)$ then $\set{\psi, \psi'} \subseteq \subexpr(e)$
    \item if $\lnot\psi \in \subexpr(e)$ then $\psi \in \subexpr(e)$
    \item if $\tup{\pi} \in \subexpr(e)$, then $\pi \in \subexpr(e)$
    \item if $\pi \star \pi' \in \subexpr(e)$, then $\set{\pi,\pi'} \subseteq \subexpr(e)$ for every $\star \in \set{\circ,\cup,\cap}$,
    \item if $\pi^* \in \subexpr(e)$, then $\pi \in \subexpr(e)$,
    \item if $\psi? \in \subexpr(e)$, then $\psi \in \subexpr(e)$,
    \item if $C[x_s,x_t] \in \subexpr(e)$ and $\pi(x,y) \in C$, then $\pi \in \subexpr(e)$.
\end{itemize}
We also define the \AP""size@@pdl"" $\intro*\pdlsize{e}$ of an $\ICPDLp$ "expression" $e$:
\begin{align*}
    &\pdlsize{e} = \pdlsize{\bar a} = \pdlsize{\epsilon} \eqdef 1
        &&\text{for $e \in \Prog \cup \Prop$, $a \in \Prog$,}
    \\
    &\pdlsize{e^*} = \pdlsize{e?} = \pdlsize{\tup{e}} = \pdlsize{\lnot e} \eqdef 1 + \pdlsize{e},
    \\
    &\pdlsize{e_1 \star e_2} \eqdef \pdlsize{e_1} + \pdlsize{e_2}
        &&\text{for $\star \in \set{\cup,\cap, \circ, \land}$, and}
    \\
    &\pdlsize{C[x_s,x_t]} \eqdef
      \begin{aligned}
        &\sum \set{\arity(R) : R(\bar x) \text{ "r-atom" of } C} + {}\\
        &\sum \set{1{+}\pdlsize{\pi} : \pi(x,y) \text{ "p-atom" of } C }
      \end{aligned}
        &&\text{for a "conjunctive program" $C[x_s,x_t]$.}
\end{align*}

%% file: translation.tex
\newcommand{\twtcpdl}{\CPDLg{\Tw[2]}}
\newcommand{\twocpdl}{\CPDLg{\Tw[1]}}
\newcommand{\twncpdl}{\CPDLg{\Tw[n]}}

\newcommand{\tr}{{\rm Tr}\xspace}

\section{Relation to \texorpdfstring{$\loopCPDL$}{loop-CPDL} and \texorpdfstring{$\ICPDL$}{ICPDL}}
\label{sec:relation-to-loop-icpdl}
For logics $\+L_1$ and $\+L_2$ with expressions consisting of "formulas" and "programs", we say that $\+L_2$ is ""at least as expressive as"" $\+L_1$, denoted $\+L_1{\AP\intro*\lleq}\+L_2$, 
\phantomintro{\notlleq}%
if there is a translation 
mapping $\+L_1$-"formulas" to $\+L_2$-"formulas" and $\+L_1$-"programs" to $\+L_2$-"programs" which preserves "equivalence".
\AP
We write $\+L_1 \AP\intro*\lleqs \+L_2$ to denote that $\+L_1 \lleq \+L_2$ and $\+L_2 \not\lleq \+L_1$. We write $\+L_1 \intro*\langsemequiv \+L_2$ to denote that $\+L_1 \lleq \+L_2$ and $\+L_2 \lleq \+L_1$, in which case we say that $\+L_1$ and $\+L_2$ are ""equi-expressive"".
We will also restrict these notions over the subclass $\Kripke$ of "Kripke structures", in which case we write $\reintro*\+L_1 \lleq[\Kripke] \+L_2$, $\reintro*\+L_1 \lleqs[\Kripke] \+L_2$ and $\+L_1 \reintro*\langsemequiv[\Kripke] \+L_2$.

We will also compare the expressive power of logics of different nature: on the one hand, a logic $\+L_1$ being an extension of first-order logic (in particular, possibly with free variables), and on the other hand a logic $\+L_2$ with "formulas" and "programs". In this case we say that $\+L_2$ is \reintro{at least as expressive as} $\+L_1$, denoted $\+L_1\mathrel{\AP\lleq}\+L_2$, if there is 
\begin{enumerate}
    \item a translation $T$ mapping $\+L_1$-formulas $\phi(x)$ with free variable $x$ to an $\+L_2$-"formula" $T(\phi(x))$ preserving "equivalence" in the sense that for any $\Rels$-structure and $u\in\dom{K}$, we have $K\models\phi(x)[x\mapsto u]$ iff $K,u\models T(\varphi(x))$ and 
    \item a translation $T$ mapping $\+L_1$-formulas $\phi(x,y)$ with free variables $x$ and $y$ ($x\neq y$) to an $\+L_2$-"program" $T(\phi(x,y))$ preserving "equivalence" in the sense that for any $\Rels$-structure and $u,v\in\dom{K}$, we have $K\models\phi(x,y)[x\mapsto u,y\mapsto v]$ iff $K,u,v\models T(\varphi(x,y))$.
\end{enumerate}
 In the same way, we say that $\+L_1$ is \reintro{at least as expressive as} $\+L_2$, denoted $\+L_2 \mathrel{\AP\lleq}\+L_1$, if there is 
 \begin{enumerate}
    \item a translation $T$ mapping $\+L_2$-formulas $\phi$ to $\+L_1$-"formulas" $T(\phi)$ with a free variable $x$  preserving "equivalence" in the sense that for any $\Rels$-structure and $u\in\dom{K}$, we have $K,u\models\phi$ iff $K\models T(\varphi)[x\mapsto u]$ and 
    \item a translation $T$ mapping $\+L_2$-programs $\pi$ to $\+L_1$-"formulas" $T(\pi)$ with free variables $x$ and $y$ preserving "equivalence" in the sense that for any $\Rels$-structure and $u,v\in\dom{K}$, we have $K,u,v\models\pi$ iff $K\models T(\pi)[x\mapsto u,y\mapsto v]$.
 \end{enumerate}
In this setting we also use the notation $\+L_1 \lleq[\Kripke] \+L_2$ and $\+L_2 \lleq[\Kripke] \+L_1$ with the obvious meaning.

\AP
Let $\intro*\Graphloop \eqdef \set{G}$ be the class having just one "graph" $G=(\set{v},\set{\set{v}})$ consisting of just one self-looping "edge", and
\AP
$\intro*\Graphcap \eqdef \set{G}$ be the class having just one "graph" $G = (\set{v,v'}, \set{\set{v,v'}})$  containing one "edge". Observe that $\Graphloop$ and $\Graphcap$ are contained in $\Tw[1]$.

We next show that $\loopCPDL\langsemequiv[\Kripke]\CPDLg{\Graphloop}$, and $\ICPDL\langsemequiv[\Kripke]\CPDLg{\Graphcap}\langsemequiv[\Kripke]\CPDLg{\Tw[1]}\langsemequiv[\Kripke]\CPDLg{\Tw[2]}$ via polynomial time translations.

\subsection{\texorpdfstring{$\loopCPDL$}{loop-CPDL} and \texorpdfstring{$\CPDLg{\Graphloop}$}{CPDL⁺(𝒢↺)} are Equi-expressive}

\begin{proposition}\AP\label{prop:loopCPDL-eq-CPDLgLoop}
    $\loopCPDL\langsemequiv[\Kripke]\CPDLg{\Graphloop}$ via polynomial time translations.
\end{proposition}

\paragraph{Translation from $\loopCPDL$ to $\CPDLg{\Graphloop}$}\label{sec:from_loopCPDL}

\knowledgenewrobustcmd{\trOne}{\cmdkl{\tr_1}}
We define the following translation $\AP\intro*\trOne$ from $\loopCPDL$-"formulas" to $\CPDLg{\Graphloop}$-"formulas" and from 
$\loopCPDL$-"programs" to $\CPDLg{\Graphloop}$-"programs":
\begin{align}
\trOne(\star)&\eqdef \star&&\text{for $\star\in\Prop\cup\Prog\cup\{\epsilon\}$,}\nonumber\\
\trOne(\bar a)&\eqdef \bar a&&\text{for $a\in\Prog$,}\nonumber\\
\trOne(\varphi\land\psi)&\eqdef \trOne(\varphi)\land\trOne(\psi),\nonumber\\
\trOne(\lnot\varphi)&\eqdef \lnot\trOne(\varphi),\nonumber\\
\trOne(\langle\pi\rangle)&\eqdef \langle\trOne(\pi)\rangle,\nonumber\\
\trOne(\varphi?)&\eqdef \trOne(\varphi)?,\nonumber\\
\trOne(\pi_1\star\pi_2)&\eqdef \trOne(\pi_1)\star\trOne(\pi_2)&&\text{for $\star\in\{\cup,\circ\},$}\nonumber\\
\trOne(\pi^*)&\eqdef \trOne(\pi)^*,\nonumber\\
\trOne(\lo(\pi))&\eqdef \{\trOne(\pi)(x,x)\}[x,x].\label{eqn:translation_loop}
\end{align}

\begin{proposition}\AP
$\loopCPDL\lleq[\Kripke]\CPDLg{\Graphloop}$ via $\trOne$.
\end{proposition}

\paragraph{Translation from $\CPDLg{\Graphloop}$ to $\loopCPDL$}\label{sec:to_loopCPDL}

In $\CPDLg{\Graphloop}$ the only allowed "conjunctive programs" $C[x,y]$ are such that $x=y$ and all "atoms" in $C$ are of the form $\pi(x,x)$, for $\pi$ a $\CPDLg{\Graphloop}$-"program".

\knowledgenewrobustcmd{\trTwo}{\cmdkl{\tr_2}}
Consider the following translation $\AP\intro*\trTwo$ from $\CPDLg{\Graphloop}$-"formulas" to $\loopCPDL$-"formulas" and from $\CPDLg{\Graphloop}$-"programs" to $\loopCPDL$-"programs": take all clauses of $\trOne$ above
except \eqref{eqn:translation_loop} replacing $\trOne$ by $\trTwo$ and adding
\begin{align*}
\trTwo(C[x,x])&=\left(\bigwedge_{\pi(x,x)\in C}\lo({\trTwo(\pi)})\right)? 
\end{align*}

\begin{proposition}\AP
$\CPDLg{\Graphloop}\lleq[\Kripke]\loopCPDL$ via $\trTwo$.
\end{proposition}

\subsection{\texorpdfstring{$\ICPDL$}{ICPDL}, \texorpdfstring{$\CPDLg{\Tw[1]}$}{CPDL⁺(TW₁)}, and \texorpdfstring{$\CPDLg{\Tw[2]}$}{CPDL⁺(TW₂)} are Equi-expressive}
\label{ssec:ICPDL-CPDL-equivalence}

As a consequence of \Cref{rk:intersecion_does_not_add_expressive_power}, the next \Cref{prop:translation_from_ICPDL,prop:translation_from_TW2} and the fact that $\Graphcap\subseteq\Tw[1]$, we obtain the following theorem:
\begin{thm}\label{thm:ICPDL_equals_TW1_equals_TW2}
$\ICPDL\langsemequiv[\Kripke]\CPDLg{\Tw[1]}\langsemequiv[\Kripke]\CPDLg{\Tw[2]}$ via polynomial time translations.
\end{thm}

\paragraph{Translation from $\ICPDL$ to $\CPDLg{\Graphcap}$}
\knowledgenewrobustcmd{\trThree}{\cmdkl{\tr_3}}
Consider the translation $\AP\intro*\trThree$ from $\ICPDL$-"formulas" to $\CPDLg{\Graphcap}$-"formulas" and from $\ICPDL$-"programs" to $\CPDLg{\Graphcap}$-"programs" defined as follows: 
Take all clauses of $\trOne$ from Section~\ref{sec:from_loopCPDL} replacing $\trOne$ by $\trThree$ except \eqref{eqn:translation_loop} and adding
\begin{align*}
\trThree(\pi_1\cap\pi_2)&=\{\trThree(\pi_1)(x,y),\trThree(\pi_2)(x,y)\}[x,y].
\end{align*}
Observe that the "underlying graph" of any  $C[x,y]\in\subexpr(\trThree(\pi))$ is in $\Graphcap$, and hence $\trThree(\pi)$ is in $\CPDLg{\Graphcap}$.
\color{black} %
\begin{proposition}\AP\label{prop:translation_from_ICPDL}
$\ICPDL\lleq[\Kripke]\CPDLg{\Graphcap}$ via $\trThree$.
\end{proposition}

\paragraph{Translation from $\twtcpdl$ to $\ICPDL$}\label{sec:to_ICPDL}
\AP
We start with a technical result. For a graph $G$ and $B\subseteq V(G)$, we say that $B$ ""is a clique in"" $G$ if 
$\set{\set{x,y}\mid x,y\in B,x\neq y}\subseteq E(G)$.
The following lemma 
states that when dealing with $C'[x,y]$ in $\CPDLg{\Tw[n]}$, one may suppose that $\uGraph{C[x_1,x_2]}$ has a "tree decomposition" of "tree-width" $\leq n$ with the additional property that each bag "is a clique in" $\uGraph{C[x_1,x_2]}$. 
\begin{lemma}\AP\label{lem:bags_are_cliques}
Given a $\CPDLg{\Tw[n]}$-"conjunctive program" $C'[x,y]$ one can compute in polynomial time a $\CPDLg{\Tw[n]}$-"conjunctive program" $C[x,y]\semequiv C'[x,y]$ with $\vars(C)=\vars(C')$, and a "tree decomposition" $(T, \bagmap)$ of "tree-width" $\leq n$ of $\uGraph{C[x,y]}$ such that for any $b\in V(T)$, $\bagmap(b)$ "is a clique in" $\uGraph{C[x,y]}$.
\end{lemma}
\begin{proof}
Let $(T, \bagmap)$ be a "tree decomposition" of "tree-width" at most $n$ of $\uGraph{C'[x_1,x_2]}$. Let $A = \Prog\cap\subexpr(C'[x,y])$ be the set of "atomic programs" in $C'$ or its "subexpressions". Define $C$ as $C'\cup D$, where 
\[
D=\{((\cup_{a \in A} a) \cup (\cup_{a \in A} \bar a))^*(z_1,z_2)\mid b\in V(T), z_1,z_2\in\bagmap(b),z_1\neq z_2\}.
\]
Notice that $(T, \bagmap)$ is a "tree decomposition" of $\uGraph{C}$ where every bag is a clique in $\uGraph{D}$. Since $\uGraphC{C'}$ is connected and since the semantics of $\CPDLp$ are compositional and hence indifferent for "atoms" not occurring in the expression, then $C[x,y]\semequiv C'[x,y]$ and it is clear that $\vars(D)\subseteq\vars(C')$. Furthermore, $(T, \bagmap)$ is also a "tree decomposition" of $\uGraph{C[x,y]}$. Since $D$ has at most $|\vars(C')|^2$ "atoms" and computing a "tree decomposition" of "tree-width" at most $n$ can be done in linear time \cite{bodlaender1993linear}, then the whole construction remains in polynomial time.
\end{proof}

\noindent
The next lemma is the key ingredient of the translation from $\twtcpdl$ to $\ICPDL$: 
\begin{lemma}\AP\label{lemita}
Let $C$ be a finite set of "atoms" of the form $\pi(z_1,z_2)$, 
where $\pi$ is an \ICPDL-"program" and $z_1,z_2 \in \Vars$, 
$\uGraphC{C}$ is a clique, and $|\vars(C)|\leq 3$. Then:
\begin{enumerate}
\item For any $x,y\in\vars(C)$, $x\neq y$, there is an $\ICPDL$-"program" $\pi_{C[x,y]}$ such that $\pi_{C[x,y]}\semequiv C[x,y]$.

\item For any $x\in\vars(C)$ there is an $\ICPDL$-"program" $\pi_{C[x,x]}$ such that $\pi_{C[x,x]}\semequiv C[x,x]$.
\end{enumerate}
Further, these translations are in polynomial time.
\end{lemma}
\begin{proof}
    For $z_1,z_2\in\vars(C)$, $z_1\neq z_2$, let
    \begin{align*}
    \Pi_{z_1z_2}&=\left(\bigcap_{\pi(z_1,z_2)\in C}\pi\right) \cap \left(\bigcap_{\pi(z_2,z_1)\in C} \reverseof{\pi}\right)\\
    \Pi_{z_1}&=\left(\bigwedge_{\pi(z_1,z_1)\in C} \tup{\pi\cap\epsilon}\right)?
    \end{align*}
    where a conjunction with empty range is defined by $\tup{\epsilon}$.
    
    For $x\neq y$ define the $\ICPDL$-"programs" $\pi_{C[x,y]}$ and $\pi_{C[x,x]}$ as follows:
    \begin{itemize}
    \item If $\vars(C)=\{x,y,z\}$ with $x\neq z \neq y$ then
    \begin{align*}
    \pi_{C[x,y]}&=\Pi_x\circ( \Pi_{xy}\cap(\Pi_{xz}\circ\Pi_z\circ\Pi_{zy}))\circ\Pi_y\\
    \pi_{C[x,x]}&=\tup{\pi_{C[x,y]}}?
    \end{align*}
    
    \item If $\vars(C)=\{x,y\}$ then
    \begin{align*}
    \pi_{C[x,y]}&=\Pi_x\circ \Pi_{xy} \circ\Pi_y\\
    \pi_{C[x,x]}&=\tup{\pi_{C[x,y]}}?
    \end{align*}
    
    \item If $\vars(C)=\{x\}$ then
    \begin{align*}
    \pi_{C[x,x]}&=\Pi_x
    \end{align*}
    \end{itemize}
    It can be shown that $\pi_{C[x,y]}\semequiv C[x,y]$ and $\pi_{C[x,x]}\semequiv C[x,x]$.
    \end{proof}

\noindent
Figure \ref{fig:lemita} illustrates an example of $\pi_{C[x,y]}$ and of $\pi_{C[x,x]}$ in Lemma \ref{lemita}. The general case is more involved, as $C$ may contain -- as illustrated in Figure \ref{fig:example_ugraph}.A -- several "atoms" $\pi_1(z_1,z_2),\dots\pi_n(z_1,z_2)$ for the same variables $z_1,z_2$ and also contain "atoms" of the form $\pi(z_1,z_1)$ not shown in the example. Furthermore, one would sometimes need to `reverse' the direction of the "programs" to obtain the desired $\ICPDL$-"program". 

Using \Cref{lemita}, one can show the translation which, for technical reasons, as it simplifies the proof, translates $\twticpdl$ "expressions" (instead of $\twtcpdl$) to $\ICPDL$.
\begin{figure}
\ifarxiv
    \includegraphics[scale=0.25]{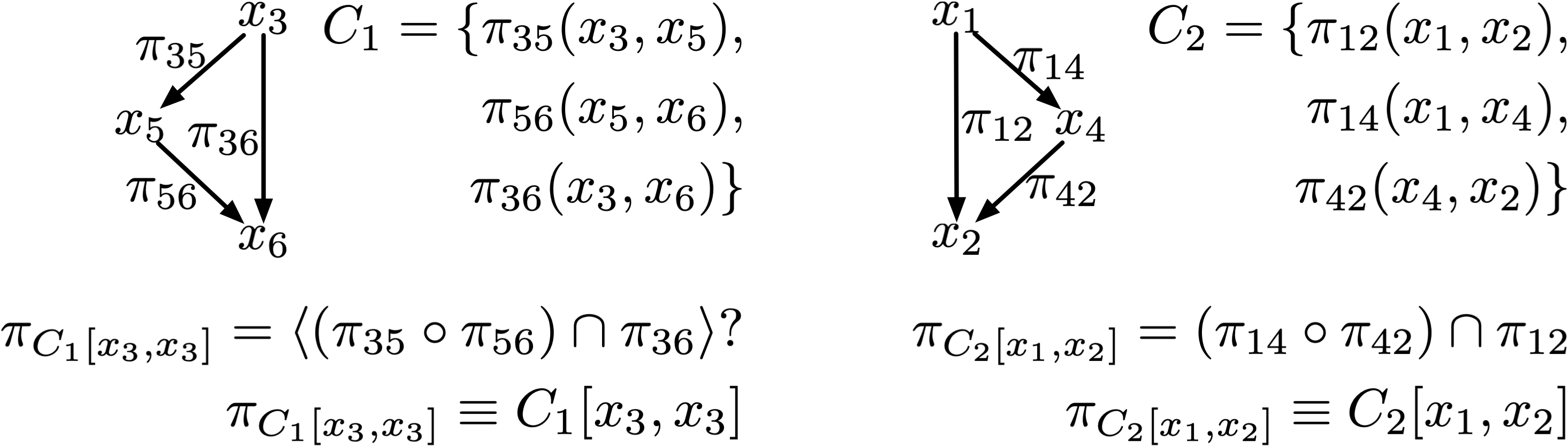}
\else
    \includegraphics[width=.47\textwidth]{img/lemita.png}
\fi
\caption{An example of translation of a "conjunctive program" made up of only $\ICPDL$ "atoms" and 3 variables to \ICPDL. The names of the $\ICPDL$-"programs" $\pi_{C[x_3,x_3]}$ and $\pi_{C[x_1,x_2]}$ are the ones used in Lemma \ref{lemita}.}
\label{fig:lemita}
\end{figure}

\begin{figure}
\ifarxiv
    \includegraphics[scale=0.25]{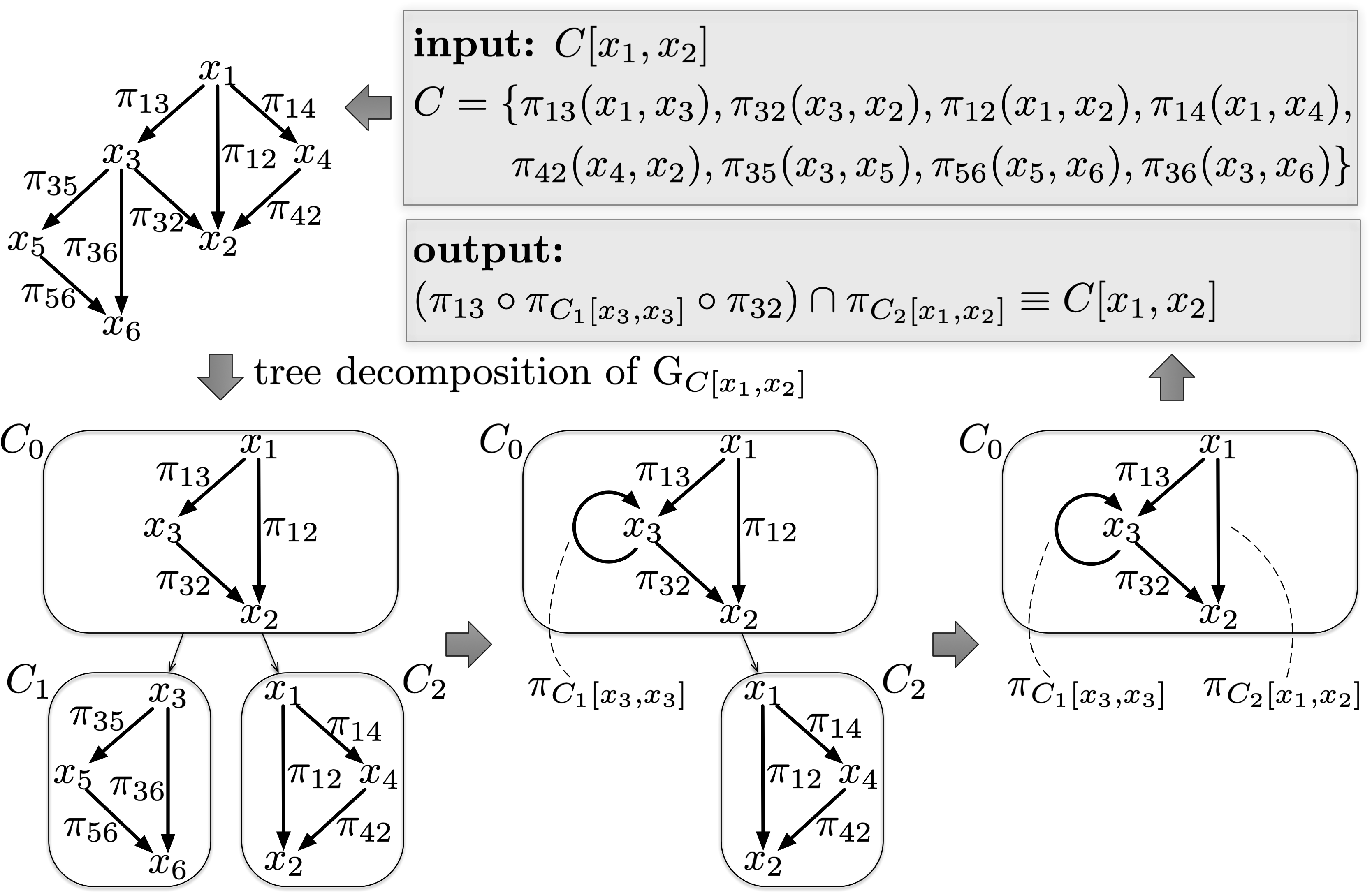}
\else
    \includegraphics[scale=0.20]{img/example_treedecomp.png}
\fi
\caption{A "conjunctive program" $C[x_1,x_2]$ made up of only $\ICPDL$ "programs" and the steps of the algorithm to construct a semantically equivalent $\ICPDL$ formula from a "tree decomposition" of "tree-width" 2 of $\uGraph{C[x_1,x_2]}$. Rounded boxes represent the bags of such "tree decomposition" where variables $x_1,\dots,x_6$ are grouped (additional information such as edges and labels is not part of the "tree decomposition"). See Figure~\ref{fig:lemita} for the definition of $\pi_{C[x_3,x_3]}$ and $\pi_{C[x_1,x_2]}$.}
\label{fig:example_treedecomp}
\end{figure}

\begin{proposition}\AP\label{prop:translation_from_TW2}
$\twticpdl\lleq[\Kripke]\ICPDL$ via a polynomial time translation.
\end{proposition}
\begin{proof}
    \color{santichange}
    For an $\twticpdl$ "expression" $e$, 
    let: 
    \begin{itemize}
        \item $c_1(e)=\bigcup_{C[x,y]\in\subexpr(e)}\vars(C)$, 
        \item $c_2(e)$ be the set of expressions in $\subexpr(e)$ which are not $\ICPDL$-expressions, and let 
        \item $c(e)=(|c_1(e)|+|c_2(e)|,\pdlsize{e})$.
    \end{itemize}
     We define a translation $\tr(e)$ from $\twticpdl$-"formulas" to $\ICPDL$-"formulas" and from $\twticpdl$-"programs" to $\ICPDL$-"programs" by recursion. We guarantee that the base case of the definition is reached because in each recursive call the value of $c$ strictly decreases in the lexicographic order of $\Nat\times\Nat$. 
    \color{black}
    
    For the cases other than "conjunctive programs", $\tr$ is defined as $\trOne$ from Section~\ref{sec:from_loopCPDL} replacing $\trOne$ by $\tr$, and $\star\in\{\cup,\circ\}$ by $\star\in\{\cup,\circ,\cap\}$. Suppose that $\tr(e)=\dots\tr(e')\dots$ according to one of these rules. If $c(e)=(m,n)$ then $c(e')$ is of the form $(m',n')$ with $m'\leq m$ and $n'<n$. Therefore $c(e')<c(e)$. 
    
    For the case of a "conjunctive program" $C[x,y]$, suppose
    \begin{align*}
    C&=\set{\pi_1[x_1,y_1],\dots,\pi_n[x_n,y_n]}.
    \end{align*}
    be such that $\uGraph{C[x,y]}\in\Tw[2]$.
    If all $\pi_i$ are $\ICPDL$-"programs" and $|\vars(C)|\leq3$, then $C$ satisfies the hypothesis of Lemma \ref{lemita}. In case $x=y$, there is an $\ICPDL$-"program" $\pi_{C[x,x]}\semequiv C[x,x]$ and we define $\tr(C[x,y])=\pi_{C[x,x]}$. In case $x\neq y$, there is an $\ICPDL$-"program" $\pi_{C[x,y]}\semequiv C[x,y]$ and we define $\tr(C[x,y])=\pi_{C[x,y]}$.
    
    Else there is some $\pi_i$ that is not an $\ICPDL$-"program" or $|\vars(C)|>3$. We define $\tr(C[x,y])=\tr(D[x,y])$ for a set of "atoms" $D$ to be defined next.
    
    \begin{enumerate}
    \item\label{trad:case:1} If there is $i$ such that $\pi_i$ is not an $\ICPDL$-"program", then $D=C\setminus\{\pi_i[x_i,y_i]\}\cup\{\tr(\pi_i)[x_i,y_i]\}$. It is straightforward that $\tr(C[x,y])\semequiv\tr(D[x,y])$. Observe that in this case 
    (1) $c_1(D[x,y])\subseteq c_1(C[x,y])$, since by construction we do not add new variables to $D$; and
    (2) $c_2(D[x,y])\subsetneq c_2(C[x,y])$ since in $D$ we removed at least one expression which was not an $\ICPDL$-"program" and we do not add new ones.
    Hence $c(D[x,y])<c(C[x,y])$.

    \item\label{trad:case:2} Else, all $\pi_i$ are already $\ICPDL$-"programs" and $|\vars(C)|>3$. Let $(T, \bagmap)$ be a "tree decomposition" of the "underlying graph" of $C$ with the root containing $x$ and $y$ (observe that this is possible by the fact that $\uGraph{C[x,y]}$ contains the "edge" $\{x,y\}$), and of width $\leq 2$ such that any leaf of $T$ is a bag of size at least 2 (and at most 3), and is not included in its parent bag. If $T$ is a single bag $B$, then we are done: in case $x=y$, define $\tr(C[x,y])=\pi_{C[x,x]}$ and otherwise define $\tr(C[x,y])=\pi_{C[x,y]}$.
    If $T$ has at least two bags, pick any leaf $B$ from $T$ and let $B'$ be the parent of $B$ in $T$. Let $C_B$ be the set of all "atoms" $A\in C$ such that $\vars(A)\subseteq B$. Observe that $C_B$ satisfies the hypothesis of Lemma \ref{lemita}.
    \begin{itemize}
        \item  Suppose that $B'\cap B=\{z_1\}$ and $B\supseteq\{z_1,z_2\}$ ($z_1\neq z_2$). We define $D$ as the set of all "atoms" in $C$ except those 
        in $C_B$
        plus the \ICPDL-"program" $\pi_{C_B[z_1,z_1]}\semequiv C_B[z_1,z_1]$ from Lemma \ref{lemita}. Notice that $z_2\in \vars(C)\setminus\vars(D)$.
    
        \item  Suppose that $B'\cap B=\{z_1,z_2\}$ ($z_1\neq z_2$) and $B=\{z_1,z_2,z_3\}$ ($z_1\neq z_3\neq z_2$). We define $D$ as the set of all "atoms" in $C$ except those 
        in $C_B$
        plus the \ICPDL-"program" $\pi_{C_B[z_1,z_2]} \semequiv C_B[z_1,z_2]$ from Lemma \ref{lemita}. Notice that $z_3\in \vars(C)\setminus\vars(D)$.
    \end{itemize}
    It is straightforward to see that $\tr(C[x,y])\semequiv\tr(D[x,y])$. Furthermore, observe that 
    $c_1(D[x,y])\subsetneq c_1(C[x,y])$, and
    $c_2(D[x,y])\subseteq c_2(C[x,y])$.
    Hence $c(D[x,y])<c(C[x,y])$.
    \end{enumerate}
    We are only left with the task of showing that the translation is in polynomial time, which we show next.

We first analyse the following algorithm \ref{algo_particular} emerging from case \ref{trad:case:2} 
that translates a $\twtcpdl$-"conjunctive program" of the form $C[x,y]$, where $C$ is a set of "atoms" of the form $\pi(z_1,z_2)$, for $\pi$ an \ICPDL-"program" and $z_1,z_2 \in \Vars$, into a semantically equivalent \ICPDL-"program".

\begin{algorithm}
\SetKwData{Left}{left}\SetKwData{This}{this}\SetKwData{Up}{up}
\SetKwFunction{Union}{Union}\SetKwFunction{FindCompress}{FindCompress}
\SetKwInOut{Input}{input}\SetKwInOut{Output}{output}
\Input{A $\twtcpdl$-"program" $C'[x,y]$, where $C'$ is a 
set of "atoms" of the form $\pi(z_1,z_2)$, 
for $\pi$ an \ICPDL-"program" and $z_1,z_2 \in \Vars$}
\Output{An \ICPDL-"program" $\pi$ such that $\pi\semequiv C'[x,y]$}
\BlankLine

Let $C[x,y]\semequiv C'[x,y]$ with $\vars(C)=\vars(C')$ and let $(T, \bagmap)$ be a "tree decomposition" of "tree-width" at most $2$ of $\uGraph{C[x,y]}$ such that any $b\in V(T)$, $\bagmap(b)$ "is a clique in" $\uGraph{C[x,y]}$ (Lemma \ref{lem:bags_are_cliques}).

\While{$T$ is not a single bag}
    {
    let $B$ be a leaf of $T$\; 

    let $B'$ be the parent of $B$ in $T$\;

    $C_B:=\{\pi(z_1,z_2)\in C\mid \{z_1,z_2\}\subseteq B\}$\;

    \If{$B'\cap B=\{z_1\}$ and $B\setminus B'\neq\emptyset$}
    {

    $C:=C\setminus C_B \cup\{\pi_{C_B[z_1,z_1]}\}$\;

    }
    \If{$B'\cap B=\{z_1,z_2\}$ ($z_1\neq z_2$) and $B=\{z_1,z_2,z_3\}$ ($z_1\neq z_3\neq z_2$)}
    {

    $C:=C\setminus C_B\cup\{\pi_{C_B[z_1,z_2]}\}$\;
    }
    Remove $B$ from $T$
    }
    {

\eIf{$x\neq y$}
    {\Return $\pi_{C[x,y]}$}
    {\Return $\pi_{C[x,x]}$}
}

\caption{Computing the translation from $\twticpdl$ to $\ICPDL$ in a special case}\label{algo_particular}
\end{algorithm}

\noindent
Observe that the first line can be done in polynomial time by Lemma \ref{lem:bags_are_cliques}, and that $|C|=|C'|+O(|\vars(C')^2|)$.

For $\pi$ an \ICPDL-"program", let $c(\pi)$ be the syntactic complexity of $\pi$ and for a set $D$ of "atoms" of the form $\pi(z,z')$, where $\pi$ is an \ICPDL-"program", let $c(D)=\sum_{\pi(z,z')\in D}c(\pi)$.

The construction of the \ICPDL-"programs" $\pi_{C_B[z_1,z_1]}$ and $\pi_{C_B[z_1,z_2]}$ from Lemma \ref{lemita} can be done in time $O(c(C_B))$. Observe that the above algorithm removes from $C$ the set of "atoms" in $C_B$ but adds a new \ICPDL-"program". We check that the syntactic complexity of this "program" does not grow too much along the iterations of the cycle.
By inspecting the definition of $\pi_{C_B[z_1,z_1]}$ and $\pi_{C_B[z_1,z_2]}$ from Lemma \ref{lemita}, one can see that there is a constant $k$ such that for any $B$ and $C_B$ in any iteration of the cycle we have 
$
c(\pi_{C_B[z_1,z_1]}),c(\pi_{C_B[z_1,z_2]})\leq k|C_B|+c(C_B).
$
Hence the calculation of $\pi_{C_B[z_1,z_1]}$ and $\pi_{C_B[z_1,z_2]}$ can be done in time $O(k|C|+c(C))$; also the output $\pi$ of algorithm \ref{algo_particular} satisfies $c(\pi)\leq k|C|+c(C)$.

As mentioned in the proof of Lemma \ref{lem:bags_are_cliques}, the "tree decomposition" $T$ of "tree-width" at most $2$ can be computed in polynomial time and hence there are polynomially many bags in $T$. Since at each step of the cycle one bag of $T$ is removed, there are polynomially many iterations. Hence algorithm \ref{algo_particular} runs in polynomial time.

We next see how to use algorithm \ref{algo_particular} to compute the translation of any $\twtcpdl$-"program" into a semantically equivalent \ICPDL-"program". Given a $\twtcpdl$-"program" $\pi$ we proceed as follows. If there is no $C[x,y]\in\subexpr(\pi)$, then $\pi$ is already an \ICPDL-"program" and we are done. Else, pick $C[x,y]\in\subexpr(\pi)$ where $C$ is a set of \ICPDL-"programs". Apply algorithm \ref{algo_particular} to obtain an \ICPDL-"program" $\pi'$ that is semantically equivalent to $C[x,y]$. Replace $C[x,y]$ by $\pi'$ in $\pi$ and repeat the procedure. This algorithm is clearly polynomial and computes the desired \ICPDL-"program".
\end{proof}

\noindent
Figure \ref{fig:example_treedecomp} illustrates an example of a "conjunctive program" $C[x_1,x_2]$ 
where $C$ contains only "atoms" of the form $\pi(z_1,z_2)$ for $\pi$ an $\ICPDL$-"program". 
Given a "tree decomposition" $(T, \bagmap)$ of $\uGraph{C[x_1,x_2]}$ of "tree-width" at most 2 whose bags are all cliques in $\uGraph{C[x_1,x_2]}$ (this can be assumed without loss of generality by Lemma \ref{lem:bags_are_cliques}), one can successively remove its leaves until a single bag is obtained. Each time a leaf $B$ is removed, all "atoms" in $C$ using (only) variables from $B$ are removed, but -- thanks to Lemma \ref{lemita} -- an $\ICPDL$-"program" with the information of $B$ is added instead. The general scenario is again more complex, since $C$ may contain "programs" $\pi$ which are not $\ICPDL$-"programs" (namely, "conjunctive programs"). In this case, we need to recursively eliminate such $\pi$s starting with those containing only $\ICPDL$-"programs".

%% file: characterization.tex
\section{Indistinguishability}\label{sec:caracterization}
In this section we characterize the expressive power of $\CPDLg{\Tw}$ via a restricted form of (bi)simulation using pebbles, in what resembles the $k$-pebble game for characterizing finite variable first-order fragments.

\subsection{Simulation Relation}\label{sec:simulation}
\AP
We will define the notion of "$k$-simulation" between pairs $(K,K')$ of 
"$\Rels$-structures"
via a two-player zero-sum graph game $\intro*\kSimGame$. 
\AP
The arena of the game has a set of positions $S \cup D$, %
where
\begin{align*}
    S &= \set{s} \times \kHoms(K,K')\\
    D &= \set{d_1,\dotsc, d_k} \times (\dom{K}^k \times \dom{K'}^k)
\end{align*}
where Spoiler owns all positions from $S$ and Duplicator all positions from $D$.
The set of moves of $\kSimGame$ is the smallest set satisfying the following:\footnote{For ease of notation we write $(s,\bar u,\bar v)$ instead of $(s,(\bar u,\bar v))$ and the same for $(d_i,\bar u,\bar v)$.}
\begin{enumerate}
    \item There is a move from $(s,\bar u, \bar v)$ to $(d_i,\bar u',\bar v)$ 
    if $\bar u' = \bar u[i \mapcoord w]$, where $w$ is a "world" from $K$ at "distance" $\leq 1$ from $\bar u[j]$, for some $1 \leq j \leq k$ with $i\neq j$; and 
    \item There is a move from $(d_i,\bar u',\bar v)$ to $(s,\bar u',\bar v')$ 
    if $\bar v' = \bar v[i \mapcoord w]$, where $w$ is a "world" from $K'$ at "distance" $\leq 1$ from $\bar v[j]$, for some $1 \leq j \leq k$ with $i\neq j$.
\end{enumerate}
The winning condition for Duplicator is just any infinite play, which is a form of ``Safety condition'', which implies (positional) determinacy of the game.\footnote{These games are often presented as having infinite duration. Note that Spoiler is never `stuck' as she can always play by putting in component 1 the "world" contained in component 2, for example. However, Duplicator can get stuck.
\newcommand{\spoilerWin}{(s,\textit{Win})}%
The infinite duration variant would be adding a new position `$\spoilerWin$' owned by Spoiler, and having moves from $\spoilerWin$ to $\spoilerWin$, and from every position of $D$ to $\spoilerWin$. In this way, the winning condition of Duplicator are all plays which avoid going through position $\spoilerWin$. Since it is a Safety (hence parity) condition, it follows that the game is positionally determined \cite{DBLP:conf/focs/EmersonJ91}.}
For more information on this kind of game we refer the reader to~\cite[Chapter~2]{gradel2003automata}.

Inspired by the pebble games \cite{immerman1982upper} and the existential pebble games \cite{DBLP:journals/jcss/KolaitisV95,DBLP:journals/jcss/KolaitisV00}, the intuition behind the game $\kSimGame$ is that Spoiler owns $k$ numbered pebbles placed over the elements of $K$ and Duplicator owns $k$ pebbles placed over the elements of $K'$. In each round, Spoiler can move one of her pebbles, provided that the destination position is at distance at most 1 from some other pebble owned by her in $K$. Duplicator answers in the same way with his corresponding pebble (with the same restriction of moving close to another pebble owned by him) but in structure $K'$.

The notion of ""$k$-$\Univ$-simulation"" is defined analogously via the game $\intro*\kSimGameE$, defined as $\kSimGame$ but with the following modifications:
The arena of the game has now a set of positions $S \cup D$, %
where $S$ is as before and
\begin{align*}
    D &= \set{d_0,d_1,\dotsc, d_k} \times (\dom{K}^k \times \dom{K'}^k)
\end{align*}
Besides rules (1) and (2) above, two more rules are added:
\begin{enumerate}
    \item[(3)] There is a move from $(s,\bar u, \bar v)$ to $(d_0,\bar u',\bar v)$ 
    if $\bar u'[1] = \dots = \bar u'[k]$; and 
    \item[(4)] There is a move from $(d_0,\bar u', \bar v)$ to $(s,\bar u',\bar v')$ (observe that since $\bar u'[1] = \dots = \bar u'[k]$, then $\bar v'[1] = \dots = \bar v'[k]$)
\end{enumerate}
The intuition of these new rules is that at any round, Spoiler can grab all her pebbles and stack them together in the same destination element $u'$ of $K$, and Duplicator has to respond in the same way, stacking all his pebbles together in the same destination element $v'$ of $K'$.

Let  $K, K'$ be 
"$\Rels$-structures", and let $\bar v\in \dom{K}^k$ and $\bar v'\in \dom{K'}^k$.

\AP
We say that $K',\bar v'$ ""$k$-simulates"" $K,\bar v$, denoted $K,\bar v \intro*\pebblesim k K',\bar v'$, \phantomintro{\notpebblesim} if 
the following holds: in case all "worlds" in $\bar v$ are in the same connected component of $K$, then 1) $(s,\bar v, \bar v')$ is a valid position of $\kSimGame$ on $(K,K')$ (\ie, $\bar v, \bar v'$ induce a "partial homomorphism"),
and 2) Duplicator has a winning strategy in $\kSimGame$ from $(s,\bar v, \bar v')$; in case $\bar v$ has "worlds" in different connected components of $K$ then we assume that $K,\bar v \pebblesim k K',\bar v'$ holds automatically. See \Cref{fig:game} for an example. 

We say that $K',\bar v'$ ""$k$-$\Univ$-simulates"" $K,\bar v$, denoted $K,\bar v \intro*\pebblesimE k K',\bar v'$, \phantomintro{\notpebblesimE} if 
1) $(s,\bar v, \bar v')$ is a valid position of $\kSimGameE$ on $(K,K')$ (\ie, $\bar v, \bar v'$ induce a "partial homomorphism"),
and 2) Duplicator has a winning strategy in $\kSimGameE$ from $(s,\bar v, \bar v')$. 

\begin{remark}\label{rem:connected}
Observe that if all the "worlds" in $\bar v$ are in the same connected component of $K$, any move of $\kSimGame$ [resp.\ of $\kSimGameE$] preserves this property regarding the connected component of $K$.
\end{remark}

\begin{figure*}
\ifarxiv
    \includegraphics[width=1\textwidth]{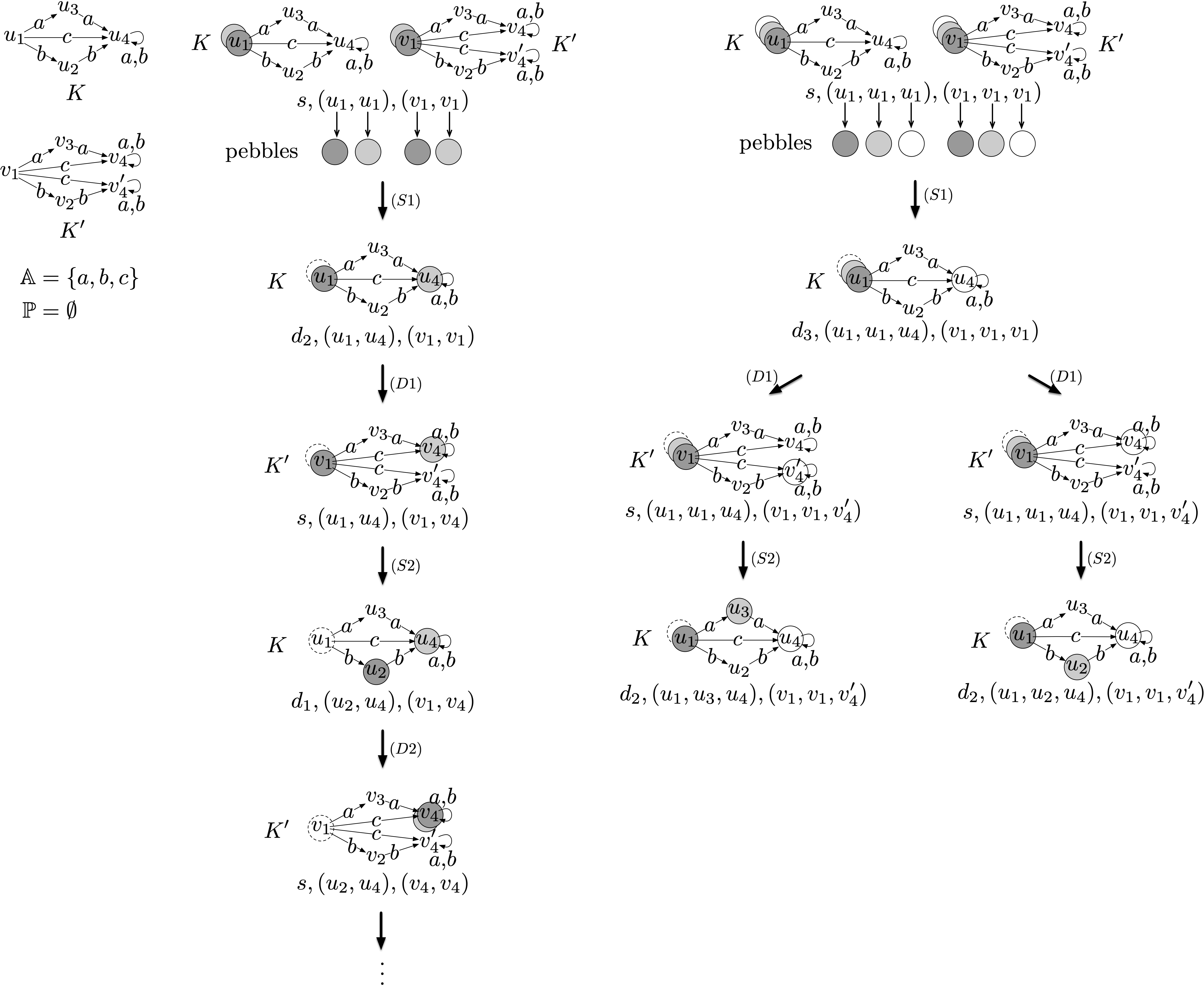}
\else
    \includegraphics[width=1\textwidth]{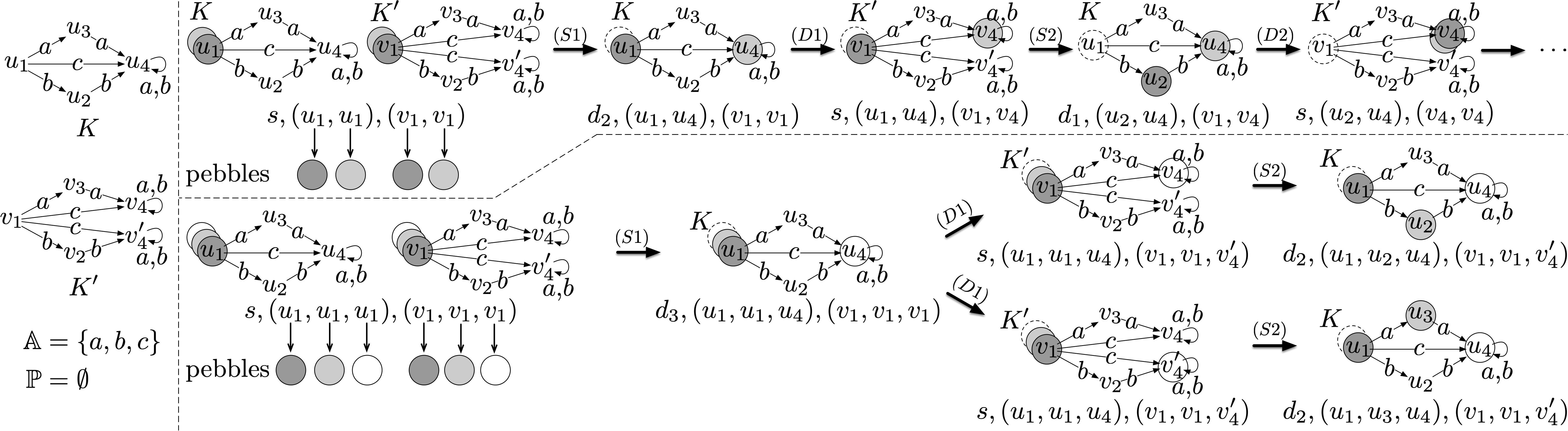}
\fi

    \caption{
    Left: Kripke structures $K$ and $K'$. 
    \ifarxiv
        Center: 
    \else
        Top:
    \fi
    an example of a game $\kSimGame[2]$ over $K$ and $K'$. The order of moves is: (S1), (D1), (S2), (D2). In fact, Duplicator can always answer to any Spoiler's move in $\kSimGame[2]$, and so $K,u_1 \pebblesim{2} K',v_1$. 
    \ifarxiv
        Right: 
    \else
        Bottom:
    \fi
    a tree representation of a winning strategy for Spoiler in $\kSimGame[3]$ over $K$ and $K'$. All possible answers by Duplicator are shown. Duplicator cannot answer after (S2) and so $K,u_1 \not\pebblesim{3} K',v_1$. Observe that $K,u_1,u_1\models\pi$ and $K,v_1,v_1\not\models\pi$ for $\pi=\{c(x_s,x_4),a(x_s,x_3),a(x_3,x_4),b(x_s,x_2),b(x_2,x_4)\}[x_s,x_s]$. Graphically the vectors in the positions can be thought of as moving colored pebbles on the "Kripke structures". Only "Kripke structures" that change pebble positions with respect to the previous move are shown.}

\label{fig:game}
\end{figure*}

\noindent
For any $k' < k$, we will usually write $K,(v_1, \dotsc, v_{k'}) \pebblesim k K',(v'_1, \dotsc, v'_{k'})$ to denote 
\[
K,(v_1, \dotsc, v_{k'}, \underbrace{v_{k'}, \dotsc, v_{k'}}_{\text{$k-k'$ times}}) \pebblesim k K',(v'_1, \dotsc, v'_{k'},\underbrace{v'_{k'}, \dotsc, v'_{k'}}_{\text{$k-k'$ times}}).
\]
We use the same notation for $\pebblesimE k$ instead of $\pebblesim k$, and also for other notions to be introduced later on such as $\pebblequasibisim k$, $\pebblequasibisimE k$, $\pebblebisim k$ and $\pebblebisimE k$. Furthermore, for simplicity of notation we will drop parentheses and write $K,v_1, \dotsc, v_k$ instead of $K,(v_1, \dotsc, v_k)$.

\begin{thm}\label{thm:sim-char}
    Let $k\geq 2$.
    Given 
    "$\Rels$-structures"
    $K,K'$ where $K'$ is of "finite degree" [resp.\ finite] and "worlds" $u,v\in\dom{K}$ and $u',v'\in\dom{K'}$, the following are equivalent
    \begin{enumerate}
        \item for every "positive" $\CPDLg{\Tw[k]}$-"formula" [resp.\ positive $\UCPDLg{\Tw[k]}$-"formula"] $\phi$, we have $K,v \models \phi$ implies $K',v' \models \phi$; and
        \item $K,v \pebblesim{k+1} K',v'$ [resp.\ $K,v \pebblesimE{k+1} K',v'$];
    \end{enumerate}
    and the following are equivalent
    \begin{enumerate}
        \item for every "positive" $\CPDLg{\Tw[k]}$-"program" [resp.\ positive $\UCPDLg{\Tw[k]}$-"program"] $\pi$, we have $K,u,v \models \pi$ implies $K',u',v' \models \pi$; and
        \item $K,u,v \pebblesim{k+1} K',u',v'$ [resp.\ $K,u,v \pebblesimE{k+1} K',u',v'$].
    \end{enumerate}
Furthermore, the hypothesis of "finite degree" [resp.\ finiteness] is only needed for the $1$-to-$2$ implications.
\end{thm}
\color{black}
\begin{proof}

\input{simchar-1-to-2}

    \medskip

\input{simchar-2-to-1}
\end{proof}

\noindent
The following claim aims at stating necessary conditions for getting the logical implication from one pointed structure to another. The logic $\+L$ is a general logic whose semantics are based on $\sigma$-structures and which expresses properties of elements of such structures (namely, the formulas are unary).
\begin{claim}\AP\label{cla:conditions_for_logical_implication}
Let $K$ and $K'$ be $\sigma$-structures of a logic $\+L$, and let $u\in\dom{K}$ and $u'\in\dom{K'}$. If
\begin{enumerate}
    \item for any $\+L$-formula $\varphi$, if $K,u\models\varphi$ then there is a finite substructure $K_f$ of $K$ with $u\in\dom{K_f}$ such that $K_f,u\models\varphi$;
    \item any $\+L$-formula is closed under homomorphisms; and
    \item if $K_f$ is a finite substructure of $K$ with $u\in\dom{K_f}$ then there is an embedding from $K_f$ to $K'$ that maps $u$ to $u'$
\end{enumerate}
then any formula of $\+L$ true in $K,u$ is also true in $K',u'$.
\end{claim}

\color{santichange}
\begin{lemma}\AP\label{lem:hyp_oneway_charact_needed}
The $1$-to-$2$ implications of Theorem \ref{thm:sim-char} fail if the corresponding hypothesis of finite degree [resp.\ finiteness] is dropped.
\end{lemma}
\color{black}

\begin{proof}
We show $\sigma$-structures $A$ and $B$ and worlds $u\in\dom{A},u'\in\dom{B}$ such that $A,u\models\phi\Leftrightarrow B,u'\models\phi$ for all positive $\UCPDLp$-formulas $\phi$ (it implies item (1) of Theorem  \ref{thm:sim-char}) but $A,u \not\pebblesim{2} B,u'$ (which in turn implies the negation of item (2) of Theorem  \ref{thm:sim-char}).

Consider the well known example of a tree $B$ with one branch of length $n$ of each length $n>0$, and the tree $A$ with one branch of length $n$ of each length $n>0$ and also an infinite branch (see "eg",~\cite[Fig.\ 2.5]{BRV01}). Formally, let $\sigma=\{a\}$, where $\arity(a)=2$, 
let $B$ be the $\sigma$-structure with domain $\{0\} \cup \{(i,j)\colon i\geq 1, j\in\{1,\dots,i\}\}$ and $a^{B}=\{(0,(i,1))\colon i\geq 1\} \cup \{((i,j),(i,j+1))\colon i,j\geq 1,j+1\leq i\}$ and let 
$A$ be the $\sigma$-structure with domain $\{0\} \cup \{(i,j)\colon i\geq 1, j\in\{1,\dots,i\}\}\cup\{(\omega,j)\colon j\in\Nat\}$ and $a^{A}=\{(0,(i,1))\colon i\geq 1,i\in\Nat\cup\{\omega\}\} \cup \{((i,j),(i,j+1))\colon i,j\geq 1,j+1\leq i\} \cup \{((\omega,j),(\omega,j+1))\colon j\geq 1\}$. 
One can check that the "positive" fragment of $\UCPDLp$ in the role of $\+L$, the structures $K=A$ and $K'=B$ and the points $u=u'=0$ satisfy conditions (1)--(3) of Claim \ref{cla:conditions_for_logical_implication}. Then for any "positive" $\UCPDLp$-"formula" $\phi$, we have that $A,0 \models \phi$ implies $B,0 \models \phi$; and for every "positive" $\UCPDLp$-"program" $\pi$, we have that $A,0,0 \models \pi$ implies $B,0,0 \models \pi$. 
However $A,0 \not\pebblesim{2} B,0$ (and so $A,0 \not\pebblesimE{2} B,0$). Spoiler's strategy consists of playing successively in the infinite branch $(\omega,0),(\omega,1),(\omega,2),\dots$ in $A$ alternating her two pebbles until Duplicator, forced to successively play in some finite branch $(i,0),(i,1),(i,2),\dots$ in $B$ for a fixed $i$, reaches $(i,i)$; at this point Spoiler wins.
\end{proof}

\color{santichange}

\noindent
Lemma~\ref{lem:hyp_oneway_charact_needed}  shows that the finiteness assumptions in Theorem~\ref{thm:sim-char} cannot simply be removed altogether. We do not know whether the finite degree [resp.\ finiteness] is the `optimal' hypothesis for the $1$-to-$2$ implications. In the spirit of the classical Hennessy-Milner theorem for modal logic, it may well be possible to replace these assumptions by weaker model-theoretic conditions, such as suitable saturation hypotheses. We leave this question open, since it falls outside the scope of the present paper.
\color{black}

\input{bisimulation}

\subsection{Computing the (bi)Simulation}\label{sec:bisim-calc}
Observe that when $K,K'$ are finite "structures@@kripke", the games $\kSimGame$ and $\kBisimGame$ are finite as well, with $O(|\dom{K}|^k \cdot |\dom{K'}|^k)$ positions. The winning condition for Duplicator is in both cases a ``Safety condition''. Hence, the set of winning positions can then be computed in time linear in the number of moves and positions of the game, for example via the classical McNaughton-Zielonka computation of its attractor decomposition. 
This yields a polynomial algorithm, which runs in time $(|\dom{K}|\cdot|\dom{K'}|)^{O(k)}$, for computing the "$k$-simulation" and "$k$-bisimulation" relations $\pebblesim{k}$ and $\pebblebisim{k}$.

\begin{proposition}\AP
    For every $k$, the relations $\pebblesim k$ and $\pebblebisim k$ on finite "structures@@kripke" $K,K'$ can be computed in polynomial time. So can the relations $\pebblesimE k$ and $\pebblebisimE k$.
\end{proposition}

%% file: simchar-1-to-2.tex
    We remark first that the $k$-simulation relation is closely related to the  ``existential $k$-pebble game'' \cite{DBLP:journals/jcss/KolaitisV95,DBLP:journals/jcss/KolaitisV00} on first-order structures with binary and unary signatures. It is known \cite[Theorem~4]{DBLP:conf/cp/AtseriasKV04} that for every pair $K,K'$ of \emph{finite} "structures@@kripke", if for every Boolean "conjunctive query" $q$ of "tree-width" $<k$ we have that $K \models q$ implies that $K' \models q$, then Duplicator wins the existential $k$-pebble game on $K,K'$. Since "conjunctive programs" are basically "conjunctive queries", it is not surprising that one can adapt this result to our setting.

\proofcase{1-to-2}
    We first show it for $\CPDLg{\Tw[k]}$. \color{black} By contrapositive, assume $K,u,v \notpebblesim{k+1} K',u',v'$. 
    For ease of notation we write $(w_1,w_2)$ instead of the $k$-tuple $(w_1,w_2,\dots,w_2)$. 
    Then $u$ and $v$ are in the same connected component of $K$ and either $((u,v), (u',v')) \not\in \kHoms(K,K')$ or 
    Spoiler has a strategy to win in a bounded number of rounds in $\kSimGame[k+1]$ starting from $(s,(u,v),(u',v'))$. 

    Suppose $((u,v),(u',v')) \not\in \kHoms(K,K')$. We show that there is a "positive" $\CPDLg{\Tw[k]}$-"program" $\pi$, such that $K,u,v \models \pi$ and $K',u',v' \not\models \pi$. By hypothesis there is $R\in\sigma$ such that $(u,v,\dots,v)\in R^K$ and $(u',v',\dots,v')\notin R^{K'}$. If $R\in\Prop$ ("ie", $\arity(R)=1$) then we define $\pi=R?\circ a_1\circ\dots\circ a_\ell$ where $a_i$ are of the form $b_i$ or $\bar b_i$, for $b_i\in\Prog$ and correspond to the labels of a path of length $\ell$ from $u$ to $v$ in $K$. If $R\in\Prog$ ("ie", $\arity(R)=2$) then we define $\pi=R$. If $\arity(R)>2$ then we define $\pi=\{R(x,y,\dots,y)\}[x,y]$. 

Suppose that Spoiler has a strategy to win in a bounded number of rounds in $\kSimGame[k+1]$ starting from $(s,(u,v),(u',v'))$. 
    Consider the winning strategy of Spoiler, described by a "tree" of finite height whose vertices are labeled with positions from $\kSimGame[k+1]$. In particular: (i) the root must be labeled with $(s,(u,v),(u',v'))$, (ii) any vertex labeled $(d_i,\bar u, \bar v)$ has a child labeled $(s,\bar u, \bar v')$ for every possible move of $\kSimGame[k+1]$, (iii) any vertex labeled $(s,\bar u, \bar v)$ has exactly one child 
    and it is of the form
    $(d_i,\bar u', \bar v)$,
    whose label is consistent with a move of $\kSimGame[k+1]$,
    (iv) for every leaf (labeled $(d_i,\bar u, \bar v)$) there is no possible move from that position in $\kSimGame[k+1]$.
    Observe that, since $K'$ is of "finite degree", there is a finite number of moves departing from any position of $\kSimGame[k+1]$ owned by Duplicator, and thus the branching of the "tree" is finite. Therefore, the strategy "tree" is finite.
\color{black}
    For such a winning strategy for Spoiler, consider the "tree" $T$ and labeling $\lambda \colon V(T) \to \dom{K}^{k+1}$ resulting from: 
     (1) removing all vertices $y$ labeled by Spoiler positions except the root (\ie, the initial configuration, owned by Spoiler), and adding an edge between the parent of $y$ and the (sole) child of $y$, 
    and (2) projecting the labeling of each vertex onto its $\dom{K}^{k+1}$ component (that is, $\lambda(x)=\bar u$ if the label of $x$ in the strategy "tree" was $(d_i,\bar u, \bar v)$). 
    From this "tree" one can construct a "conjunctive program" $C[x_s,x_t]$ of "tree-width" $k$ containing all the homomorphism types of the tuples.
    Let $Q$ be the set of all pairs $(y,i)$ where $y$ is a vertex from $T$ and $1 \leq i \leq k+1$, and let $\approx$ be the finest equivalence relation on $Q$ such that: (a) $(y,i) \approx (y,i')$ if $\lambda(y)[i] = \lambda(y)[i']$, and (b) $(y,i) \approx (y',i)$ if $y'$ is a child of $y$ and $\lambda(y)[i] = \lambda(y')[i]$. The "conjunctive program" $C$ will use a variable for each equivalence class of $\approx$, denoted by $[y,i]_\approx$, and it contains: 
    (i)  a "p-atom" \color{black} $a([y,i]_\approx,[y',j]_\approx)$ if $y=y'$ and $(\lambda(y)[i],\lambda(y)[j]) \in \dbracket{a}_K$, 
    (ii) a "p-atom" \color{black} $p?([y,i]_\approx,[y,i]_\approx)$ if $\lambda(y)[i] \in \dbracket{p}_K$, 
    and (iii) an "r-atom" $R([y_1,i_1]_\approx,\dots,[y_n,i_n]_\approx)$ if $y_1=\dots=y_n$ and $(\lambda(y_1)[i_1],\dots,\lambda(y_1)[i_n]) \in R^K$. 
    The variable $x_s$ \resp{$x_t$} is the class $[y,i]_\approx$ where $y$ is the root of $T$ and $\lambda(y)[i] =u$ \resp{$\lambda(y)[i] = v$}. 
    See Figure \ref{fig:1-to-2} for an illustration.
    To ensure that $x_s,x_t\in\vars(C)$ and that $\uGraphC{C}$ is connected, we add the "p-atom" $a_1\circ\dots\circ a_\ell(x_s,x_t)$ where $a_i$ are of the form $b_i$ or $\bar b_i$, for $b_i\in\Prog$ and correspond to the labels of a path of length $\ell$ from $u$ to $v$ in $K$. 
    \begin{claim}\AP
        $C[x_s,x_t] \in \CPDLg{\Tw}$.
    \end{claim}
    \begin{proof}
        Due to \Cref{rem:connected} and the fact that $u$ and $v$ are in the same component of $K$, $\uGraph{C}$ is connected. On the other hand, $\uGraph{C[x_s,x_t]}$ has "tree-width" $k$ because $(T,\lambda)$ induces a "tree decomposition" $(T,\lambda')$ of "width" $k$ of $\uGraph{C[x_s,x_t]}$, where $\lambda'(y) = \set{[y,i]_\approx : 1 \leq i \leq k+1}$ for every $y \in V(T)$.
    \end{proof}
    \begin{claim}\AP
        $K,u,v \models C[x_s,x_t]$ and $K',u',v' \not\models C[x_s,x_t]$.
    \end{claim}
    \begin{proof}
    The fact that $K,u,v \models C[x_s,x_t]$ is a consequence of  $(y,i) \approx (y',i')$ implying $\lambda(y)[i] = \lambda(y')[i']$ by definition. Indeed, $f=\set{[y,i]_\approx \mapsto \lambda(y)[i] : y \in V(T), 1 \leq i \leq k+1}$ is a "$C$-satisfying assignment".

    We show $K',u',v' \not\models C[x_s,x_t]$ by  contradiction. If $K',u',v' \models C[x_s,x_t]$, consider the corresponding "$C$-satisfying assignment" $f'\colon \vars(C) \to \dom{K'}$. Using $f'$ we can select a branch of Spoiler's strategy as follows:

    for every vertex $y$ labeled $(d_i, \bar u, \bar v)$ of the strategy tree, pick the child $y'$ of $y$ such that $y'$ is labeled $(s, \bar u, \bar v')$ where $\bar v' = \bar v[i \mapcoord f'([y',i]_\approx)]$ and for every vertex $y$ labeled $(s, \bar u, \bar v)$ we pick the only child of $y$.

    It remains to show that there is a move from $(d_i, \bar u, \bar v)$ to $(s, \bar u, \bar v')$ in $\kSimGame[k+1]$.
    By definition of $C$, it contains all possible "atoms" derived from the substructure of $K$ induced by $\bar u$ using the variables $[y',1]_\approx, \dotsc, [y',k+1]_\approx$. Thus, $\set{\bar u[t] \mapsto f'([y',t]_\approx)}_t = \set{\bar u[t] \mapsto \bar v'[t]}_t$ witnesses a "partial homomorphism" from $K$ to $K'$, and hence $(d_i, \bar u, \bar v) \rightarrow (s, \bar u, \bar v')$ is a valid move of $\kSimGame[k+1]$. 
    This means that Spoiler's strategy tree cannot be finite since by descending the tree in this way we end up in a leaf $(d_i, \bar u, \bar v)$ of the strategy tree which has an outgoing move in $\kSimGame[k+1]$, contradicting the fact that the strategy is winning for Spoiler (more precisely point (iv)).
    \end{proof}
    \noindent
    Observe that, strictly speaking, $C$ may have an \emph{infinite} number of "atoms" (over the \emph{finite} set $\vars(C)$ of "variables"). 
        This is because for any tuple $\bar u$ of "worlds" of a "finite degree" $\Rels$-structure $K$, there may be infinitely many symbols $R\in\Rels$ such that $\bar u\in R^K$.
    However, from the fact that $K',u',v' \not\models C[x_s,x_t]$ and that $K'$ has "finite degree", one can extract a finite set of "atoms" $C' \subseteq C$ such that $K',u',v' \not\models C'[x_s,x_t]$ as we describe next.
    
    Consider the set $F$ of mappings $f \colon \vars(C) \to \dom{K'}$ such that 
        (i) $f(x_s,x_t) = (u',v')$, and 
        (ii) $f(x)$ is at "distance" at most $|\vars(C)|$ from $u'$ or $v'$.
    Notice that, since $K'$ has "finite degree", $F$ is finite.

    For every $f \in F$, choose an "atom"
    $a^f(x^f_1,\dots,x^f_n)$ of $C$ which does not map to $K'$, \ie, such that $(f(x^f_1),\dots,f(x^f_n)) \not\in \dbracket{a^f}_{K'}$.\footnote{
        Observe that if $a^f$ is a "p-atom" then $n=2$. If $a^f=R\in\sigma$ and $\arity(R)>2$, then by $\dbracket{a^f}_{K'}$ we mean $R^{K'}$} Let $C'$ be the set of all such "atoms", plus some (finite) set of "atoms" from $C$ to assure (a) $\vars(C') = \vars(C)$ and (b) that $\uGraphC{C'}$ is connected. Notice that $C'$ is now finite.
    We claim that $K',u',v' \not\models C'[x_s,x_t]$. Indeed, if there was a "$C'$-satisfying assignment" $f \colon \vars(C') \to \dom{K'}$, then $f \in F$ and $C'$ should include the "atom" $a^f(x^f_1,\dots,y^f_n)$ which is not mapped by $f$ to $K'$. Hence, there is no "$C'$-satisfying assignment" and $K',u',v' \not\models C'[x_s,x_t]$.
\begin{figure}\centering
\ifarxiv
    \includegraphics[scale=0.25]{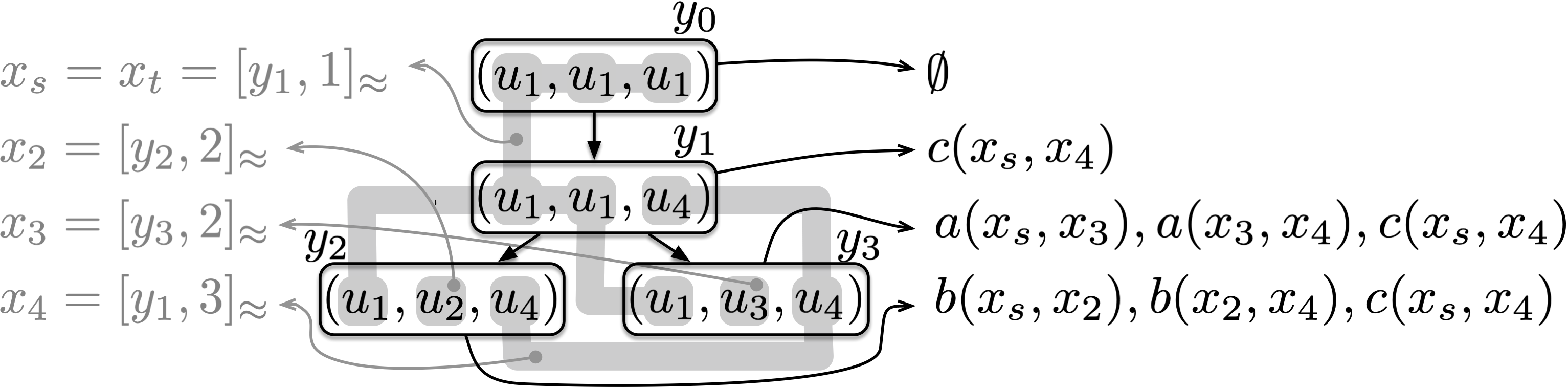}
\else
    \includegraphics[width=.45\textwidth]{img/1-to-2.png}
\fi
\caption{An example of the tree $T$ with nodes $y_0,y_1,y_2,y_3$ in the proof of Theorem \ref{thm:sim-char}\proofcase{1-to-2} built from Spoiler's strategy of Figure \ref{fig:game} (right). 
It holds that $K,u_1,u_1\models\pi$ and $K',v_1,v_1\not\models\pi$ for $\pi=C[x_s,x_t]$. The grey area represents the partition induced by $\approx$.}
\label{fig:1-to-2}
\end{figure}

\smallskip
The case of $K,u \notpebblesim{k+1} K',u'$ is analogous. In this case, $C[x_s,x_t]$ is built just as before, but we set $x_t$ to be equal to $x_s$. The final "formula" is then $\tup{C[x_s,x_t]}$.

    Let us now show it for $\UCPDLg{\Tw[k]}$ when $K$ and $K'$ are finite. As for the case for $\CPDLg{\Tw[k]}$, consider the winning strategy of Spoiler starting in position $(s,(u,v),(u',v'))$, described by a "tree" $W$ of finite height whose vertices are this time labeled with positions from $\kSimGameE[k+1]$ instead of $\kSimGame[k+1]$ satisfying items (i)-(iv) as above but where index $i$ ranges over $\{0,\dots,k\}$ instead of $\{1,\dots,k\}$. Notice that we now allow for moves of the form $(s,\bar w_1, \bar w_1')$ to $(d_0,\bar w_2,\bar w_1')$, where $\bar w_2=(w_2,\dots,w_2)$ for $w_2\in\dom{K}$. 
    \AP
    We call an edge ``""jumping""'' if it connects such parent $(s,\bar w_1, \bar w_1')$ to  child $(d_0,\bar w_2,\bar w_1')$ (see Figure \ref{fig:caseU}).
    In principle, a node of $W$ labeled $(d_0,\bar w_2,\bar w_1')$ may have a child labeled $(s,\bar w_2,\bar w_2')$ for $\bar w_2'=(w_2',\dots,w_2')$ for any $w_2'\in\dom{K'}$.  
    There are finitely many such children because of the hypothesis of $K$ being finite; hence Spoiler's strategy tree $W$ is finite. 

    \begin{figure}\centering
        \includegraphics[scale=0.25]{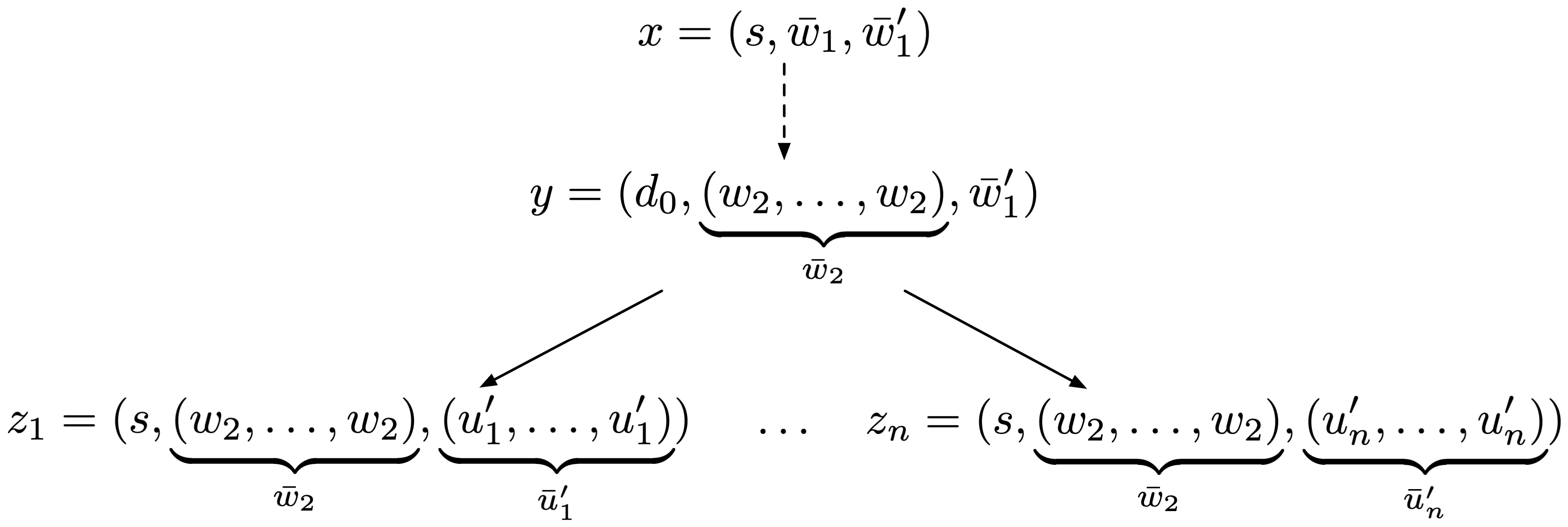}
    \caption{
    The dotted arrow represents a "jumping" edge in Spoiler's winning strategy}
    \label{fig:caseU}
    \end{figure}

    \noindent
    If $W$ has no "jumping" edges we proceed as in the case for $\CPDLg{\Tw}$, namely we construct a "conjunctive program" $C[x_s,x_t] \in \CPDLg{\Tw}$ such that $K,u,v \models C[x_s,x_t]$ and $K',u',v' \not\models C[x_s,x_t]$. 
    Suppose that there is a "jumping" edge in $W$ connecting $x$ to its only child $y$. Suppose that $x$ is labeled $(s,\bar w_1,\bar w_1')$, and that $y$ is labeled  $(d_0,\bar w_2,\bar w_1')$, where $\bar w_2=(w_2,\dots,w_2)$. Furthermore suppose that there are no "jumping" edges below $y$. Suppose that $z_1,\dots,z_n$ are the children of $y$ in $W$ for some $n\geq 0$ and that $z_i$ is labeled $(s,\bar w_2,\bar u'_i)$, for $\bar u'_i=(u'_i,\dots,u'_i)$, and $u'_i\in\dom{K'}$.  See Figure \ref{fig:caseU}.
    Since each subtree of $W$ with root $z_i$ encodes a winning strategy of Spoiler free of "jumping" edges witnessing $K,w_2,w_2 \not\pebblesim{k+1} K',u'_i,u'_i$, we can apply the above result for $\CPDLg{\Tw}$ to obtain a "conjunctive program" $C_i[z,z]\in \CPDLg{\Tw}$ such that $K,w_2,w_2\models C_i[z,z]$ and $K',u'_i,u'_i\not\models C_i[z,z]$ for all $i=1,\dots,n$. Let $C=\{C_i[z,z](z,z)\mid i=1,\dots,n\}$. We have $K,w_2,w_2\models C[z,z]$ and $K',u'_i,u'_i\not\models C[z,z]$ for all $i=1,\dots,n$. Let $U'$ be the (possibly empty) set $\{u'_i\mid i=1,\dots,n\}$ and consider the set $V'=\dom{K'}\setminus U'$, which is finite.
    For any $v'\in V'$ there is an "atom" $\tau$ such that $K,w_2,\dots,w_2\models \tau(z,\dots,z)$ and $K',v',\dots,v'\not\models \tau(z,\dots,z)$. Let $D$ be the set of all such atoms $\tau(z,\dots,z)$. Then $K,w_2,w_2\models D[z,z]$ and $K',v',v'\not\models D[z,z]$ for all $v'\in V'$. Finally for the "conjunctive program" $\pi=\{ C[z,z](z,z), D[z,z](z,z)\}[z,z]$ we have $K,w_2,w_2\models\pi$
    and $K',w',w'\not\models \pi$ for all $w'\in \dom{K'}$. 
    This implies that $\pi'=\Univ\circ\pi\circ\Univ\in\UCPDLg{\Tw}$ is true at any pair of nodes of $K$ and false at any pair of nodes of $K'$. 
    In particular $K,u,v\models\pi'$ and $K',u',v'\not\models\pi'$.

\color{black}

%% file: simchar-2-to-1.tex
\proofcase{2-to-1}
We now show the 2-to-1 implication 
for the case $\CPDLg{\Tw[k]}$
by structural induction on the expression.
For any $(\bar u,\bar v)\in\set{(v,v'),((u,v),(u',v'))}$, assuming $K, \bar v \pebblesim {k+1} K', \bar v'$ %
and $K,\bar v \models e$ 
    for a "positive" $\CPDLg{\Tw[k]}$ expression $e$, we show $K',\bar v' \models e$.
    Observe that if $\bar u=(u,v)$ and $u$ and $v$ are in different connected components of $K$ the result is trivial since there is no "positive" $\CPDLg{\Tw[k]}$-"program" true in $K,u,v$. So we assume that all "worlds" of $\bar u$ (and so of $\bar v$) are in the same connected component of $K$.

    \proofsubcase{$e=p$ for $p\in\Prop$} If $K,v \models p$ then it follows that since there is a "partial homomorphism" mapping $v$ to $v'$ we have $K',v' \models p$.

    \proofsubcase{$e=a$ for $a\in\Prog$} If $K,u,v \models a$ then it follows that since there is a "partial homomorphism" mapping $u\mapsto u'$ and $v\mapsto v'$ we have $K',u',v' \models a$.
        
    \proofsubcase{$e=\bar a$ for $a\in\Prog$} Analogous to the case $e=a$.

    \proofsubcase{$e=\phi_1\land\phi_2$} If $K,v \models \phi_1\land\phi_2$ then for $i=1,2$ we have $K,v \models \phi_i$ and by inductive hypothesis we have $K',v' \models \phi_i$ for $i=1,2$, which implies $K',v' \models \phi_1\land\phi_2$.

    \proofsubcase{$e=\pi_1\cup\pi_2$} Analogous to the previous case.

    \proofsubcase{$e=\phi?$} If $K,u,v \models \phi?$ then $u=v$ and $K,u \models \phi$. By inductive hypothesis $K',u' \models \phi$ and so $K',u',v' \models \phi$.

    \proofsubcase{$e=\tup{\pi}$} 
    If $K,v \models \tup{\pi}$ then there is some $\tilde v$ in $K$ such that $(v,\tilde v) \in \dbracket{\pi}_K$. Since we have: (1) $v$ and $\tilde v$ are part of the same connected component in $K$, (2) $(K, v) \pebblesim{k+1} (K', v')$, 
    and (3) $k+1 \geq 3$, it follows that there is some $\tilde v'$ such that  
    $(K, v,\tilde v) \pebblesim{k+1} (K', v', \tilde v')$.
    This can be obtained by ``navigating'' from $v$ to $\tilde v$ in $K$ using the definition of $\pebblesim{k+1}$.
    We can now apply the inductive hypothesis obtaining that $K',v',\tilde v' \models \pi$, and thus that $K',v' \models \tup{\pi}$.

    \proofsubcase{$e = C[x_s,x_t]$} 
    We will actually work with the more general version of "programs" $C[\bar z]$ where $\bar z$ may have any number $1 \leq t \leq k+1$ of variables from $\vars(C)$, and the "underlying graph" $\uGraph{C[\bar z]}$ of $C[\bar z]$ is $\uGraph{C}$ plus the edges $\set{y,y'}$ for every pair of distinct variables $y,y'$ from $\bar z$. Hence,  $\uGraph{C[\bar z]} \in \Tw[k]$ means that all variables $\bar z$ appear in a bag of its "tree decomposition".
     The semantics $\dbracket{C[\bar z]}_{\tilde K}$ on a "Kripke structure" $\tilde K$ is as expected: the set of $t$-tuples $\bar u$ of "worlds" from $\tilde K$ such that $f(\bar z) = \bar u$ for some "$C$-satisfying assignment" $f$.

    Suppose we have $e = \hat C[x_s,x_t]$. We show the following:
    \begin{claim}\AP
        For every $C \subseteq \hat C$ and every tuple $x_1, \dotsc, x_n$ of variables from $\vars(C)$ with $n \leq k+1$ such that $\uGraph{\hat C[x_1, \dotsc, x_n]} \in \Tw[k]$, for all $\bar v\in\dom{K}^n$ and $\bar v'\in\dom{K'}^n$: if $(K, \bar v) \pebblesim{k+1} (K', \bar v')$ and $\bar v  \in \dbracket{C[x_1, \dotsc, x_n]}_K$, then $\bar v'  \in \dbracket{C[x_1, \dotsc, x_n]}_{K'}$.
    \end{claim}
    \noindent
    First note that the claim above implies $K',u',v' \models \hat C[x_s,x_t]$. We now prove the claim by induction on the size\sidediego{cardinality?} $|C|$ of $C$. 
    Let $(T,\bagmap)$ be a "tree decomposition" of $\uGraph{\pi}$ of "width" $k$ for $\pi = C[x_1, \dotsc, x_n]$. Without loss of generality, suppose that the root bag $b$ of $T$ is such that $\bagmap(b) = \set{x_1, \dotsc, x_n}$. 
    
    Suppose $(K, \bar v) \pebblesim{k+1} (K', \bar v')$ and $\bar v  \in \dbracket{\pi}_K$. 
    We will show $\bar v'  \in \dbracket{\pi}_{K'}$.
    Let $h$ be a "$C$-satisfying assignment" on $K$ such that $\bar v = (h(x_1), \dotsc, h(x_n))$. We need to exhibit a "$C$-satisfying assignment" $h'$ on $K'$ such that $\bar v' = (h'(x_1), \dotsc, h'(x_n))$.

    Suppose that the root $b$ of $T$ has $\ell$ children $b_1, \dotsc, b_\ell$, and let $T_1, \dotsc, T_\ell$ be the corresponding "subtrees@trees". Let $\bar y_i$ be any vector of the (pairwise distinct) variables from $\bagmap(b_i) \cap \bagmap(b)$, that is, the variables which are both in the root "bag" of $T$ and in the "bag" of its $i$-th child. We can assume that $\bar y_i$ is of dimension $\leq k$ (otherwise parent and child would have exactly the same "bag").
    Let $\bar v_0 = \bar v$, and $\bar v_i = h(\bar y_i)$ for every $1 \leq i \leq \ell$.
    Let $C_0$ be the set of all "atoms" of $C$ in the root bag (\ie, such that 
    variables are in $\bagmap(b)$). 
    Let $C_i \subseteq (C \setminus C_0)$ be the set of all "atoms" of $C \setminus C_0$ contained in "bags" of $T_i$ (\ie, such that 
    all
    variables are in some "bag" of $T_i$). 
    
    It then follows that $K,\bar v \models C[x_1, \dotsc, x_n]$ if{f} $K,\bar v_i \models C_i[\bar y_i]$ for every $0 \leq i \leq \ell$.
    Now we have two possibilities, either $C_0 = C$ or $C_0 \subsetneq C$.

        Case $C_0 = C$:~~ If $C_0 = C$,  
        this means that all the "atoms" of $C$ use variables from $\set{x_1, \dotsc, x_n}$.
        Consider the function $h' \colon \vars(C)\to K'$ such that $h'(x_i)=\bar v'[j]$ if $h(x_i)=\bar v[j]$. Let us show that $h'$ is a "$C$-satisfying assignment" on $K'$.
        For every 
        "p-atom" 
        $\pi'(x_i,x_j) \in C=C_0$, we have that $(h(x_i),h(x_j)) = (\bar v[i],\bar v[j])$ 
        is in $\dbracket{\pi'}_K$ by hypothesis on $h$.
        On the other hand we also have $K,\bar v[i], \bar v[j] \pebblesim{k+1} K',\bar v'[i],\bar v'[j]$. Hence, $(\bar v'[i], \bar v'[j]) \in \dbracket{\pi'}_{K'}$ by inductive hypothesis on $\pi'$. 
        The reasoning for "r-atoms" is analogous.
        This means that $h'$ is a "$C$-satisfying assignment" on $K'$.

        Case $C_0 \subsetneq C$:~~
        If $C_0 \subsetneq C$, then $|C_i| < |C|$ for every $i$.\footnote{There is one degenerate case in which $C_0=\emptyset$, but since the "graph" $\uGraphC{C}$ of $C$ is connected, we can also assume that the "tree decomposition" is such that there is at least one "atom" in $C_0$.} We can then apply the inductive hypothesis on each $C_i$ since we have $K,\bar v_i \models C_i[\bar y_i]$ and $K,\bar v_i \pebblesim{k+1} K',\bar v'_i$, obtaining $K',\bar v'_i \models C_i[\bar y_i]$ for some "$C$-satisfying assignment" $h'_i$ on $K'$ sending $\bar y_i$ to $\bar v'_i$. This means that $h' \eqdef h'_0 \cup \dotsb \cup h'_\ell$ is a "$C$-satisfying assignment" on $K'$ and thus that $K',h'(x_1), \dotsc, h'(x_n) \models C[x_1, \dotsc, x_n]$.

    \proofsubcase{$e = \pi^*$} Suppose $K, u,v \models \pi^*$ 
    and $K,u,v \pebblesim{k+1} K',u',v'$. We proceed by induction on the number of iterations. The base case $0$ corresponds to $K, u,v \models \pi^0$ which happens if and only if $u=v$, in which case $u'=v'$ by definition of $\pebblesim{k+1}$ and thus $K, u',v' \models \pi^0$.
    For the inductive case, $K, u,v \models \pi^n$. Note that $\pi^n$ is equivalent to $\tilde \pi = C[x_0,x_n]$ for $C = \set{\pi(x_i,x_{i+1}) \mid 0 \leq i < n}$. Since the "subexpressions" of $\tilde \pi$ are the same as those of $\pi^*$ (namely, $\pi$ and the "subexpressions" therein), we can replace $\pi^n$ with $\tilde \pi$ and invoke the previous case, obtaining  $K',u',v' \models \pi^n$.

    \proofsubcase{$\pi \circ \pi'$} This case is similar to the case $\pi^*$.

The implication 2-to-1 for the case $\UCPDLg{\Tw[k]}$ is analogous  observing that the case
 
    \proofsubcase{$e=\Univ$} is trivial, since $\Univ$ is true at any pair of nodes.

This concludes the proof.

%% file: bisimulation.tex
\subsection{Bisimulation Relation}

We define here the notion of bisimulation on $(K,K')$, as before, using a two-player game
$\intro*\kBisimGame$, as was done for the simulation game $\kSimGame$. Assume, without loss of generality, that $K,K'$ have disjoint sets of "worlds".  This time the arena of $\kBisimGame$ has a set of positions $S \cup D $, where
\begin{align*}
    S &= \set{s} \times (\kHoms(K,K')\dcup \kHoms(K',K)),\\
    D &= \set{d_1,\dotsc, d_k} {\times}\big((\dom{K}^k {\times} \dom{K'}^k) \dcup (\dom{K'}^k {\times} \dom{K}^k)\big),
\end{align*}
Spoiler owns all positions from $S$, and Duplicator all positions from $D$. 
The set of moves of $\kBisimGame$ is the smallest set satisfying the following:
\begin{enumerate}
    \item There is a move from $(s,\bar u, \bar v)$ to $(d_i,\bar u',\bar v)$ 
    if $\bar u' = \bar u[i \mapcoord w]$, where $w$ is a "world" at "distance" $\leq 1$ from $\bar u[j]$, for some $1 \leq j \leq k$ with $i\neq j$; and 
    \item \label{it:kbisim:neg} There is a move from $(s,\bar u, \bar v)$ to $(s,\bar v,\bar u)$ if $\bar u[1] = \dotsb = \bar u[k]$ (and hence $\bar v[1] = \dotsb =\bar v[k]$).
    \item There is a move from $(d_i,\bar u',\bar v)$ to $(s,\bar u',\bar v')$ if $\bar v' = \bar v[i \mapcoord w]$, where $w$ is a "world" at "distance" $\leq 1$ from $\bar v[j]$, for some $1 \leq j \leq k$ with $i\neq j$.
\end{enumerate}
Again, the winning condition for Duplicator is just any infinite play.

The intuition of the game $\kBisimGame$ is like the one of $\kSimGame$, but now Spoiler may decide to switch models at any round, provided all her pebbles (and hence all Duplicator's pebbles) are stacked over the same element. Spoiler's pebbles become Duplicator's and vice-versa. The players keep playing as usual but over the new models until Spoiler decides to switch again or the game terminates.

We also define the game $\intro*\kBisimGameE$, which is analogous to $\kBisimGame$, but with the following modifications:
The arena of the game has now a set of positions $S \cup D$, %
where $S$ is as before and
\begin{align*}
    D &= \set{d_0,\dotsc, d_k} {\times}\big((\dom{K}^k {\times} \dom{K'}^k) \dcup (\dom{K'}^k {\times} \dom{K}^k)\big),
\end{align*}
Besides rules (1), (2) and (3) above, two more rules are added:
\begin{enumerate}
    \item[(4)] There is a move from $(s,\bar u, \bar v)$ to $(d_0,\bar u',\bar v)$ 
    if $\bar u'[1] = \dots = \bar u'[k]$; and 
    \item[(5)] There is a move from $(d_0,\bar u', \bar v)$ to $(s,\bar u',\bar v')$ (observe that since $\bar u'[1] = \dots = \bar u'[k]$, then $\bar v'[1] = \dots = \bar v'[k]$)
\end{enumerate}
As before, the intuition of these new rules is that at any moment, Spoiler can stack all her pebbles in the same node $u'$ of $K$ [resp.\ of $K'$], and Duplicator has to respond in the same way, stacking all his pebbles in the same node $v'$ of $K'$ [resp.\ of $K$].
\color{black}

We now define two related notions.
\AP
\newcommand{\quasibisim}{$k$-half-bisimulation\xspace}
\newcommand{\quasibisimE}{$k$-half-$\Univ$-bisimulation\xspace}
The ""\quasibisim"" notion $\pebblequasibisim k$ is obtained as follows.
Let $K, K'$ be "Kripke structures" with disjoint sets of "worlds" (this is without loss of generality), and let $\bar v\in \dom{K}^k$ and $\bar v'\in \dom{K'}^k$. 

We  say that there is a "\quasibisim" from $K,\bar v$ to $K',\bar v'$, denoted $K,\bar v \intro*\pebblequasibisim k K',\bar v'$, 
$\phantomintro{\notpebblequasibisim}$%
if the following holds: in case all "worlds" in $\bar v$ are in the same connected component of $K$, then 1) $(s,\bar v, \bar v')$ is a valid position of $\kBisimGame$ on $(K,K')$  (\ie, they induce a "partial homomorphism"), 
and 2) Duplicator has a winning strategy in $\kBisimGame$ from $(s,\bar v, \bar v')$; in case $\bar v$ has "worlds" in different connected components of $K$ then $K,\bar v \pebblesim k K',\bar v'$ holds automatically.

\AP
We say that there is a ""$k$-bisimulation"" between $K,\bar v$ and $K',\bar v'$, and we note it $K,\bar v \intro*\pebblebisim k K', \bar v'$ if $K,\bar v \pebblequasibisim k K',\bar v'$ and $K',\bar v' \pebblequasibisim k K,\bar v$.

We  say that there is a ""\quasibisimE"" from $K,\bar v$ to $K',\bar v'$, denoted $K,\bar v \intro*\pebblequasibisimE k K',\bar v'$, 
if
1) $(s,\bar v, \bar v')$ is a valid position of $\kBisimGameE$ on $(K,K')$  (\ie, they induce "partial homomorphisms"), 
and 2) Duplicator has a winning strategy in $\kBisimGameE$ from $(s,\bar v, \bar v')$.

We say that there is a ""$k$-$\Univ$-bisimulation"" between $K,\bar v$ and $K',\bar v'$, and we denote it $K,\bar v \intro*\pebblebisimE k K', \bar v'$ if $K,\bar v \pebblequasibisimE k K',\bar v'$ and $K',\bar v' \pebblequasibisimE k K,\bar v$.

\color{black}

\begin{remark}\label{rem:bisim-imples-sim}
    If $K_1,\bar v_1 \pebblequasibisim k K_2,\bar v_2$ then $K_1,\bar v_1 \pebblesim k K_2,\bar v_2$
; and if $K_1,\bar v_1 \pebblequasibisimE k K_2,\bar v_2$ then $K_1,\bar v_1 \pebblesimE k K_2,\bar v_2$.
\color{black}
\end{remark}

\noindent
A similar characterization result as \Cref{thm:sim-char} can be shown for "\quasibisim" and "\quasibisimE", and its proof follows the same general blueprint:

\begin{lemma}\AP\label{lem:quasibisim-char}
    Let $k\geq 2$.
       Given "$\Rels$-structures"
    $K,K'$ of "finite degree" [resp.\ finite] and "worlds" $u,v\in\dom{K}$ and $u',v'\in\dom{K'}$, the following are equivalent
    \begin{enumerate}
        \item  for every $\CPDLg{\Tw[k]}$ [resp.\ $\UCPDLg{\Tw[k]}$] "formula" $\phi$, we have $K,v \models \phi$ implies $K',v' \models \phi$; and
        \item $K,v \pebblequasibisim{k+1} K',v'$ [resp.\ $K,v \pebblequasibisimE{k+1} K',v'$];
    \end{enumerate}
    and the following are equivalent
    \begin{enumerate}
        \item for every $\CPDLg{\Tw[k]}$ [resp.\ $\UCPDLg{\Tw[k]}$] "program" $\pi$, we have $K,u,v \models \pi$ implies $K',u',v' \models \pi$; and
        \item $K,u,v \pebblequasibisim{k+1} K',u',v'$ [resp.\ $K,u,v \pebblequasibisimE{k+1} K',u',v'$].
    \end{enumerate}
Furthermore, the hypothesis of "finite degree" [resp.\ finite] is only needed for the $1$-to-$2$ implications.
\end{lemma}
\color{black}
\begin{proof}
    We show it for $\CPDLg{\Tw[k]}$. The proof for $\UCPDLg{\Tw[k]}$ follows the same adaptation as in the proof of \Cref{thm:sim-char}. \color{black} 

    \proofcase{1-to-2}
    We proceed as in the proof of \Cref{thm:sim-char}, by contrapositive. Assume $K,u,v \notpebblequasibisim{k+1} K',u',v'$. 
    The case where $((u,v), (u',v')) \not\in \kHoms(K,K')$ is  trivial, so let us assume the contrary.
    This means that Spoiler has a strategy to win in a bounded number of rounds in $\kBisimGame[k+1]$ starting from $(s,((u,v),(u',v')))$.

    The winning strategy of Spoiler is again a finite "tree", whose vertices are labeled with positions from $\kBisimGame[k+1]$. In particular: (i) the root must be labeled with $(s,((u,v),(u',v')))$, (ii) any vertex labeled $(d_i,\bar u, \bar v)$ has a child labeled $(s,\bar u', \bar v)$ for every possible move of $\kBisimGame[k+1]$, (iii) any vertex labeled $(s,\bar u, \bar v)$ has exactly one child, whose label is consistent with a move of $\kBisimGame[k+1]$ (remember that Spoiler has two types of moves this time), (iv) for every leaf (labeled $(d_i,\bar u, \bar v)$) there is no possible move from that position in $\kBisimGame[k+1]$.
    Since both $K,K'$ are of "finite degree", there is a finite number of moves departing from any position of $\kBisimGame[k+1]$, and thus the branching of the "tree" is finite. Therefore, the strategy tree is finite.
    Without loss of generality we may assume that there is no path of at most two nodes in the strategy tree labeled with $(s,\cdot,\cdot)$, namely, if a node labeled $(s,\bar u,\bar v)$ is the parent of a node labeled $(s,\bar v,\bar u)$ then this last node has as (only) child a node labeled $(d_i,\bar v,\bar u)$. This is because it is useless for Spoiler to choose to play three or more times in a row the rule (2) of the game.

    From such a winning strategy for Spoiler, consider the "tree" $T$ and labeling $\lambda \colon V(T) \to (\dom{K}^{k+1} \dcup \dom{K'}^{k+1})$ resulting from: 
    (1) replacing all pairs of nodes $x,x'$, such that $x$ and $x'$ are labeled by Spoiler, and $x$ is the parent of $x'$ in the strategy tree by a special new kind of 
\AP
    node called ``""switching node""'' $w$ (by our assumption $x'$ has an only child labeled by Duplicator, and if $x$ has a parent $z$ then $z$ is also labeled by Duplicator) and defining $\lambda(w)=\bar v$ if $x'$ was labeled $(s,\bar v,\bar u)$ (and therefore $x$ was labeled $(s,\bar u,\bar v)$, and furthermore $\bar u = (u, \dotsc, u)$ and $\bar v=(v, \dotsc, v)$),
     (2) removing all remaining vertices $y$ labeled by Spoiler positions except the root (\ie, the initial configuration, either owned by Spoiler or being a "switching node"), and adding an edge between the parent of $y$ and the (sole) child of $y$, 
    and (3) defining $\lambda(x)=\bar u$ if $x$ is labeled by Duplicator (in particular, not a "switching node") and the label of $x$ in the strategy "tree" was $(d_i,\bar u, \bar v)$. See Figure \ref{fig:bisim_proof}.

    Observe that there are no "switching nodes" $u,v$ in $T$ such that $u$ is the parent of $v$, and that if $w$ is a "switching node" then $\lambda(w)=\bar u$, for some $\bar u$ such that $\bar u = (u, \dotsc, u)$. 
    \AP
    An edge of $T$ is called a ``""switching edge""'' if it connects $x$ to $x'$ and $x'$ is a "switching node".
    Then, $T$ can be seen as a ``"tree" of "trees"'', each "tree" being a component of $T$ after removing all "switching edges". Consider all the "trees" $T_1, \dotsc, T_n$ resulting from removing such edges from $T$. The root of each $T_i$ is then a "switching node", except perhaps for the tree $T_i$ containing the root of $T$.
    \begin{figure}
    \includegraphics[scale=0.25]{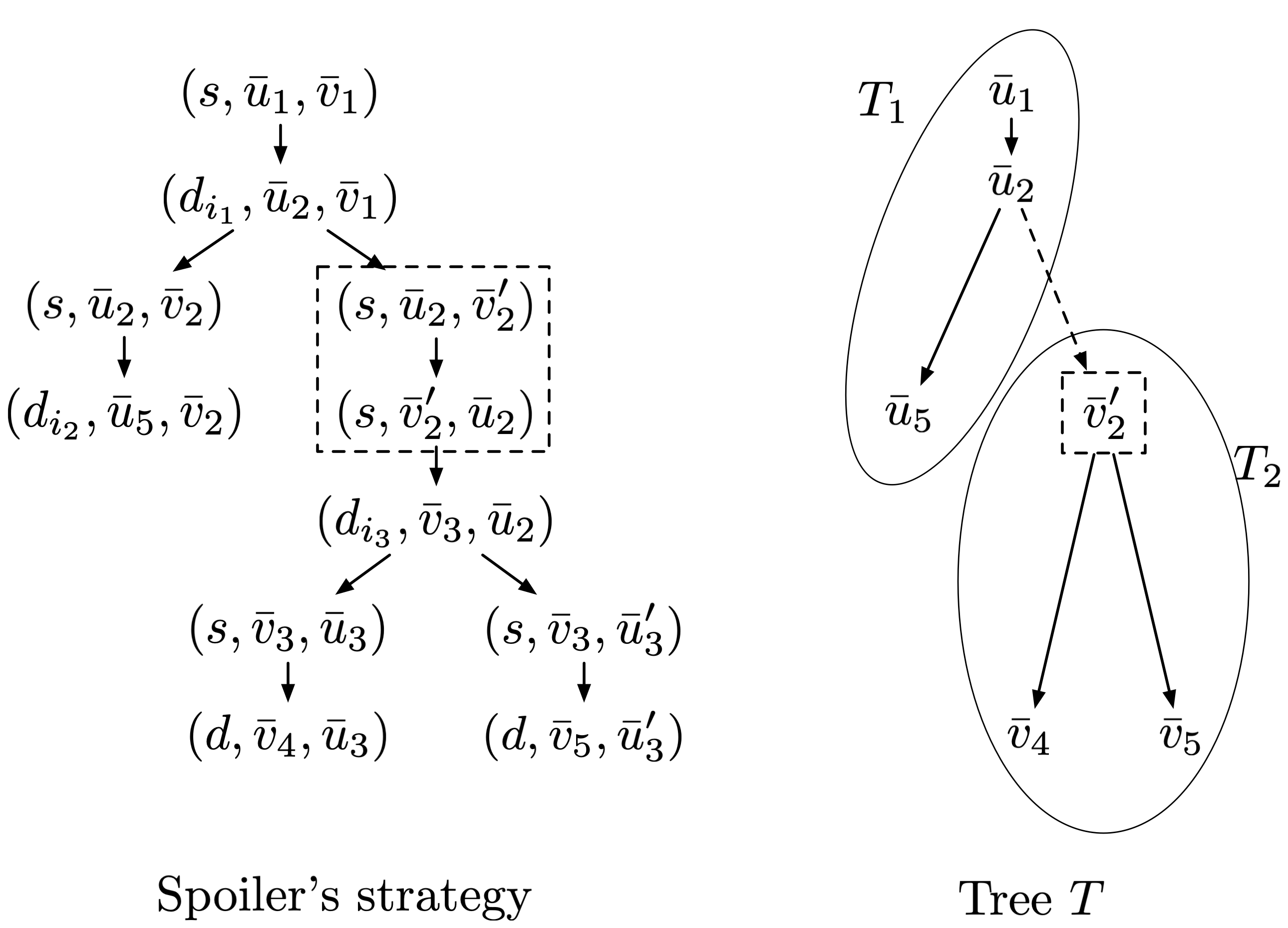}
    \caption{
    An example of a Spoiler's strategy and the tree $T$ constructed from it. Dotted nodes and dotted edges are switching. After removing "switching edges" we obtain trees $T_1$ and $T_2$. Observe that $\bar u_2$ is of the form $(u_2,\dots,u_2)$, and so $\bar v_2$ and $\bar v'_2$ are of the form $(v_2,\dots,v_2)$ and $(v'_2,\dots,v'_2)$ respectively.}
    \label{fig:bisim_proof}
    \end{figure}

    \color{black}
    We now proceed bottom-up: For each $T_i$ such that 
    no leaf has a "switching edge" to another tree,
    we proceed as in the proof of \Cref{thm:sim-char}, and construct a "conjunctive program" $C_i[x,x]$. For the inductive case, suppose we have some tree $T_j$ whose leaves may have switching edges to other trees. Then we construct the "conjunctive program" for $T_j$ as before, but we add some extra "atoms":     we add the "atom" $\lnot\tup{C_i[x,x]}? (z,z)$ for the variable $z$ of $C_j$ of the form $[v,\ell]_\approx$ where $v$ is the leaf of $T_j$ incident to the "switching edge" connecting $T_j$ with $T_i$.
    \footnote{Note that in this case $[v,\ell]_\approx = [v,\ell']_\approx$ for all $\ell,\ell'$, since the definition of $\kBisimGame[k+1]$ restricts these `"switching@switching node"' moves to be applied only to positions of the form $(s,(u,\dotsb, u), (v, \dotsb, v))$.}
    By a similar remark as in the proof of \Cref{thm:sim-char}, we can assume that each $C_i$ has a finite number of "atoms" (this time using "finite degree" of the `current' "structure@@kripke").

    Finally, one can show that the "conjunctive program" $C[x_s,x_t]$ resulting from processing all the $T_i$'s faithfully describes the winning strategy for Spoiler, and thus that it is not possible to have $K',u',v' \models C[x_s,x_t]$.

    \proofcase{2-to-1} Assume $K,v \pebblequasibisim{k+1} K',v'$. The proof goes by structural induction on the "formula" as in \Cref{thm:sim-char}. 
    The base cases are just like in \Cref{thm:sim-char} due to \Cref{rem:bisim-imples-sim}.
    If $\phi = \lnot \psi$, then by \Cref{it:kbisim:neg} we have $K',v' \pebblequasibisim{k+1} K,v$ and thus by inductive hypothesis we have $K',v' \models \psi$ implies $K,v \models \psi$. The latter is the same as $K,v \models \phi$ implies $K',v' \models \phi$.
    The remaining cases are exactly as in the proof of \Cref{thm:sim-char}, with the only difference that we use the inductive hypothesis on "formulas" which may contain negation.
\end{proof} 

\noindent
Observe that, since "formulas" are closed under negation, from the previous characterization \ifarxiv (\Cref{lem:quasibisim-char}) \fi  we obtain that $K,v \pebblequasibisim{k+1} K',v'$ if, and only if, for every $\CPDLg{\Tw[k]}$ "formula" $\phi$, we have $K,v \models \phi$ if{f} $K',v' \models \phi$. In order to have a similar result for "programs" we just need the symmetrical closure of this notion, which is the "$k$-bisimulation":

\begin{thm}
    \label{them:bisim-charac}
    Let $k\geq 2$. Given
    "$\Rels$-structures"
    $K,K'$ of "finite degree" [resp.\ finite] and "worlds" $u,v\in\dom{K}$ and $u',v'\in\dom{K'}$, the following are equivalent
    \begin{enumerate}
        \item for every $\CPDLg{\Tw[k]}$ [resp.\ $\UCPDLg{\Tw[k]}$] "formula" $\phi$, we have $K,v \models \phi$ if{f} $K',v' \models \phi$; and
        \item $K,v \pebblebisim{k+1} K',v'$ [resp.\ $K,v \pebblebisimE{k+1} K',v'$];
    \end{enumerate}
    and the following are equivalent
    \begin{enumerate}
        \item for every $\CPDLg{\Tw[k]}$ [resp.\ $\UCPDLg{\Tw[k]}$] "program" $\pi$, we have $K,u,v \models \pi$ if{f} $K',u',v' \models \pi$; and
        \item $K,u,v \pebblebisim{k+1} K',u',v'$ [resp.\ $K,u,v \pebblebisimE{k+1} K',u',v'$].
    \end{enumerate}
Furthermore, the hypothesis of "finite degree" [resp.\ finite] is only needed for the $1$-to-$2$ implications.
\end{thm}
\color{black}

\begin{claim}\AP\label{cla:equiv_positive_implies_equiv_full}
Let $K$ and $K'$ be $\sigma$-structures. If $K,u \models \phi \Leftrightarrow K',u' \models \phi$ for every "positive" $\UCPDLp$-"formula" $\phi$ and every $u\in\dom{K},u'\in\dom{K'}$, then $K,u \models \psi \Leftrightarrow K',u' \models \psi$ for every $\UCPDLp$-"formula" $\psi$ and every $u\in\dom{K},u'\in\dom{K'}$.
\end{claim}
\begin{proof}[Proof idea.]
We use induction in the number of negations of $\psi$ and the fact that negation can only appear in front of "formulas" (not "programs"). This allows us to apply the inductive hypothesis. The base case corresponds to the hypothesis for the "positive" formulas.
\end{proof}

\noindent
We obtain an analogue of \Cref{lem:hyp_oneway_charact_needed} for \Cref{them:bisim-charac} 

\begin{lemma}\AP\label{lem:hyp_twoway_charact_needed}
\color{santichange}
The $1$-to-$2$ implications of Theorem \ref{them:bisim-charac} fail if the corresponding hypothesis of finite degree [resp.\ finiteness] is dropped.
\color{black}
\end{lemma}

\begin{proof}
We show $\sigma$-structures $Q$ and $N$ and worlds $q\in\dom{Q},n\in\dom{N}$ such that $Q,q\models\phi \Leftrightarrow N,n\models\phi$ for all $\UCPDLp$-formulas $\phi$ (it implies item (1) of Theorem  \ref{them:bisim-charac}) but $Q,q \not\pebblesim{3} N,n$ (which in turn implies the negation of item (2) of Theorem  \ref{them:bisim-charac}).

Let $\sigma=\{a\}$, where $\arity(a)=2$, and consider the $\sigma$-structures $N=(\Nat,<)$, $Q=(\Rat,<)$, where $a$ is interpreted as the order $<$ in both structures. One can verify that the hypotheses of Claim \ref{cla:conditions_for_logical_implication} are satisfied for the "positive" fragment of $\UCPDLp$ in the role of $\+L$, $K,u=Q,q$ and $K',u'=N,n$ for any $q\in\Rat$ and $n\in\Nat$ and also $K,u=N,n$ and $K',u'=Q,q$ for any $q\in\Rat$ and $n\in\Nat$. We conclude then that for any "positive" $\UCPDLp$-formula $\varphi$ and any $q\in\Rat$ and $n\in\Nat$ we have 
$Q,q\models\varphi$ iff $N,n\models\varphi$. 
By Claim 
\ref{cla:equiv_positive_implies_equiv_full}
we conclude that $Q,q \models \psi$ iff $N,n \models \psi$ for every $\UCPDLp$-"formula" $\psi$ and any $q\in\Rat$ and $n\in\Nat$.
However, $Q,0 \not\pebblesim{3} N,0$ (and hence $Q,0 \not\pebblebisim{3} N,0$). The idea for Spoiler's strategy is to enclose Duplicator in smaller and smaller intervals. Since $\Nat$ is not dense, Spoiler will win. In more detail: Spoiler leaves her first pebble in $0$ forever. First she moves her second pebble to $1$. Duplicator answers by moving his second pebble to say $m\in\Nat$. Then Spoiler successively uses her second and third pebbles to visit $1/2,1/3,1/4,\dots$. After at most $m$ rounds, Spoiler will force Duplicator to move his second or third pebble into the point $0$, a position that is winning for Spoiler. 
\end{proof}

\color{black}

%% file: separation.tex
\section{Separation}
\label{sec:separation}

In this section we show that, while $\CPDLg{\Tw[1]}$ and $\CPDLg{\Tw[2]}$ are "equi-expressive", we have a strict hierarchy $\CPDLg{\Tw[2]} \lleqs \CPDLg{\Tw[3]} \lleqs \dotsb$ from $\Tw[2]$. 
We also separate $\CPDLp$ from other studied logics.

\subsection{Separation of the Tree-width Bounded Hierarchy of \texorpdfstring{$\CPDLp$}{CPDL⁺}}

\color{black}
\begin{thm}\label{thm:separation}
    For every $k \geq 2$ we have:
    \begin{center}
     $\CPDLg{\Tw[k]} \lleqs \CPDLg{\Tw[k+1]}$ \quad and \quad $\CPDLg{\Tw[k]} \lleqs[\Kripke] \CPDLg{\Tw[k+1]}$.   
    \end{center}
\end{thm}
\noindent
Essentially we show that the presence of a $(k+1)$-clique can be expressed in $\CPDLg{\Tw[k]}$ but not in $\CPDLg{\Tw[k-1]}$, for every $k\geq 3$. 
\AP
Consider the following ""$(k+1)$-clique formula"" 
\[
    \xi_{k+1} \eqdef \tup{C[x_1,x_{k+1}]} \land \lnot \tup{C'[x_1,y]}
\] 
of $\CPDLg{\Tw[k]}$ where $C = \set{a(x_i,x_j) \mid 1 \leq i < j \leq k+1}$ and $C' = \set{a(x_1,y), a(y,y)}$, for any fixed $a \in \Prog$. We will show that $\xi_{k+1}$ cannot be expressed in $\CPDLg{\Tw[k-1]}$, for every $k \geq 3$.
This inexpressivity result will be a direct consequence of the following ``tree-like model property'', which is of independent interest.
\begin{proposition}\AP\label{prop:treewidth-k-unravelling}
    For every $k \geq 2$, "Kripke structure" $K$, and world $u \in \dom{K}$, there exists a "Kripke structure" $\hat K$ of "tree-width" $\leq k-1$ and "world" $\hat u \in \dom{\hat K}$ such that $K,u \pebblebisim{k} \hat K,\hat u$.
    Further, if $K$ is countable, $\hat K$ has a countable "good" "tree decomposition" of "width" $\leq k-1$.
\end{proposition}
\begin{proof}
    \color{santichange}
    The idea of the proof is to build from $K$ a tree-like companion structure $\hat K$ in which each node of an auxiliary "tree" remembers only a local configuration of at most $k$ "worlds" of $K$. We then replace each "world" of $K$ by copies indexed by nodes of this "tree", and identify copies whenever they represent the same "world" in two consecutive local configurations. In this way, $\hat K$ unfolds $K$ into a structure of "tree-width" at most $k-1$, while still preserving exactly the information that is relevant for the "$k$-pebble bisimulation@$k$-bisimulation" game. Thus, $\hat K$ is tree-like enough to have small "tree-width", but remains "$k$-bisimilar" to the original structure $K$.
    \color{black}

    Let us now give all the technical details. 
    \color{diegochange}
    Consider the set $\+S$ of all nonempty sets of at most $k$ "worlds" from $K$. Consider the infinite "tree" $T$ having $V(T) = \set{\set{u} \cdot \bar x \mid \bar x \in \+S^*}$ 
    \color{black}
    (\ie, vertices are finite sequences on $\+S$ starting with $\set u$) and having an edge $\set{w,v}$ if $w$ is a prefix of $v$ of length $|v|-1$. Consider the root vertex $r$ of $T$ to be $r = \set{u}$.
    \color{santichange}
    Let $\lambda \colon V(T) \to S$ be defined as follows: if $v=u_1\cdots u_\ell$ with
$\ell\geq 1$, $u_i\in S$ for every $i$, and $u_1=\set{u}$, then $\lambda(v)\eqdef u_\ell$.
In other words, $\lambda(v)$ is the last set occurring in the finite sequence $v$.
    \color{black}

    For a pair of worlds $w,w'$ of $K$ and a pair of vertices $v,v'$ of the "tree" $T$, let us write $\approx$ to denote the reflexive, symmetric and transitive closure of $\set{((w,v),(w',v'))\mid v' \text{ is child of } v \text{ in $T$, }w=w' \text{ and } w\in \lambda(v) \cap \lambda(v')}$. Let $[w,v]_\approx$ denote the equivalence class of $(w,v)$ with respect to $\approx$.

    We now define $\hat K$ as having 
    \begin{enumerate}
        \item \label{item:Kunravel:1} $\dom{\hat K} = \set{[w,v]_\approx \mid w \in \dom{K}, v \in V(T), w \in \lambda(v)}$;
        \item for every "atomic program" $a \in \Prog$ we have that $\dbracket{a}_{\hat K}$ consists of all pairs $([w,v]_\approx,[w',v]_\approx)$ of "worlds" such that $(w,w') \in \dbracket{a}_{K}$; and
        \item for every "atomic proposition"  $p \in \Prop$ we have that $\dbracket{p}_{\hat K}$ consists of all "worlds" $[w,v]_\approx$ such that $w \in \dbracket{p}_{K}$.
    \end{enumerate}
    Finally, consider the labeling $\bagmap \colon V(T) \to \pset{\dom{\hat K}}$ mapping every vertex $v$ to $\set{[w,v]_\approx \mid w \in \lambda(v)}$.

    \color{santichange}
    One may visualize the pair $(T,\bagmap)$ as follows. Each vertex $v$ of the "tree" $T$ carries a small local configuration $\lambda(v) \in \+S$ of at most $k$ "worlds" from $K$, and the corresponding "bag" $\bagmap(v)$ contains one local copy $[w,v]_\approx$ of each $w \in \lambda(v)$. Thus, moving downwards in $T$ amounts to changing the current local configuration, while the relation $\approx$ glues together those local copies that correspond to the same "world" and occur in two consecutive configurations. In this picture, a single "world" of $\hat K$ is obtained by following the same original "world" through a connected portion of $T$, and every "bag" records only the local copies visible at one vertex.
    \color{black}

    It follows that $(T,\bagmap)$ is a "tree decomposition" of $\hat K$ of "tree-width" at most $k-1$. 

    \begin{claim}\AP
        $(T,\bagmap)$ is a "tree decomposition" of $\hat K$ of "tree-width" at most $k-1$.
    \end{claim}
    \begin{proof}
        $(T,\bagmap)$ has to verify the three conditions ("A@@treedec", "B@@treedec", "C@@treedec") of a "tree decomposition". 

        Condition "A@@treedec" holds  since every "world" $[w,v]_\approx$ of $\hat K$ is contained in the "bag" of $v$.

        Condition "B@@treedec" holds for any "atomic program" $a \in \Prog$ for the same reason: every pair $([w,v]_\approx,[w',v']_\approx)$ of "worlds" such that $v=v'$ have that both $[w,v]_\approx$ and $[w',v']_\approx$ are contained in the "bag" of $v$.
        
        For the ``connectedness'' condition "C@@treedec", consider any "world" $[w,v]_\approx$ of $\hat K$, and observe that $[w,v]_\approx \in \bagmap(v')$ for any vertex $v'$ from $V_{w,v} = \set{v'\in V(T) \mid (w,v) \approx (w,v')}$. Further, $V_{w,v}$ forms a connected "subtree@tree" of $T$, and any vertex $v'$ of $T$ with a "bag" containing $[w,v]_\approx$ must be such that $v' \in V_{w,v}$. Hence, $(T,\bagmap)$ verifies condition "C@@treedec".

        \color{santichange}
        Finally, since every "bag" contains at most $k$ elements, the "tree-width" of $(T,\bagmap)$ is at most $k-1$. Note that the bound on the "tree-width" comes directly from the definition of $\+S$: since $\lambda(v)\in \+S$ for every vertex $v$ of $T$, we have $|\lambda(v)|\leq k$, and therefore every "bag" $\bagmap(v)=\set{[w,v]_\approx \mid w\in \lambda(v)}$ has size at most $k$.
        \color{black}
    \end{proof}
    \noindent
    We show that, by construction, $\hat K$ is "$k$-bisimilar" to $K$.
    \begin{claim}\AP 
        $K,u \pebblebisim{k} \hat K,\hat u$, where $\hat u = [u,r]_\approx$ and $r$ is the root of~$T$.
    \end{claim}
    \begin{proof}
        We show that, more generally, for any $k$-tuple $\bar u$ of "worlds" from $K$ 
        and any vertex of $T$ of the form $v = w \cdot \set{\bar u[1], \dotsc, \bar u[k]}$, we have
        $K,\bar u \pebblebisim{k} \hat K,\bar{\hat u}$,
        where $\bar{\hat u}[i] = [\bar u[i],v]_\approx$ for every $i$.

        Note that $\set{\bar u[i] \mapsto \bar{ \hat u}[i]}_{i\leq k}$ is a function, and further it is a bijection: if $\bar{\hat u}[i] = \bar{\hat u}[j]$ then $(\bar u[i],v) \approx (\bar u[j],v)$ which means that $\bar u[i] = \bar u[j]$. By definition of $\dbracket{a}_{\hat K}$ in fact both $\set{\bar u[i] \mapsto \bar{\hat u}[i]}_i$ and $\set{\bar{\hat u}[i] \mapsto \bar u[i]}_i$ are "partial homomorphisms" from $\kHoms(K,\hat K)$ and $\kHoms(\hat K, K)$, respectively. In this way we
        verify the first conditions of 
        $K,\bar u \pebblequasibisim k \hat K,\bar{\hat u}$ and $ \hat K,\bar{\hat u}\pebblequasibisim k K,\bar{u}$.

        Let $u'_i$ in $K$ be at "distance" $\leq 1$ from some $\bar u[j]$. We can then show 
        $K,\bar u[i\mapcoord u'_i] \pebblebisim{k} \hat K,\bar{\hat u}[i\mapcoord \hat u'_i]$ for $\hat u'_i = [u'_i,v']_\approx$ and $v' = v \cdot (\set{\bar u[i'] \mid i'\neq i} \cup \set{u'_i}) \in V(T)$.
        Observe that for every $i' \neq i$ we have
        $\bar{\hat u}[i'] = [\bar u[i'],v]_\approx = [\bar u[i'],v']_\approx$
        and thus that $\bar {\hat u}[i\mapcoord \hat u'_i]$ is of the required form.

        We proceed symmetrically for any $\hat u'_i$ in $\hat K$ at "distance" $\leq 1$ from some $\bar{\hat u}[j]$: by definition of $\hat K$ there must be some $v$
        and $u'_i$ such that $\hat u'_i = [u'_i,v]_\approx$, $\bar{\hat u}[j] = [\bar u[j],v]_\approx$, and  $u'_i$, $\bar u[j]$ are at "distance" $\leq 1$ in $K$. We then continue with 
        $K,\bar u[i\mapcoord u'_i]$ and
         $\hat K,\bar{\hat u}[i\mapsto \hat u'_i]$, of the required form.

        We can then repeat the same strategy ad aeternam, showing that $K,\bar u \pebblebisim{k} \hat K,\bar{\hat u}$.
    \end{proof}
    \noindent
    Finally, observe that $(T,\bagmap)$ and $\hat K$ are countable if $K$ is countable.
    The fact that it has a "good" "tree decomposition" immediately follows from \cite[Lemma 4.1]{DBLP:journals/jsyml/GollerLL09}, which shows that every countable "tree decomposition" of "width" $k$ can be transformed into one which is "good".
\end{proof}

\noindent
\AP
As is the case of $\ICPDL$, $\CPDLp$ is definable in Least Fixed Point logic, and thus it inherits the Löwenheim-Skolem ""countable model property"": if a $\CPDLp$ "formula" is satisfiable, it is satisfied in a countable "structure@@kripke".
We therefore obtain the following corollary from \Cref{prop:treewidth-k-unravelling} and \Cref{them:bisim-charac}.
\begin{corollary}\label{cor:treewidth-k-model-property}
    For every $k\geq 2$, $\CPDLg{\Tw}$ has the ``\,""$\Tw$-model property""'': if a formula $\phi \in \CPDLg{\Tw}$ is satisfiable, then it is satisfied in a countable "Kripke structure" of "tree-width@@structure" at most $k$. 
\end{corollary}

\noindent
We can now proceed to the proof of \Cref{thm:separation}.
\begin{proof}[Proof of \Cref{thm:separation}]
Let $k\geq 2$.  Note that the  "$(k+1)$-clique formula" $\xi_{k+1} \in \CPDLg{\Tw[k]}$  implies that there exists a directed $(k+1)$-clique of $a$'s starting from the current "world", and it is trivially satisfied in a directed $(k+1)$-clique "structure@@kripke". 
However, there cannot be any $(k+1)$-cliques in a "structure@@kripke" of "tree-width" $\leq k-1$. Hence, in light of \Cref{cor:treewidth-k-model-property}, $\xi_{k+1}$ cannot be expressed in $\CPDLg{\Tw[k-1]}$.
\end{proof}

\begin{remark}
    While \Cref{prop:treewidth-k-unravelling} works for every $k \geq 2$, it does not follow that $\CPDLg{\Tw[1]}$ has the "$\Tw[1]$-model property", because the characterization between $\pebblebisim{k}$ and $\CPDLg{\Tw[k-1]}$  of \Cref{them:bisim-charac} holds only for $k\geq 3$. In fact, since $\CPDLg{\Tw[1]} \langsemequiv \CPDLg{\Tw[2]}$ by \Cref{thm:ICPDL_equals_TW1_equals_TW2}, the "$3$-clique formula" is actually expressible in $\CPDLg{\Tw[1]}$.
\end{remark}

\subsection{Separation from Other Studied Logics}
\AP
""Unary negation fragment of first-order logic extended with regular path expressions"" (\UNFOreg), introduced by Jung et al.\ \cite{jung2018querying}, is the fragment of first-order logic with equality over relational structures given by the following grammar, where $P$ is a relation symbol, $R$ is a binary relation symbol of $\sigma$ and $\phi(x)$ has no free variables besides (possibly) $x$:
\begin{align*}
\phi &\eqqdef P(\bar x) \mid x = y \mid \phi \lor \phi \mid \phi \land \phi \mid \exists x\ \phi \mid \lnot \phi(x) \mid E(x,y)
\\
E&\eqqdef R \mid \overline R \mid E\circ E \mid E^* \mid \phi(x)? 
\end{align*}
The formal semantics is given in \cite{jung2018querying}, but notice that $E$ is a regular expression with tests that can be used as binary relations in formulas.

The following is stated without proof in \cite{jung2018querying}:
\begin{claim}\AP
$\ICPDL \not\leqq_{\Kripke} \UNFOreg$ and $\UNFOreg \not\leqq_{\Kripke} \ICPDL$.
\end{claim}
\noindent
For the first claim, consider the class $\+A$ of pointed "Kripke structures" $K,u$ such that there is $n>0$ and $v_1,\dots,v_n\in\dom{K}$ with $(v_i,v_{i+1})\in\dbracket{a}_K\cap\dbracket{b}_K$ for $i=1,\dots,n-1$, $v_1=u$ and $u_n\in\dbracket{p}_K$, where $a,b\in\Prog$ and $p\in\Prop$. It can be shown that $\+A$ is not expressible in $\UNFOreg$, but expressible in $\ICPDL$ by the formula $\tup{(a\cap b)^*\circ p?}$. 

For the second claim, observe that the class of pointed "Kripke structures" $K,u$ where $u$ satisfies the `4-clique' formula $\xi_4$ of the proof of Theorem \ref{thm:separation} is expressible by the translation of $\xi_4$ to $\UNFOreg$, namely $\varphi(x_1) = \exists x_2,x_3,x_4.\bigwedge_{1\leq i<j\leq 4}a(x_i,x_j)\land\lnot\exists y.a(x_1,y)\land a(y,y)$, but not expressible in $\CPDLg{\Tw[2]}$, which is equal to $\ICPDL$ by Theorem \ref{thm:ICPDL_equals_TW1_equals_TW2}.

\AP ""Guarded negation first-order logic""  ($\intro*\GNFO$), introduced in Barany et al.\ \cite{BaranyCS15}, is the fragment of first-order logic with equality over relational structures given by the following grammar, where $P$ is a relation symbol of $\sigma$ and $\alpha$ is a relation symbol or equality:
\[
\phi \eqqdef P(\bar x) \mid x = y \mid \phi \lor \phi \mid \phi \land \phi \mid \exists x\ \phi \mid \alpha(\bar x \bar y) \land \lnot \phi(\bar y)
\]
Here $P(\bar x)$ and $\alpha(\bar x \bar y)$ are atoms whose free variables are $\bar x$ and $\bar x \bar y$, respectively.
Observe that in $\GNFO$, formulas with 0 or 1 free variables are closed under negation -- the latter through the equivalence $\lnot\phi(x) \semequiv (x=x \land \lnot\phi(x))$.
$\GNFO$ enjoys many desirable properties \cite{Segoufin17}.

\begin{proposition}\AP
$\GNFO \not \lleq[\Kripke]\UCPDLp$.
\end{proposition}
\begin{proof}
Let $\phi(x)\eqdef \exists y\ a(x,y)\land\lnot b(x,y)$ be a formula in $\GNFO$ (in fact it is even in guarded-quantification FO) over a signature with binary relation symbols $a$ and $b$. Let $K$ be the "Kripke structure" with $\dom{K}=\{u\}$, $a^K=\{(u,u)\}$ and $b^K=\{(u,u)\}$, and let $K'$ be the "Kripke structure" with $\dom{K'}=\{u',v'\}$, $a^{K'}=\{(u',u'),(u',v'),(v',v')\}$ and $b^{K'}=\{(u',u'),(v',v')\}$, see Figure \ref{fig:GNFOnotleqUCPDL}.  
\begin{figure}
\includegraphics[scale=0.25]{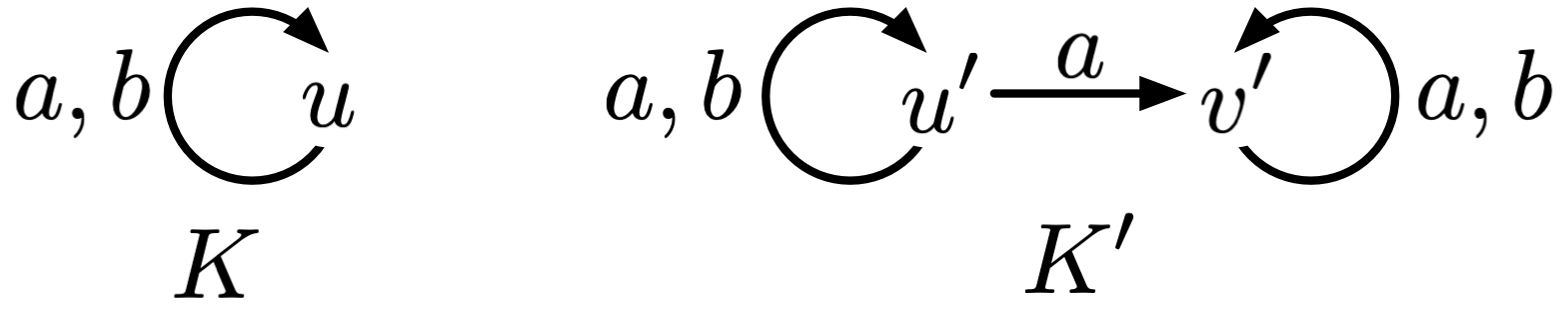}
\caption{$K,u$ and $K',u'$ are bisimilar for $\UCPDLp$ but distinguishable in $\GNFO$.}
\label{fig:GNFOnotleqUCPDL}
\end{figure}
One can show that $K,u$ and $K',u'$ are bisimilar for $\UCPDLp$ (\ie, "$k$-$\Univ$-bisimilar" for every $k$). Since $K\not\models\phi[x\mapsto u]$ and $K'\models\phi[x\mapsto u']$, we conclude that there is no $\UCPDLp$-formula equivalent to $\phi(x)$. 
\end{proof}

\noindent
\AP
It is also the case that there are properties of $\CPDLp$ (or even $\UNFOreg$) which cannot be expressed either in $\GNFO$ or in its extension with guarded fixed-points, known as $\intro*\GNFP$~\cite{BaranyCS15}.
\begin{proposition}\AP
    $\CPDLp \not\lleq[\Kripke] \GNFP$ (and hence $\UCPDLp \not\lleq[\Kripke]\GNFO$).
\end{proposition}
\begin{proof}
    The unary property stating ``$x$ is in an $a$-cycle'' can be easily expressed in $\ICPDL$ or $\UNFOreg$ (hence also in $\CPDLp$), but it cannot be expressed in $\GNFP$, as shown in \cite[Proof of Proposition 2 (Appendix A.1 of full version)]{DBLP:conf/lics/BenediktBB16}.
\end{proof}

\color{black}

%% file: satisfiability-altbis.tex
\section{Satisfiability}
\label{sec:sat}
The ""satisfiability problem"" for $\UCPDLp$ is the problem of, given a $\UCPDLp$ "formula" $\phi$, whether there exists some "structure" $K$ and a "world" $w$ thereof such that $K, w \models \phi$.

Observe that any countable "tree decomposition" can be turned into a "tree decomposition" whose underlying "tree" is binary, preserving the "width".
Decidability of the "satisfiability problem" for $\UCPDLp$ follows readily from the previous \Cref{cor:treewidth-k-model-property}, the fact that $\UCPDLp$ expressions can be effectively expressed in Monadic Second Order Logic (MSO), and the known result that satisfiability for MSO on structures of bounded "tree-width" is decidable (by MSO interpretations onto "tree"-structures and Rabin's theorem, see eg~\cite[Fact~2]{DBLP:journals/apal/Seese91}). 

\begin{proposition}\AP
    $\UCPDLp$-"satisfiability@satisfiability problem" is decidable.
\end{proposition}

\paragraph*{Complexity}
Next, we will pin down the complexity of the satisfiability problem.
The main result of this section is the following:
\begin{description}
    \item[Main result] (\Cref{thm:sat-cpdlp}): The "satisfiability@satisfiability problem" for $\UCPDLp$ is decidable in doubly-exponential time, and it is decidable singly-exponential if we bound the `intricacy' of the "conjunctive programs", which we call \reintro{conjunctive width}.
\end{description}

\noindent
\color{diegochange}
Let us begin by first defining the "conjunctive width" of a $\UCPDLp$ "program". Intuitively, it is the maximum number of conjuncts that a "conjunctive program" may have. 
For example, the "conjunctive width" of 
$C[x_s,x_t]$ for $C=\set{ a(x_s,x_t), b(x_s,z), c(x_t,z)}$ is 3, but it is 4 for $C'[y_s,y_t]$ if $C'=\set{ a(y_s,y_t), (b\circ \pi)(y_s,y_t)}$ and $\pi=C[x_s,x_t]$ since it is effectively like having the 4 atoms $\set{ a(y_s,y_t), (b\circ a)(y_s,y_t), (b\circ b)(y_s,z), (b\circ c)(y_t,z)}$.\footnote{This measure is inspired by the `"intersection width"' of $\ICPDL$ programs from the work by G\"oller et al.~\cite{DBLP:journals/jsyml/GollerLL09} and it can be seen as a natural generalization in the context of "conjunctive programs".} 
For arbitrary arity structures we also take into account the arity of higher-arity relations, where the extension to "r-atoms" of arity $k$ is done by simply weighing the atom by $k$ instead of $1$.
\color{black}
\AP
Formally, the ""conjunctive width"" $\intro*\cqsize{\pi}$ of a $\UCPDLp$ "program" $\pi$ is the following measure:
\begin{align*}
    &\cqsize{a} = \cqsize{\Univ} = \cqsize{\phi?} = \cqsize{\epsilon} = \cqsize{\bar a} \eqdef 1
        &&\text{ for } a \in \Prog,
    \\
    &\cqsize{\pi_1 \star \pi_2} \eqdef \max \set{ \cqsize{\pi_1}, \cqsize{\pi_2}}
        &&\text{ for $\star \in \set{\circ,\cup}$},
    \\
    &\cqsize{\pi^*} \eqdef \cqsize\pi,
    \\
    &\cqsize{C[x_s,x_t]} \eqdef \sum \set{\arity(R) : R(\bar x)  \text{ "r-atom" of } C} + {}
    \\
    &\hspace{.4cm} \sum \set{\cqsize\pi : \pi(x,y) \text{ "p-atom" of } C }.
\end{align*}
We define $\AP\intro*\Cqsize{e}$
$\phantomintro*\Cqsizealt$%
for any "expression" $e$ to be the maximum "conjunctive width" of a "program" therein (or 1 if it contains no "program").
\begin{remark}\label{rk:treewidth-conjunctive-width}
    For every $\UCPDLp$ "expression" $e$ we have $e \in \UCPDLg{\Tw[k]}$ for $k=\Cqsize{e}$.
\end{remark}

\noindent
To prove the main result, we follow the proof strategy of \cite[Theorem~4.8]{DBLP:journals/jsyml/GollerLL09} for satisfiability for $\ICPDL$, and adapt it to the more general setting of $\UCPDLp$. 
The proof goes by a series of reductions \ref{eq:redux:0} $\rightarrow$ \ref{eq:redux:1} $\rightarrow$ \ref{eq:redux:2} $\rightarrow$ \ref{eq:redux:2.5} $\rightarrow$ \ref{eq:redux:3} $\rightarrow$ \ref{eq:redux:4} $\rightarrow$ \ref{eq:redux:5}:
\begin{enumerate}[(a),wide=0pt]
    \item \label{eq:redux:0} From the "satisfiability problem" of $\UCPDLp$ on   $\Rels$-structures to 
    \item \label{eq:redux:1} the "satisfiability problem" of $\UCPDLp$ in `"unnested form"' (meaning that the `"nesting depth"' of "test programs" is at most 2) on "$\Rels$-structures", to 
    \item \label{eq:redux:2} the "unnested@unnested form" "satisfiability problem" ("ie", the "satisfiability problem" for "expressions" in "unnested form") on "Kripke structures", to
    \item \label{eq:redux:2.5} the "unnested@unnested form" "satisfiability problem" of $\CPDLp$ ("ie", no $\Univ$ "program") on "Kripke structures", to
    \item \label{eq:redux:3} the "unnested@unnested form" "satisfiability problem" of $\CPDLp$ on "Kripke structures" of "tree-width@@structure" $k$, where $k$ is the "tree-width@@pdl" of the input formula, to
    \item \label{eq:redux:4} the "unnested@unnested form" "satisfiability problem" of $\CPDLp$ over ``$\omega$-regular trees'', to
    \item \label{eq:redux:5} the emptiness problem of an "automata model@TWAPTA" on "trees".
\end{enumerate}

\paragraph{Organization} 
In \Cref{ssec:redux-simple-expressions} we define the "unnested form" for $\UCPDLp$ expressions, and show the reductions \ref{eq:redux:0} $\to$ \ref{eq:redux:1} $\to$ \ref{eq:redux:2} $\to$  \ref{eq:redux:2.5}.
The reduction \ref{eq:redux:2.5} $\rightarrow$ \ref{eq:redux:3} is trivial from \Cref{cor:treewidth-k-model-property}.
In \Cref{ssec:omega-regular-tree-sat} we define the "tree automata model@TWAPTA" and the "$\omega$-regular tree satisfiability" problem relative to this model.
In \Cref{ssec:redux-sat-to-omegatreesat} we show the reduction \ref{eq:redux:3} $\rightarrow$ \ref{eq:redux:4}.
In \Cref{sec:solving-omega-reg-sat} we prove a complexity upper bound for the "$\omega$-regular tree satisfiability" problem (\Cref{prop:omega-sat-pb-2Exp}) via a reduction \ref{eq:redux:4} $\rightarrow$ \ref{eq:redux:5} --this is the most technical part.
Finally, in \Cref{ssec:main-sat-result} we put everything together and deduce our main result for the "satisfiability problem" of $\UCPDLp$ in \Cref{thm:sat-cpdlp}.

To simplify the complexity statements of this section, we shall write $\AP\intro*\polyfun(n)$ to denote $n^{\+O(1)}$, and $\intro*\expfun(n)$ to denote $2^{n^{O(1)}}$.

\subsection{A Polynomial-time Reduction to Simple Expressions}
\label{ssec:redux-simple-expressions}
We shall first simplify our satisfiability problem to formulas of a simple form having: low nesting, no universal quantification and no "r-atoms". Further, the reduction preserves the "conjunctive width".

\color{diegochange}
\knowledgenewrobustcmd{\nestingdepth}{\cmdkl{nd}}
The \AP""nesting depth"" $\AP\intro*\nestingdepth(e)$ of a $\UCPDLp$ "expression" $e$ is the maximum number of nested "test programs" of the form `$\phi?$' it contains, defined as follows:
\begin{align*}
    \nestingdepth(p) = \nestingdepth(a) = \nestingdepth(\bar a) = \nestingdepth(\epsilon) = \nestingdepth(\Univ) &\eqdef 0 \qquad \text{ for } p \in \Prop, a \in \Prog,\\
    \nestingdepth(\lnot \phi) &\eqdef \nestingdepth(\phi),\\
    \nestingdepth(\phi_1 \land \phi_2) &\eqdef \max \set{\nestingdepth(\phi_1), \nestingdepth(\phi_2)},\\
    \nestingdepth(\tup{\pi}) &\eqdef \nestingdepth(\pi),\\
    \nestingdepth(\pi_1 \star \pi_2) &\eqdef \max \set{\nestingdepth(\pi_1), \nestingdepth(\pi_2)} \qquad \text{ for } \star \in \set{\circ, \cup},\\
    \nestingdepth(\pi^*) &\eqdef \nestingdepth(\pi),\\
    \nestingdepth(\phi?) &\eqdef 1 + \nestingdepth(\phi),\\
    \nestingdepth(C[x_s,x_t]) &\eqdef \max (\set{0} \cup \set{\nestingdepth(\pi) : \pi(x,y) \text{ "p-atom" of } C}).
\end{align*}
\color{black}
Let us call a $\UCPDLp$ "expression" to be in \AP""unnested form"" if \textit{(i)} its "nesting depth" is at most $2$ and \textit{(ii)} every nested "test program" is of the form `$p?$' for an "atomic proposition" $p \in \Prop$. 
We will show that each $\UCPDLp$ "expression" $e$ can be transformed into one $e'$ in "unnested form" with no "universal programs" nor "r-atoms", such that $e$ and $e'$ are \AP""equi-satisfiable"", that is, $e$ is satisfiable if{f} $e'$ is satisfiable.

\begin{proposition}\AP\label{prop:polyredux-simpleform}
    For every $\UCPDLp$ "expression" $e$ one can produce, in polynomial-time, an "equi-satisfiable" "expression" $e'$ such that
    \begin{itemize}
        \item $e'$ is in "unnested form",
        \item $e'$ does not contain "r-atoms",
        \item $e'$ does not contain $\Univ$ "programs",
        \item $\Cqsize{e}=\Cqsize{e'}$,
    \end{itemize}
\end{proposition}
\begin{proof}
    Follows directly by applying the next \Cref{lem:unnested-form,lem:redux:struct-to-kripke,lem:removeUniv} below.
\end{proof}

\noindent
We will next show how to reduce to "unnested form" (\Cref{lem:unnested-form}), then how to remove "r-atoms" (\Cref{lem:redux:struct-to-kripke}), and how to remove "universal programs" (\Cref{lem:removeUniv}). \Cref{prop:polyredux-simpleform} will follow as a direct consequence of these reductions.
\paragraph{Reduction to low nesting depth expressions}

\begin{lemma}\AP\label{lem:unnested-form}
    For every $\UCPDLp$ "expression" $e$ one can produce in polynomial time an "equi-satisfiable" "expression" $e'$ in "unnested form" such that $\Cqsize{e}=\Cqsize{e'}$.
\end{lemma}
\diego{prueba modificada (antes era en árboles, ahora en estructuras)}
\begin{proof}
    Let $\phi \in \UCPDLp$.
    Let $\psi?$ be a "subexpression" of $\phi$ such that
    \begin{itemize}
        \item $\psi$ is not an "atomic program" and has "nesting depth" $\leq 1$, and
        \item $\psi?$ appears inside a "test program" of $\phi$.\footnote{That is, $\phi$ contains a "subexpression" of the form `$\tau?$', and $\tau$ contains $\psi?$ as "subexpression".}
    \end{itemize} 
    Let $\phi'$ be the result of replacing every occurrence of `$\psi?$' with `$p_{\psi}?$' in $\phi$, where $p_{\psi} \in \Prop$ is a fresh "atomic proposition" not used in $\phi$. We then produce the "equi-satisfiable" expression $\hat\phi \eqdef \phi' \land \lnot\tup{\Univ \circ ((p_\psi \land \lnot \psi)\lor (\lnot p_\psi \land \psi))?}$. Note that the second conjunct of $\hat \phi$ expresses that there is no node where the truth value of $p_\psi$ differs from that of $\psi$, and hence that instead of testing for $p_\psi$ one can equivalently test for $\psi$ in $\phi'$.
    If $\hat \phi$ is satisfied in some "structure" $K$, then $\phi$ is also satisfied in $K$, due to the `$\psi \Leftrightarrow p_\psi$' property above.
    Conversely, if $\phi$ is satisfied in some "structure" $K$, consider $K'$ the "structure" resulting from extending $K$ by making $p_\psi$ hold, precisely, wherever $\psi$ holds true; it follows that $K' \models \hat \phi$. Observe that $\Cqsize{\phi} = \Cqsize{\hat \phi}$.
    By applying such reduction a linear number of times we obtain an "equi-satisfiable" "expression" in "unnested form".
\end{proof}

\paragraph{Reduction to "Kripke structures"}
\label{ssec:reduction-to-kripke-sat}
We show a polynomial time translation which yields, for every $\UCPDLp$ "formula" $\phi$, a "formula" $\phi'$ which uses only relations from $\Prog \cup \Prop$ (\ie, no relation of "arity" greater than $2$), such that $\phi$ is satisfiable if{f} $\phi'$ is satisfiable (over "Kripke structures").

\begin{lemma}\AP\label{lem:redux:struct-to-kripke}
    For every $\UCPDLp$ "expression" $e$ one can produce in polynomial time an "equi-satisfiable" "expression" $e'$ such that
    \begin{enumerate}
        \item $e'$ does not contain "r-atoms",
        \item if $e$ is in "unnested form", so is $e'$,
        \item $\Cqsize{e}=\Cqsize{e'}$.
    \end{enumerate}
\end{lemma}
\begin{proof}
    \color{santichange}
    The reduction is simply the result of replacing every "r-atom" $R(x_1, \dotsc, x_n)$ in a "conjunctive program" with the set of "p-atoms" $R^{[1]}(z,x_1), \dotsc, R^{[n]}(z,x_n)$, where $z$ is a fresh variable and $R^{[1]}, \dots, R^{[n]} \in \Prog$ are pairwise distinct "atomic programs" not used elsewhere. 
    We claim that $e$ is satisfiable iff $e'$ is satisfiable. For the left-to-right direction, suppose that $e$ is satisfiable in some "Kripke structure" $K$. We expand $K$ into a "Kripke structure" $K'$ by keeping the old "worlds" and interpretations of all symbols already occurring in $e$, and for every tuple $(a_1,\dotsc,a_n) \in \dbracket{R}_K$ we add a fresh "world" $b_{(a_1,\dotsc,a_n)}$ such that $(b_{(a_1,\dotsc,a_n)},a_i) \in \dbracket{R^{[i]}}_{K'}$ for every $1\leq i\leq n$. Since the variable $z$ is fresh, any assignment satisfying the original atom $R(x_1,\dotsc,x_n)$ can be extended by mapping $z$ to the corresponding witness $b_{(a_1,\dotsc,a_n)}$, and hence satisfies its translation. 
    For the converse direction, suppose that $e'$ is satisfiable in some "Kripke structure" $K'$. We keep the same "worlds" and the same interpretations of all symbols already occurring in $e$, and interpret each relation symbol $R$ by setting $(a_1,\dotsc,a_n) \in \dbracket{R}_K$ iff there exists some "world" $b$ such that $(b,a_i) \in \dbracket{R^{[i]}}_{K'}$ for every $1\leq i\leq n$. Since the "atomic programs" $R^{[1]},\dotsc,R^{[n]}$ are fresh and do not occur elsewhere, this definition does not interfere with the interpretation of the rest of $e$, and any assignment satisfying the translated atom in $K'$ satisfies the original atom in $K$. Thus $e$ is satisfiable iff $e'$ is satisfiable.
    \color{black}

    To show that this reduction preserves the "conjunctive width", we proceed by induction in the structural complexity of $e$. The only interesting case is when $e$ is a "conjunctive program" $C[x,y]$. Each "r-atom" $R(x_1, \dotsc, x_n)$ of $C$ contributes with an additive factor $n$ to the "conjunctive width" of $e$. In the translation $e'$ of $e$, such "r-atom" is no longer present, but instead it was replaced by $n$ "p-atoms" in $\Prog$, each of them having "conjunctive width" equal to 1 by definition. Therefore, $\Cqsize{e}=\Cqsize{e'}$.     
\end{proof}

\paragraph{Removal of "universal programs"}
We show a standard way to remove $\Univ$ modalities.
\begin{lemma}\AP\label{lem:removeUniv}
    For every $\UCPDLp$ "expression" $e$ one can produce in polynomial time an "equi-satisfiable" "expression" $e'$ such that
    \begin{enumerate}
        \item $e'$ does not contain $\Univ$ "programs",
        \item if $e$ does not contain "r-atoms", neither does $e'$,
        \item if $e$ is in "unnested form", so is $e'$,
        \item $\Cqsize{e}=\Cqsize{e'}$.
    \end{enumerate}
\end{lemma}
\begin{proof}
    Given a $\UCPDLp$ "expression" $e$, let $\set{a_1, \dotsc, a_n}$ be the set of "atomic programs" present in $e$, and let $a_0 \in \Prog$ be an "atomic program" not used in $e$. Let $e'$ be the result of replacing each appearance of $\Univ$ with the "program" $ (a_0 \cup \dotsb \cup a_n \cup \bar a_0 \cup \dotsb \cup \bar a_n)^*$.
    
    Note that this reduction does not add "r-atoms" and does not increase the "nesting depth" of $e$.
    It does not either increase the "conjunctive width" since $\cqsize{\Univ} = \cqsize{(a_0 \cup \dotsb \cup a_n \cup \bar a_0 \cup \dotsb \cup \bar a_n)^*}=1$.
    Further the expressions are "equi-satisfiable".
    \begin{claim}\AP
        $e$ is satisfiable if{f} $e'$ is satisfiable.
    \end{claim}
    \noindent
    Suppose first that $e'$ is satisfiable. Then there is some "Kripke structure" $K$ and element $\bar t$ such that $\bar t \in \dbracket{e'}_K$ (where $\bar t$ can be a pair of elements or just one element depending on whether $e'$ is a "program" or a "formula"). It then follows that $\dbracket{e}_{K'} \neq \emptyset$ where $K'$ is the substructure of $K$ induced by the $A$-connected component of $\bar t$, for $A = \set{a_0, \dotsc, a_n}$.
    
    Suppose now that $e$ is satisfiable. Then $\dbracket{e}_K \neq \emptyset$ for some "Kripke structure" $K$. Let $K'$ be the result of adding an $a_0$-relation between every pair of elements of $K$. It then follows that $\dbracket{e}_{K'} \neq \emptyset$.
\end{proof}

\subsection{Definition of \texorpdfstring{$\omega$}{ω}-regular Tree Satisfiability Problem}
\label{ssec:omega-regular-tree-sat}
In this section we define the "$\omega$-regular tree satisfiability" problem, which is the problem of deciding, given a $\CPDLp$ "formula" $\phi$ in "unnested form" and without "r-atoms" and a "two-way alternating parity tree automaton" ("TWAPTA") $\+T$ (defined in the next subsection), whether there is an infinite "tree" in $\langTWAPTA{\+T}$ that, when viewed as a "Kripke structure", is a model of $\phi$.

We first define the notion of tree we use.
\AP
Let $\intro*\SigmaN$ be a finite alphabet of ""node labels"" and $\intro*\SigmaE$ a finite alphabet of ""edge labels"". 
A ""$\SigmaN$-labeled $\SigmaE$-tree"" is a partial function $T \colon \SigmaE[*] \to 
\SigmaN$ whose domain, denoted $\intro*\domT(T)$, is prefix-closed. The elements of $\domT(T)$ are 
the nodes of $T$. If $\domT(T) = \SigmaE[*]$, then $T$ is called ""complete"". 
\AP
In the context of a "complete" tree $T$, 
a node $va \in \domT(T)$ with $a \in \SigmaE$, is called the ""$a$-successor"" of $v$, and $v$ is the ""$a$-predecessor"" of $va$. Notice that the "$a$-successor" of $v$ is unique, since nodes are words over $\SigmaE$.

In the rest of the section we work with "complete" trees.
\AP
We use $\intro*\tree(\SigmaN, \SigmaE)$ to denote the set of all "complete" "$\SigmaN$-labeled $\SigmaE$-trees". 
If $\SigmaE$ is not important, we simply talk of \reintro{$\SigmaN$-labeled trees}.

The trees accepted by "TWAPTAs" are "complete" "$\pset{P}$-labeled $A$-trees", where $P \subseteq \Prop$ and $A \subseteq \Prog$ are finite sets of "atomic propositions" and "atomic programs", respectively. 
Such a "tree" $T$ can be identified with the "Kripke structure" $K$ with $\dom{K} = A^*$, with
$\dbracket{p}_K = \set{u \in A^* \mid \mbox{the label of $u$ contains $p$}}$
for every $p \in P$, and $\dbracket{a}_K = \set{(u,ua) \mid u \in A^*}$
for every $a \in A$.

Observe that "Kripke structures" derived in this way are deterministic and total 
with respect to $A$, \ie,  the transition relation $a^K$ is a total function 
for each $a \in A$.

"TWAPTAs" are walking automata which, at each transition, can read the current node label from $\SigmaN$ and move to a child checking an edge label from $\SigmaE$, move to the parent by checking an edge label from $\SigmaE$, or stay at the same node. Further, the automata model features alternation, so transitions are "positive Boolean combinations@\posB" of these kinds of moves. The acceptance is based on a "parity condition@successful". The formal definition is given in the next section.

\subsubsection{Definition of Two-Way Alternating Parity Tree Automata}
\label{sec:def:twapta}
\AP
To define "TWAPTAs", we need a few preliminaries. For a finite set $X$, we denote by $\intro*\posB(X)$ the set of all positive Boolean formulas where the elements of $X$ are used as variables. The constants \textit{true} and \textit{false} are admitted, "ie", we have $\textit{true}$, $\textit{false}$ $\in \posB (X)$ for any set $X$. A subset $Y \subseteq X$ can be seen as a valuation in the obvious way by assigning $\textit{true}$ to all elements of $Y$ and $\textit{false}$ to all elements of $X \setminus Y$. 
\AP
For an edge alphabet $\SigmaE$, let $\intro*\barSigmaE \eqdef \set{ \bar a \mid a \in \SigmaE}$ be a disjoint copy of $\SigmaE$. For $u \in \Sigma^*$ and $d \in \SigmaE \cup \barSigmaE \cup \set \epsilon$ define:
\begin{align*}
u {\cdot} d &= \begin{cases}
    ud & \text{if $d \in \SigmaE$}\\
    u & \text{if $d = \epsilon$}\\
    v & \text{if there is $a \in \SigmaE$ with $d = \bar a$ and $u = va$}\\
    \textit{undefined} & \text{otherwise}
\end{cases}
\end{align*}

\noindent
\AP
A ""two-way alternating parity tree automaton"" ("TWAPTA") over "complete" $\SigmaN$-labeled $\SigmaE$-trees is a tuple $T = (S, \delta, s_0 , \textit{Acc})$, where

\begin{itemize}
    \item $S$ is a finite non-empty set of states,
    \item  $\delta \colon S \times \SigmaN \to \posB(mov(\SigmaE))$ is the transition function, where $mov(\SigmaE) = S \times (\SigmaE \cup \barSigmaE \cup \set{ \epsilon } )$ is the set of moves,
    \item  $s_0 \in S$ is the initial state, and
    \item  $\textit{Acc} \colon S \to \Nat$ is the priority function.
\end{itemize}

\noindent
For $s \in S$ and $d \in \SigmaE \cup \barSigmaE \cup \set\epsilon$, we write the corresponding move as $\tup{s, d}$. Intuitively, a move $\tup{s, a}$ with $a \in \SigmaE$, means that the automaton sends a copy of itself in state $s$ to the "$a$-successor" of the current tree node. Similarly, $\tup{s,\bar a}$
means to send a copy to the "$a$-predecessor" (if exists), and $\tup{s,\epsilon}$ means to stay in the current node. Formally, the behaviour of "TWAPTAs" is defined in terms of runs. Let $\+T$ be a "TWAPTA" as above, $T \in \tree(\SigmaN , \SigmaE)$, $u \in \SigmaE[*]$ a node, and $s \in S$ a state of $\+T$. \AP An ""$(s,u)$-run"" of $\+T$ on $T$ is a (not necessarily "complete") $(S \times \SigmaE[*])$-labeled tree $T_R$ such that
\begin{itemize}
    \item $T_R(\epsilon) = (s,u)$, and
    \item for all $\alpha \in \domT(T_R)$, if $T_R (\alpha) = (p, v)$ and $\delta(p, T(v)) = \theta$, then there is a subset $Y \subseteq mov(\SigmaE)$ that satisfies $\theta$ and such that for all $(p', d) \in Y$, $v \cdot d$ is defined and there exists a successor $\beta$ of $\alpha$ in $T_R$ with $T_R(\beta) = (p', v \cdot d)$.
\end{itemize}

\noindent
\AP
We say that an "$(s,u)$-run" $T_R$ is ""successful"" if for every infinite path $\alpha_1 \alpha_2 \dotsb$ in $T_R$ (which is assumed to start at the root), the following number is even:
\begin{align*}
    \min \{ \textit{Acc}(s) \mid {}& s \in S \text{ with } T_R (\alpha_i ) \in \set{ s } \times \SigmaE[*] 
    \text{ for infinitely many } i \}.
\end{align*} \sidesanti{quizá mejor en el min usar otra letra en vez de $s$, porque estamos definiendo un successful $(s,u)$-run}

\noindent
\AP
For $s \in S$ define
$\intro*\dbracketaut{\+T , s} \eqdef \{ (T, u) \mid T \in \tree(\SigmaN , \SigmaE ), u \in \SigmaE[*]$ , and
there exists a "successful" "$(s,u)$-run" of $\+T$ on $T$ $\}$ and
$\dbracketaut{\+T} \eqdef \dbracketaut{\+T , s_0}$.

Now the language $\intro*\langTWAPTA{\+T}$ accepted by $\+T$ is defined as
\AP
\begin{align*}
    \langTWAPTA{\+T} \eqdef \set{ T \in \tree(\SigmaN , \SigmaE ) \mid (T, \epsilon) \in \dbracketaut{\+T} }.
\end{align*}

\noindent
\AP
For a "TWAPTA" $\+T = (S, \delta, s_0 , \textit{Acc})$, we define its ""state size"" $\intro*\sizeStates{\+T} \eqdef | S |$ as its number of states and we define its ""index"" $\intro*\indexAut{\+T}$ as $\max \set{ \textit{Acc}(s) \mid s \in S }$. The ""transition size"" $\sizeTrans{\+T}$ is the sum of the lengths of all positive Boolean functions that appear in the range of $\delta$, $\intro*\sizeTrans{\+T} \eqdef \sum_{(s,a) \in S \times \SigmaN} |\delta(s,a)|$.

\color{diegochange}
The emptiness problem for "TWAPTAs" is known to be decidable within these bounds:

\begin{thmC}[{\cite{DBLP:conf/icalp/Vardi98},
  \cite[Theorem~3.1]{DBLP:journals/jsyml/GollerLL09}}]
\label{thm:TWAPTA-complexity}
For any given "TWAPTA" $\+T$, it can be checked in time $\expfun( \sizeStates{\+T}  + \indexAut{\+T}) \cdot \sizeTrans{\+T}^{O(1)}$ whether $\langTWAPTA{\+T} \neq \emptyset$. 
\end{thmC}

\color{black}

\subsubsection{The $\omega$-regular tree satisfiability problem}

We can now formally define "$\omega$-regular tree satisfiability". Let $\phi$ be a $\CPDLp$ "formula", let $A = \set{ a \in \Prog \mid a \text{ occurs in }\phi}$ and $P = \set{ p \in \Prop \mid p \text{ occurs in } \phi}$. 
\AP
The "formula" $\phi$ is ""satisfiable with respect to@@twapta"" a "TWAPTA" $\+T$ over "$\pset{P}$-labeled $A$-trees" if there is $T \in \langTWAPTA{\+T}$ such that $\epsilon \in \dbracketaut{\phi}_T$. 
\AP
Finally, ""$\omega$-regular tree satisfiability"" is the problem to decide, given $\phi \in \CPDLp$ in "unnested form" and without "r-atoms" and a "TWAPTA" $\+T$, whether $\phi$ is "satisfiable with respect to@@twapta" $\+T$.
\sidediego{observar que ahora toma un unnested formula como entrada}

\subsection{Reduction to \texorpdfstring{$\omega$}{ω}-regular Tree Satisfiability Problem}
\label{ssec:redux-sat-to-omegatreesat}
\color{diegochange}
The satisfiability problem can be narrowed down by leveraging the logic's "$\Tw$-model property". These tree-like models can be unrolled into infinite trees, allowing the problem to be reduced into checking for the existence of a suitable tree model accepted by a "TWAPTA". Intuitively, each tree accepted by a "TWAPTA" encodes the "tree decomposition" of a "Kripke structure" such that a modified version $\phi'$ of the input formula $\phi$ holds true in $T$ if, and only if, it holds in the "structure" denoted by $T$.
\color{black}
More concretely, for every $\CPDLp$ "formula" $\phi$ over "Kripke structures" there is a polynomial-time computable $\CPDLp$ "formula" $\phi'$ and an exponential-time computable "TWAPTA" $\+T$ such that the following holds.
\begin{lemma}\AP\label{lem:redux-omega-tree-sat}
    $\phi$ is satisfiable if{f} $\phi'$ is "satisfiable with respect to@@twapta" $\+T$. Moreover, $\Cqsize{\phi}=\Cqsize{\phi'}$, and $\pdlsize{\phi'}$ and $\sizeStates{\+T}$ are polynomial in $\pdlsize{\phi}$.\footnote{Observe, however, that $\+T$ may have an exponential number of transitions.} 
    \AP
    $\+T$ depends only on the set of "atomic programs" and "atomic propositions" occurring in $\phi$ as well as on the ""tree-width@@pdl"" of $\phi$ (ie, the minimum $k$ such that $\phi \in \CPDLg\Tw$).  If $\phi$ is in "unnested form", so is $\phi'$.\color{black}
\end{lemma}
\begin{proof}
This is essentially the reduction shown in \cite[\S 4.2]{DBLP:journals/jsyml/GollerLL09}. 
\color{diegochange}
The basic idea is rather classical: one builds a "TWAPTA" so that each tree of the language accepts a \sidediego{good?} "tree decomposition" of a "Kripke structure" $K$ whose "tree-width" is at most $k$, where $k$ is the "tree-width@@pdl" of $\phi$  (two, in the case of $\ICPDL$).
The encoding uses special "atomic propositions" and "atomic programs" to encode, for each node of the tree, which is the structure restricted to the associated bag. 
In other words, we visualize each node of the tree as having $k+1$ available `slots' that can be empty or filled with a "world" of $K$, corresponding to the ordered representation of $\set{0,\dotsc, k}$ of the bag's elements at the current node.
For example, if the "atomic proposition" `$(i,R,j)$' holds at a node of the tree, then there is a $R$-relation between the "world" represented by the $i$-th element in the current bag and the $j$-th element of the current bag ("ie", the pair of represented elements belongs to $R^K$).
We refer the reader to \cite[\S 4.2]{DBLP:journals/jsyml/GollerLL09} for more details on the intuition of the construction.
\color{black}

In our case the construction consists of a trivial modification of \cite[\S 4.2]{DBLP:journals/jsyml/GollerLL09}, where the only difference is that the reduction needs to be adapted so that it allows a reduction from "Kripke structures" of arbitrary "tree-width" instead of just "tree-width" $2$. 
It is not worth showing it here since the proof is literally the result of replacing everywhere, in \cite[\S 4.2]{DBLP:journals/jsyml/GollerLL09}'s construction, the set $\set{0,1,2}$ (to keep track of a "width" 2 "tree decomposition") with the set $\set{0,1, \dotsc, k}$ where $k$ is the "tree-width@@pdl" of $\phi$, together with the trivial observation that this does not induce an extra exponential blowup.
\end{proof}

\subsection{Complexity of \texorpdfstring{$\omega$}{ω}-regular Tree Satisfiability Problem}
\label{sec:solving-omega-reg-sat}

We shall now solve the "$\omega$-regular tree satisfiability" problem by translating any $\CPDLp$ "expression" into an $\ICPDL$ "expression" which is equivalent (over trees), and reuse the prior result from \cite[Theorem 3.8]{DBLP:journals/jsyml/GollerLL09} solving the analogous problem for $\ICPDL$. 
Although the translation is exponential, the produced $\ICPDL$ formulas have low `intricacy' of "program intersections", which ensures that we do not incur in an exponential blowup for the complexity. This measure of intricacy is called "intersection width" and it is akin to the "conjunctive width" we have defined but for $\ICPDL$ "programs"
(or, rather, our "conjunctive width" definition is inspired by the "intersection width"). Let us first define this width.

\AP
The ""intersection width""  $\intro*\iwidth(\pi)$ of an $\ICPDL$ "program" $\pi$ is defined in \cite{DBLP:journals/jsyml/GollerLL09} as follows:
\begin{align*}
    \iwidth(a) = \iwidth(\bar a) = \iwidth(\phi?) = \iwidth(\epsilon) &\eqdef 1 &&\text{ for } a \in \Prog,
    \\
    \iwidth(\pi_1 \star \pi_2) &
    \eqdef
    \max \set{ \iwidth(\pi_1), \iwidth(\pi_2)} &&\text{ for $\star \in \set{\circ,\cup}$},\\
    \iwidth(\pi^*) &\eqdef \iwidth(\pi),
    \\
    \iwidth(\pi_1 \cap \pi_2) &\eqdef \iwidth(\pi_1) + \iwidth(\pi_2).
\end{align*}
\AP
This is generalized to any "expression" $e$ by defining $\intro*\Iwidth(e)$ to be the maximum "intersection width" of a "program" therein (or 1 if it contains no "program").

In order to define the translation $\CPDLp \to \ICPDL$ over trees, we will first need to learn how to `factor' a "program" into smaller subprograms to cover all possible ways of satisfying a path split into several subpaths. This will be the objective of our next \Cref{ssec:ksplit}. Once this definition is in place, we shall show the translation in \Cref{ssec:translation-ICPDL-tres}.

\subsubsection{Splitting path expressions}
\label{ssec:ksplit}

The $\ksplit$ of an "\ICPDL" "program" $\pi$ consists, intuitively, of all the ways of splitting a path conforming to $\pi$ into $k$ parts, grouped into sublanguages of $\pi$. For example, the $\ksplit[2]$ of $(ab)^*$ is $\set{((ab)^*,(ab)^*), ((ab)^*a,b(ab)^*)}$ 
\color{diegochange}
since for every pair of words $u,v$ we have $u \cdot v \in (ab)^*$ if, and only if, there is $(L,L') \in \ksplit[2]((ab)^*)$ such that $u \in L, v \in L'$.
\color{black}
More generally, the $\ksplit$ of a "program" $\pi$ is a set of $k$-tuples of "programs", in such a way that for every \AP""witnessing path"" $p$ in a tree $T$ of $\pi$ 
\color{diegochange}
("ie", a simple path whose first and last nodes $(v,v')$ are in $\dbracket{\pi}_T$) 
\color{black}
and every split $p = p_1 \dotsb p_k$ of that path into $k$ subpaths, we have that there is some tuple $\bar t$ from the $\ksplit$ of $\pi$ such that each $p_i$ is "witnessing path" for $\bar t[i]$. 
We shall now formalize these concepts.

\AP
We define the $\intro*\ksplit$ of an $\ICPDL$ "program" by induction on $k$, where
\begin{align*}
    \ksplit[1](\pi) \eqdef{} & \set{(\pi)} \quad \text{("ie", a singleton set with a $1$-tuple)}\\
    \ksplit[2](\star) \eqdef {}& \set{(\star,\epsilon), (\epsilon,\star)} \quad \text{ where $\star \in \set{a,\bar a, \epsilon, \phi?}$}\\
    \ksplit[2](\pi_1 \circ \pi_2) \eqdef {}&
        \set{(\pi'_1,\pi'_1{\circ} \pi_2) {\,:\,} (\pi'_1,\pi''_1) {\,\in\,} \ksplit[2](\pi_1) }
        \cup
        \set{(\pi_1 {\circ} \pi'_2,\pi''_2) {\,:\,} (\pi'_2,\pi''_2) {\,\in\,} \ksplit[2](\pi_2) }
    \\
    \ksplit[2](\pi_1 \cup \pi_2) \eqdef {}& \ksplit[2](\pi_1) \cup \ksplit[2](\pi_2)\\
    \ksplit[2](\pi_1 \cap \pi_2) \eqdef {}&
    \set{(\pi'_1 \cap \pi'_2, \pi''_1 \cap \pi''_2) : (\pi'_1,\pi''_1) \in \ksplit[2](\pi_1), (\pi'_2,\pi''_2) \in \ksplit[2](\pi_2)}\\
    \ksplit[2](\pi^*) \eqdef {}&
    \set{(\pi^* \circ \pi', \pi'' \circ \pi^*) : (\pi',\pi'') \in \ksplit[2](\pi)}\\
    \ksplit(\pi) \eqdef {}&
    \set{(\pi_1, \pi_2, \dotsc, \pi_k) : (\pi_1,\pi'_1) \in \ksplit[2](\pi), (\pi_2, \dotsc, \pi_k) \in \ksplit[(k{-}1)](\pi'_1)}
\end{align*}
\color{diegochange}
The intuition behind the $\ksplit[2](\pi_1 \circ \pi_2)$ is that for splitting a program $\pi_1 \circ \pi_2$ the `cut' in the path could occur in the part of the path matching the first part (splitting $\pi_1$ into $\pi'_1$ and $\pi''_1$) or in the part matching the second part (splitting $\pi_2$ into $\pi'_2$ and $\pi''_2$).
On the other hand, the definition of $\ksplit[2](\pi^*)$ follows the intuition that cutting a path matching $\pi^*$ into two would necessarily yield paths such that the first one matches $\pi^*$ plus a first part of $\pi$, and the second matches a second part of $\pi$ plus $\pi^*$.
\color{black}

\knowledgenewrobustcmd{\stardepth}{\cmdkl{\textit{sd}}}
\diego{Hay que transformar esta remark en lemma + proof}
\begin{remark}\hfill\label{rk:nr-elems-ksplit}
    \begin{enumerate}
        \item The "intersection width" of any "program" in $\ksplit(\pi)$ (for $\pi \in \ICPDL$) is bounded by $\Iwidth(\pi)$.\sidediego{este lema no se usa....}
        \item For every program $\pi'$ in a coordinate of an element of $\ksplit(\pi)$ we have $sd(\pi') \leq \stardepth(\pi)$, where $\AP\intro*\stardepth(\pi)$ stands for the \AP""star-depth"" of $\pi$ ("ie", the maximum number of nested $(~)^*$-"programs" in $\pi$).
        \item Each coordinate in an element of the $\ksplit[2](\pi)$ has size bounded by $\stardepth(\pi) \cdot \pdlsize{\pi}$.
        \item For every $k>2$, each coordinate in an element of the $\ksplit(\pi)$ has size bounded by $\stardepth(\pi)^{k-1} \cdot \pdlsize{\pi}$.
        \item For every $\ICPDL$ "program" $\pi$ and $k \in \Nat$, we have $|\ksplit(\pi)| \leq \pdlsize{\pi}^k$. \sidediego{es correcto esto?}
    \end{enumerate}
\end{remark}

\begin{lemma}\AP\label{lem:ksplit}
    Let $k\geq 1$, $\pi$ an $\ICPDL$ "program", and $T \in \tree(\SigmaN , \SigmaE)$. For every pair $(u,v) \in \dbracket{\pi}_T$, and partition into $k$ subpaths of the simple path $p=p_1 \dotsb p_k$ from $u$ to $v$ in $T$ 
    (where each $p_i$ is possibly an empty path starting and ending at the same node), 
    there is a tuple $(\pi_1,\dotsc, \pi_k) \in \ksplit(\pi)$ such that each $(src(p_i),tgt(p_i)) \in \dbracket{\pi_i}_T$ for every $i$,
    \color{diegochange}
    where $src(p_i)$ and $tgt(p_i)$ are the first and last nodes of $p_i$, respectively.
    \color{black}
    Conversely, for every tuple $(\pi_1,\dotsc, \pi_k) \in \ksplit(\pi)$ and $(u_i,u_{i+1}) \in \dbracket{\pi_i}_T$ (for $1 \leq i \leq k$), we have that $(u_1,u_{k+1}) \in \dbracket{\pi}_T$.
\end{lemma}
\color{diegochange}
\begin{proof}
    \sidediego{expandí la prueba}
    \proofcase{First statement}
    We proceed by induction on $k$.

    The case $k=1$ is trivial, so assume $k=2$. Assume that $w$ is an element which lies in the path between $u$ and $v$. 
    We show by induction on the "size@@pdl" of $\pi$ that there exists $(\pi_1,\pi_2) \in \ksplit[2](\pi)$ such that $(u,w) \in \dbracket{\pi_1}_T$ and $(w,v) \in \dbracket{\pi_2}_T$.
    \begin{itemize}
        \item For the base cases $\pi \in \set{a,\bar a, \epsilon, \phi?}$ we must have that either (i) $u=w$ or (ii) $w=v$. Hence, we have that either
        (i) $(u,w) \in \dbracket{\epsilon}_T$ and $(w,v) \in \dbracket{\pi}_T$ or 
        (ii) $(u,w) \in \dbracket{\pi}_T$ and $(w,v) \in \dbracket{\epsilon}_T$.
        Since both $(\pi,\epsilon)$ and $(\epsilon,\pi)$ are in $\ksplit[2](\pi)$, the statement follows.
        \item The inductive case $\pi = \pi_1 \circ \pi_2$. Consider the node $w'$ in the path between $u$ and $v$ such that $(u,w') \in \dbracket{\pi_1}_T$ and $(w',v) \in \dbracket{\pi_2}_T$. It may be that $w'$ appears before $w$, after $w$, or $w=w'$. Assume "wlog" that it appears before or at $w$, and hence that the path goes $u \underbrace{\leadsto}_{\dbracket{\pi_1}_T} w' \underbrace{\leadsto w \leadsto}_{\dbracket{\pi_2}_T} v$. By inductive hypothesis, there exists $(\pi_2',\pi''_2) \in \ksplit[2](\pi_2)$ such that $(w',w) \in \dbracket{\pi'_2}_T$ and $(w,v) \in \dbracket{\pi''_2}_T$, that is:
        \[
            u \underbrace{\leadsto}_{\dbracket{\pi_1}_T} 
            w' \underbrace{\overbrace{\leadsto}^{\dbracket{\pi'_2}_T} 
            w \overbrace{\leadsto}^{\dbracket{\pi''_2}}}_{\dbracket{\pi_2}_T} v.
        \]
        Therefore, we have $(u,w) \in \dbracket{\pi_1 \circ \pi'_2}_T$ and $(w,v)\in \dbracket{\pi''_2}_T$. Since by definition $(\pi_1 \circ \pi'_2,\pi''_2) \in \ksplit[2](\pi)$, the statement follows.
        \item  The inductive cases $\pi = \pi_1 \cup \pi_2$ and $\pi = \pi_1 \cap \pi_2$. Both cases follow directly by definition $\ksplit[2](\pi)$ and inductive hypothesis.
        \item The inductive case $\pi^*$. Let $w'_1,w'_2$ be nodes of the path $(u,v)$ such that $(w'_1,w'_2) \in \dbracket{\pi}_T$, $w'_1$ appears before (or at) $w$ and $w'_2$ appears after (or at) $w$. Hence, we have $u \underbrace{\leadsto}_{\dbracket{\pi^*}_T} w'_1 \underbrace{\leadsto w \leadsto}_{\dbracket{\pi}_T} w'_2 \underbrace{\leadsto}_{\dbracket{\pi^*}_T} v$. By inductive hypothesis there is $(\pi',\pi'') \in \ksplit[2](\pi)$ such that 
        \[
        u \underbrace{\leadsto}_{\dbracket{\pi^*}_T} w'_1 \underbrace{\overbrace{\leadsto}^{\dbracket{\pi'}_T} w \overbrace{\leadsto}^{\dbracket{\pi''}_T}}_{\dbracket{\pi}_T} w'_2 \underbrace{\leadsto}_{\dbracket{\pi^*}_T} v
        \]
        and thus $(u,w) \in \dbracket{\pi^* \circ \pi'}_T$ and $(w,v) \in \dbracket{\pi'' \circ \pi^*}_T$, 
        where $(\pi^* \circ \pi', \pi'' \circ \pi^*) \in \ksplit[2](\pi^*)$. This shows the statement.
    \end{itemize}
    
    \noindent
    The inductive case $k>2$ follows directly by the inductive definition of $\ksplit$.

    \smallskip

    \proofcase{Second statement}
    On the other hand, it is direct by construction of $\ksplit(\pi)$ that for each $(\pi_1, \dotsc, \pi_k) \in \ksplit(\pi)$ we have that $\dbracket{\pi_1 \circ \dotsc \circ \pi_k}_T \subseteq \dbracket{\pi}_T$, from which the second statement follows.
\end{proof}
\color{black}
\subsubsection{Translation into $\ICPDL$ on trees}
\label{ssec:translation-ICPDL-tres}
We shall next show a semantics-preserving translation from $\CPDLp$ to $\ICPDL$.

\begin{lemma}\AP\label{lem:CPDLp-to-ICDPL-on-trees}
    For every $\CPDLp$ "expression" $e$ in "unnested form" there exists an $\ICPDL$ "expression" $e'$ such that
    \begin{enumerate}
        \item $e$ and $e'$ are "equivalent" over trees 
        \item there is a $\polyfun(\pdlsize{e}^{\Cqsize{e}})$-time procedure to build $e'$ from $e$,
        \item $\Iwidth(e') \leq \Cqsize{e}$.
    \end{enumerate}
\end{lemma}
\knowledgenewrobustcmd{\pathinT}{\mathrel{\cmdkl{\leadsto_{T}}}}%
\begin{proof}
For clarity of the exposition we will actually show a translation from $\ICPDLp$  to $\ICPDL$.
We will first show how to translate a "conjunctive program" $C[x_s,x_t]$ all of whose  "programs" are $\ICPDL$ "programs" -- the extension to arbitrary formulas is the recursive application of such translation.
Consider any "tree" $T$ such that
\begin{enumerate}
    \item it contains $\vars(C)$ as vertices ("ie", $\vars(C) \subseteq \intro*\vertex{T}$) and
    \item every vertex in $\vertex{T}$ is the closest common ancestor of two $\vars(C)$-vertices.
\end{enumerate} 
For any two variables $x,y \in \vars(C)$, let us write \AP$x \intro*\pathinT y$ to denote the (unique) simple path from $x$ to $y$ in $T$.
For every "atom" $\alpha=\pi'(x,x')$ of $C$, let $n_\alpha$ denote the length of the path $x \pathinT x'$ ("ie", its number of edges).
Consider any function $f$ mapping every "atom" $\alpha=\pi'(x,x')$ in $C$ to an element of $\ksplit[n_\alpha](\pi')$\sidediego{treat case length 0?}, that is, $f$ is such that $f(\alpha) \in \ksplit[n_\alpha](\pi')$ for every $\alpha \in C$.
Consider also any function $\ell \colon \vertex{T} \times \vertex{T} \to 2^{\text{"programs"}}$ labeling edges,
such that $\ell(v,v')$ (for $\set{v,v'} \in \intro*\edges{T}$) is the smallest set containing, for each "atom" $\pi'(x,x') \in C$,
\begin{itemize}
    \item the "program" $f(\pi'(x,x')) [i]$ if $(v,v')$ is the $i$-th edge in the path $x \pathinT x'$,
    \item the "program" $\reverseof{(f(\pi'(x,x')) [i])}$ if $(v',v)$ is the $i$-th edge in the path $x \pathinT x'$.
\end{itemize}

\noindent
The intuition is that with $T$ we guess the `shape' in which $C$ will be mapped to the tree, and with $f,\ell$ we guess how the path expressions will be mapped for witnessing each "atom". 
See \Cref{fig:exampleTCconstruction} for an example.
\begin{figure}
    \includegraphics[width=\textwidth]{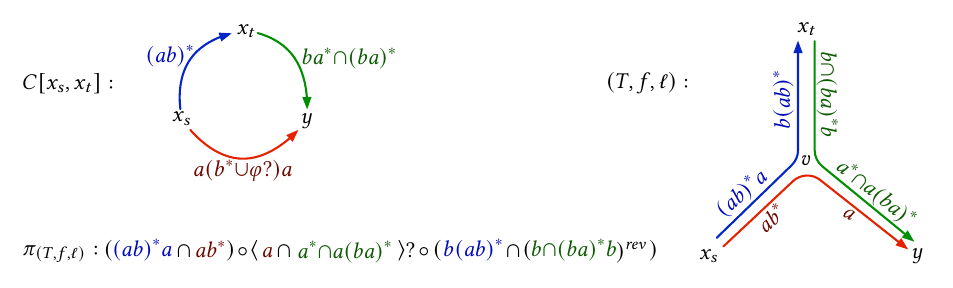}
    \caption{Example of a three-"atom" "conjunctive program" $C[x_s,x_t]$ and a tree $(T,f,\ell) \in \ShapesC$. The tree $T$ has 4 vertices $x_s,x_t,y,v$ and three (undirected) edges $\set{x_t,v}, \set{v,y}, \set{v,x_s}$.
    The edge labels of $\ell$ are implicitly shown, for example $\ell(x_t,v)$ contains $b \cap (ba)^*b$ while $\ell(v,x_t)$ contains $\reverseof{(b \cap (ba)^*b)} = \bar b \cap \bar b (\bar a \bar b)^*$. The "program" $\pi_{(T,f,\ell)}$ is the desired rewriting.}
    \label{fig:exampleTCconstruction}
\end{figure}

\noindent
\AP 
Let $\intro*\ShapesC$ be the set of all $(T,f,\ell)$ verifying the conditions above.
We will build an equivalent "program" as a union of programs, one for each $(T,f,\ell) \in \ShapesC$.
The idea is that for each $(T,f,\ell)$ the associated $\ICPDL$ "program" $\pi_{(T,f,\ell)}$ will be built by: (i) taking the intersection of labels for each edge
and (ii) concatenating the resulting labels in the path $x_s \pathinT x_t$, and (iii) transforming the remaining parts of the tree not lying on $x_s \pathinT x_t$ to "test programs" of the form `$\tup{ \dotsb }?$'. See \Cref{fig:exampleTCconstruction} for an example.

\color{diegochange}
We now give the formal procedure to build the program, which parallels the well-known ``rolling-up'' technique from description logics ("aka" ``concept rolling-up'' or ``tree-shaped query rewriting'') \cite{HorrocksT00} by iteratively collapsing side branches into existential tests. We shall use some `\AP""processed""' flag to know which vertices of the tree $T$ have already been visited, as well as some `\AP""test labels""' on vertices, which are $\ICPDL$ formulas.
\color{black}
Initially, every vertex of $T$ is \reintro{unprocessed} ("ie", not marked as "processed") and has an empty "test label". For a set $P=\set{\pi_1,\dotsc, \pi_\ell}$ of "programs" we shall write `$\bigcap P$' as short for $\pi_1 \cap \dotsb \cap \pi_\ell$ (note that its semantics is independent of the order).

Let us take an \AP""unprocessed leaf"" $z$ of $T$ ("ie", an "unprocessed" vertex having exactly one "unprocessed" neighbor) which is not on the path $x_s \pathinT x_t$. Suppose $z$ has "test label" $\set{\phi_1, \dotsc, \phi_s}$. Let $z'$ be its (sole) "unprocessed" neighbor, we then add a "test label" of $\tup{(\bigcap \ell(z',z)) \circ (\bigwedge_{1 \leq j \leq  s}\phi_j)?}$ to $z'$ and we mark $z$ as "processed" -- if the "test label" is empty, we just add $\tup{\bigcap \ell(z',z)}$.
We iterate this procedure until no more vertices can be marked as "processed", in other words, only the vertices lying on $x_s \pathinT x_t$ remain "unprocessed".

Let $x_s=x_0 \to x_1 \to \dotsb \to x_n=x_t$ be the vertices of $x_s \pathinT x_t$.
\knowledgenewrobustcmd{\topp}{\cmdkl{\top}}
Let $\psi_i$ be the conjunction of all the "test labels" of $x_i$ (or `$\topp$' if the "test label" is empty, where $\topp$ is a fixed tautology  such as, "eg", $\AP\intro*\topp \eqdef p \lor \lnot p$ for $p \in \Prop$).
We then produce the $\ICPDL$ "program":
\[
    \pi_{(T,f,\ell)} \eqdef \psi_0? \circ \left(\bigcap \ell(x_0,x_1)\right) \circ \psi_1? \circ \left(\bigcap \ell(x_1,x_2)\right) \circ \dotsb \circ \left(\bigcap \ell(x_{n-1},x_n)\right) \circ \psi_n?
\]
Finally, we define the translation of $C[x_s,x_t]$ as $\pi_{C[x_s,x_t]} \eqdef \bigcup_{(T,f,\ell) \in \ShapesC}\pi_{(T,f,\ell)}$.

We first show that indeed $\pi_{C[x_s,x_t]}$ is equivalent to $C[x_s,x_t]$.
\begin{claim}\AP
    $\pi_{C[x_s,x_t]}$ and $C[x_s,x_t]$ are "equivalent" over trees.
\end{claim}
\begin{nestedproof}
    \diego{probablemente hay que expandir esta prueba con más detalles...}
    Given a tree $K$ such that $(u,u') \in \dbracket{C[x_s,x_t]}_K$, consider any "$C$-satisfying assignment" $\nu$ such that $\nu(x_s)=u$ and $\nu(x_t)=u'$, and the tree $T$ that $\nu$ induces on $K$. By \Cref{lem:ksplit}, for every "atom" $\alpha = \pi'(x,x') \in C$ and associated path $x \pathinT x' = v_1 \to \dotsb \to v_{k+1}$ (where $v_1=x$ and $v_{k+1}=x'$), there must be some $(\pi_1, \dotsc, \pi_k) \in \ksplit(\pi')$ such that $(\nu(v_i),\nu(v_{i+1})) \in \dbracket{\pi_i}_K$ for every $i \in [k]$. From these witnesses we can build $(T,f,\ell)$ which satisfies the conditions for being in $\ShapesC$, from which it follows that $(u,u') \in \dbracket{\pi_{(T,f,\ell)}}_K$.

    Conversely, if $(u,u') \in \dbracket{\pi_{(T,f,\ell)}}_K$ for some $(T,f,\ell) \in  \ShapesC$, consider the function $\mu \colon \vertex{T} \to \dom{K}$ witnessing this fact (in particular $\mu(x_s) = u$ and $\mu(x_t) = u'$). By construction of $\pi_{(T,f,\ell)}$ combined with \Cref{lem:ksplit}, it follows that $(\mu(x),\mu(x')) \in \dbracket{\pi'}$ for every $\pi'(x,x') \in C$. In other words, $\mu$ is a "$C$-satisfying assignment", and thus $(\mu(x_s),\mu(x_t)) = (u,u') \in \dbracket{C[x_s,x_t]}_K$.
\end{nestedproof}

\noindent
We note by construction that the "intersection width" is bounded by the "conjunctive width".
\begin{claim}\AP\label{cl:IW-CW-translation}
    $\Iwidth(\pi_{C[x_s,x_t]}) \leq \Cqsize{C[x_s,x_t]}$.
\end{claim}
\begin{nestedproof}
    This is because the number of "program intersections" is bounded by the number of "atoms" in $C$.
\end{nestedproof}

\noindent
Further, the program $\pi_{C[x_s,x_t]}$ can be built in polynomial time.
\begin{claim}\AP
    $\pi_{C[x_s,x_t]}$ can be produced in $\polyfun(\pdlsize{C[x_s,x_t]}^{\Cqsize{C[x_s,x_t]}})$.
\end{claim}
\begin{nestedproof}
    First let us bound the size of $\ShapesC$.
    \color{diegochange}
    First note that the trees of $\ShapesC$ contain all variables $\vars(C)$ as vertices, plus the closest common ancestors of $\vars(C)$-vertices. In the worst case we have a complete binary tree whose leaves are $\vars(C)$-vertices, in which case the tree has $k= 2 |\vars(C)|-1$ vertices (any other case contains fewer closest common ancestors).
    \color{black}
    Hence, the number of vertices of each tree of $\ShapesC$ is bounded by $k = 2 |\vars(C)|$ and thus there are no more than $k^{k-2} \leq \pdlsize{C[x_s,x_t]}^{O(\Cqsize{C[x_s,x_t]})} $ distinct trees by Cayley's formula for counting trees.
    For each such tree there are no more than 
    \begin{align*}
        \prod_{\pi(x,y) \in C}|\ksplit[k](\pi)| &\leq 
        \max_{\pi(x,y) \in C}|\ksplit[k](\pi)|^{|C|} \\
        &\leq
        \max_{\pi(x,y) \in C}\pdlsize{\pi}^{k|C|} \tag{by \Cref{rk:nr-elems-ksplit}}\\
        &\leq 
        \pdlsize{C[x_s,x_t]}^{O(\Cqsize{C[x_s,x_t]})}    
    \end{align*}
    different possible $f$-functions to choose from.
    Finally, once a tree $T$ and a function $f$ are fixed, the edge label $\ell$ is uniquely determined.
    The elements of $\ShapesC$ can be enumerated one by one, and for each $(T,f,\ell) \in \ShapesC$ one can produce the translation $\pi_{(T,f,\ell)}$ in polynomial time.
\end{nestedproof}

\noindent
Finally, the translation for an arbitrary $\ICPDLp$ "program" goes by iteratively replacing every "conjunctive program" $C[x_s,x_t]$ such that $C$ contains no "conjunctive programs" with its equivalent $\ICPDL$ program.

\begin{claim}\AP
    For every $\ICPDLp$ formula $\phi$ in "unnested form", the translation $\phi'$ can be produced in $\polyfun(\pdlsize{\phi}^{\Cqsize{\phi}})$ and $\Iwidth(\phi') \leq \Cqsize{\phi}$.
\end{claim}
\begin{nestedproof}
    This is simply because the "nesting depth" of "conjunctive programs" is bounded~(by 2). Hence, iterating the translation only incurs a polynomial blowup. $\Iwidth(\phi') \leq \Cqsize{\phi}$ follows immediately from \Cref{cl:IW-CW-translation}.
\end{nestedproof}

\noindent
This concludes the proof of \Cref{lem:CPDLp-to-ICDPL-on-trees}.
\end{proof}

\begin{thmC}[{\cite[Theorem~3.8]{DBLP:journals/jsyml/GollerLL09}}]
\label{thm:goller:ICPDL-on-trees}
    For a "TWAPTA" $\+T$ and an $\ICPDL$ formula $\phi$, we can decide in $\expfun(\sizeStates{\+T} + \indexAut{\+T} + \pdlsize{\phi}^{\Iwidth(\phi)}) \cdot \polyfun(\sizeTrans{\+T})$ whether there exists some $T \in \langTWAPTA{\+T}$ such that $\epsilon \in \dbracket{\phi}_T$. 
\end{thmC}

\begin{proposition}\AP\label{prop:omega-sat-pb-2Exp}
    The "$\omega$-regular tree satisfiability" problem for "unnested@unnested form" $\CPDLp$ is in "2ExpTime". 
    It is in "ExpTime" for any subclass of "unnested@unnested form" $\CPDLp$ formulas with bounded "conjunctive width".
    More precisely, for an input $(\+T,\phi)$ it can be decided in
    $\expfun(\sizeStates{\+T} + \indexAut{\+T} + \pdlsize{\phi}^{\polyfun(\Iwidth(\phi))}) \cdot \polyfun(\sizeTrans{\+T})$ time.
\end{proposition}
\begin{proof}
    Given a "TWAPTA" $\+T$ and a $\CPDLp$ formula $\phi$ in "unnested form", we produce the equivalent $\ICPDL$ formula $\phi'$ given by \Cref{lem:CPDLp-to-ICDPL-on-trees}.
    Observe that the produced formula $\phi'$ is of $\polyfun(\pdlsize{\phi}^{\Cqsize{\phi}})$ "size@@pdl" and "intersection width" bounded by $\Cqsize{\phi}$. Hence, $\pdlsize{\phi'}^{\Iwidth(\phi')} \leq \polyfun(\pdlsize{\phi}^{\Cqsize{\phi}})^{\Cqsize{\phi}} \leq \pdlsize{\phi}^{\polyfun(\Cqsize{\phi})}$ is singly exponential w.r.t.\ $\pdlsize{\phi}$ and by \Cref{thm:goller:ICPDL-on-trees} we obtain a "2ExpTime" procedure.

    If the "conjunctive width" of $\phi$ is bounded by a constant, observe that $\phi'$ can be produced in polynomial time (and in particular it is of polynomial "size@@pdl") and it has constant "intersection width". Thus, in light of \Cref{thm:goller:ICPDL-on-trees} we obtain an "ExpTime" procedure.
\end{proof}

\subsection{Main Decidability Result}
\label{ssec:main-sat-result}
As a direct consequence of all the previous developments, we obtain the main result of this section:
\begin{thm}\label{thm:sat-cpdlp}
    \hfill
    \begin{enumerate}
        \item \label{it:satgeneral} $\UCPDLp$-"satisfiability@satisfiability problem" is "2ExpTime"-complete,
        \item \label{it:satCW}  For every $k \geq 1$, $\set{ \phi \in \UCPDLp \mid \Cqsize{\phi} \leq k}$-"satisfiability@satisfiability problem" is  "ExpTime"-complete. 
    \end{enumerate}
\end{thm}
\begin{proof}
    The upper-bounds are a consequence of
    \begin{enumerate}[label=(\roman*),ref=\roman*,align=left,labelwidth=!,labelsep=.5em,leftmargin=!]
        \item \Cref{prop:polyredux-simpleform}, reducing to the satisfiability of "unnested@unnested form" $\CPDLp$ over "Kripke structures",
        \item \Cref{cor:treewidth-k-model-property} reducing to satisfiability over "Kripke structures" of "tree-width" $k$,
        \item \label{it:redux:omega} \Cref{lem:redux-omega-tree-sat}, reducing to the "$\omega$-regular tree satisfiability" problem, and
        \item \Cref{prop:omega-sat-pb-2Exp}, solving the "$\omega$-regular tree satisfiability" problem.
    \end{enumerate}
    The first 4 items are polynomial-time reductions. In item \ref{it:redux:omega} we obtain a polynomial size "formula" $\phi$ and an exponential-sized "TWAPTA" $\+T$ such that $\sizeStates{\+T}$ is polynomial. In all reductions the "conjunctive width" and "unnesting@unnested form" of the formula is preserved.
    In light of the bounds of the last item (\Cref{prop:omega-sat-pb-2Exp}), this yields a double exponential time procedure, which becomes single exponential if the  "conjunctive width" is bounded.

The lower bound of \cref{it:satgeneral} follows from the known "2ExpTime" lower bound of $\ICPDL$ satisfiability \cite[Theorem~2]{DBLP:journals/jsyml/LangeL05} combined with \Cref{thm:ICPDL_equals_TW1_equals_TW2}. The lower bound of \cref{it:satCW} follows from the "ExpTime" lower bound of $\PDL$ \cite{DBLP:journals/jcss/FischerL79}.
\end{proof}

\noindent
Due to the following \cref{rem:ICPDL-to-UCPDL}, \cref{it:satCW} of \Cref{thm:sat-cpdlp} can be seen as a generalization of the result \cite[Theorem~4.8]{DBLP:journals/jsyml/GollerLL09} stating that the satisfiability problem for $\ICPDL$ "formulas" of bounded "intersection width"  is in "ExpTime".
\begin{remark}\label{rem:ICPDL-to-UCPDL}
    There is a polynomial time translation of $\ICPDL$ "formulas" of "intersection width" $k$ to $\UCPDLg{\Tw[2]}$ "formulas" of "conjunctive width" $k$ ("eg", the translation $\trThree$ from \Cref{ssec:ICPDL-CPDL-equivalence}).
\end{remark}

%% file: UNTC-intro.tex
\section{Unary Negation First-Order Logic with Transitive Closure}
\label{sec:UNTC}
\AP
We define $\intro*\UNTC$, the unary-negation fragment of first-order logic with unary transitive closure (without parameters), as the following language, where $\intro*\TC$ denotes transitive closure:
\begin{align*}
\varphi &\eqqdef R(\bar x)  \mid 
    x=y                     \mid
    \varphi\land\varphi     \mid
    \varphi\lor\varphi      \mid
    \exists x.\varphi       \mid
    \lnot\varphi(x)			\mid
    [\TC_{u,v}\varphi(u,v)](x,y)
\end{align*}

\noindent
where $R\in\Rels$, $x,y,u,v\in\Vars$, $\bar x$ is a tuple of variables of $\Vars$, $\phi(x)$ is a formula with at most one free variable $x$ and $\varphi(u,v)$ is a formula with free variables $\set{u,v}$. The semantics is defined over "$\Rels$-structures" and it is the usual for first-order logics. In particular, for any "$\Rels$-structure" $K$ and $a,b\in\dom{K}$ we have $K\models [\TC_{u,v}\varphi(u,v)](x,y)[x\mapsto a,y\mapsto b]$ iff there is $n\geq 1$ and $c_1,\dots,c_n\in\dom{K}$ such that $c_1=a$, $c_n=b$ and for all $i=1,\dots,n-1$ we have $K\models\phi[u\mapsto c_i,v\mapsto c_{i+1}]$.

\color{santichange}

As an example, consider the $\UNTC$-formula
\[
\varphi(x)\eqdef \exists y.(\lnot q(y)\land [\TC_{u,v}\exists z.\bigl(a(u,z)\land b(z,v)\land R(u,z,v)\bigr)](x,y)).
\]
This formula says that one can reach some "world" satisfying $\lnot q$ by following a finite path whose individual steps are witnessed by an intermediate "world" $z$ such that $a(u,z)$, $b(z,v)$, and the "r-atom" $R(u,z,v)$ hold simultaneously. Thus, each step is not given by a single "atomic program", but by a small existentially defined pattern. In $\UCPDLp$, this same step is naturally represented by a "conjunctive program": if
$$
C\eqdef \set{a(x_s,z),\, b(z,x_t),\, R(x_s,z,x_t)},
$$
then the corresponding $\UCPDLp$-formula is
$$
\tup{(C[x_s,x_t])^*\circ (\lnot q)?}.
$$
This correspondence is not accidental: in Section~\ref{sec:UCPDLplus-equi-UNTC} we will show that $\UCPDLp$ and $\UNTC$ are in fact equi-expressive. This example may be kept in mind when reading the translation below: the transitive closure in the first formula corresponds to the iteration operator $(-)^*$, while the existentially defined binary step corresponds to a "conjunctive program".
\color{black}

\AP
We define the ""size@@untc"" of a formula $\phi$ of $\UNTC$, denoted $\intro*\untcsize{\phi}$, as follows:
\begin{samepage}
\begingroup\allowdisplaybreaks[0]
\begin{align*}
\untcsize{R(x_1,\dots,x_n)}  &\eqdef n,\\
\untcsize{x=y} &\eqdef 1, \\
\untcsize{\varphi\star\psi} &\eqdef  \untcsize{\varphi}+\untcsize{\psi}&&\text{for $\star\in\{\land,\lor\}$},\\
\untcsize{\exists x.\varphi} &\eqdef  1+\untcsize{\varphi},\\
\untcsize{\lnot\varphi(x)} &\eqdef  1+\untcsize{\varphi},\\
\untcsize{[\TC_{u,v}\varphi(u,v)](x,y)} &\eqdef  1+\untcsize{\varphi}.
\end{align*}
\endgroup
\end{samepage}

%% file: UNTCequivUCPDL.tex
\diego{Habría que agregar un párrafo comentando el trabajo \cite{jung2018querying} al principio de esta sección, informalmente qué es lo que demuestran y cómo se compara con nuestro resultado y nuestra técnica para demostrarlo. Quizás también en la introducción una versión condensada.}

\subsection{\texorpdfstring{$\UCPDLp$}{UCPDL⁺} and \texorpdfstring{$\UNTC$}{UNFO*} are Equi-expressive}\label{sec:UCPDLplus-equi-UNTC}

In this subsection we show that $\UCPDLp$ and $\UNTC$ have the same expressive power. We begin with the translation from $\UNTC$ to $\UCPDLp$. The main difficulty in this direction is that transitive closure naturally involves formulas with two free variables, while in $\UCPDLp$ one must distinguish between "formulas" and "programs". For this reason, we first put $\UNTC$-formulas with at most two free variables into a suitable normal form, and only afterwards translate them into $\UCPDLp$.

\begin{proposition}\AP\label{prop:UNTC-ECPDL}
    \AP
    $\UNTC \lleq \UCPDLp$ via an exponential-time translation $\intro*\transUNTCECPDL\colon \UNTC \to \UCPDLp$ such that  $\Cqsize{\transUNTCECPDL(\phi)} \leq \untcsize{\phi}$ for every $\phi \in \UNTC$.
\end{proposition}
\begin{proof}
As in \cite{jung2018querying} we first define a normal form for some $\UNTC$-formulas.
Our normal form is more convoluted than the one in \cite{jung2018querying} (and also the one in \cite{segoufin2013unary}) since we have to deal with formulas of $\UNTC$ with at most two free variables instead of just one as in the referenced works. The reason is the transitive closure, a formula of two free variables, where, furthermore, the order of these two variables matters. 
We first show that any $\UNTC$-formula with at most two free variables can be transformed into an equivalent one in normal form of "size@@untc" at most exponential. Then we show that there is a polynomial time translation from $\UNTC$-formulas in normal form and with at most two free variables to $\UCPDLp$-formulas or -programs, witnessing $\UNTC \lleq \UCPDLp$.

For a set of variables $X$ and a $\UNTC$-formula $\phi$, we use $\phi(X)$ to denote that the set of free variables of $\phi$ is exactly $X$ and we use $\phi(\subseteq X)$ to denote that the free variables of $\phi$ are among those in $X$.

\paragraph{Normal form} We define a normal form of $\UNTC$-formulas with at most two free variables that only allows for disjunctions of formulas $\phi_1$ and $\phi_2$ such that the cardinality of the set of free variables in $\phi_1$ or in $\phi_2$ is at most 2. 
For example, the following $\UNTC$-formula with free variable $x$ 
\[
\phi=\TC_{u,v}\left[\exists z_1,z_2.\left(R(z_1,z_2,u,v)\vee S(z_1,z_2,u,v)\right)\land T(u,v)\right](x,x)
\] 
is not in normal form because the disjunction is between formulas having altogether more than two free variables (namely, $z_1,z_2,u,v$), but the equivalent formula
\[
\TC_{u,v}\big[ \big(\exists z_1,z_2.(R(z_1,z_2,u,v)\land T(u,v))\big) \vee \big(\exists z_1,z_2.(S(z_1,z_2,u,v)\land T(u,v))\big) \big](x,x)
\] 
is in normal form because each disjunct has free variables only $u$ and $v$.

Formally, the normal form of $\UNTC$ with at most two free variables is defined through the following grammar:
\begin{align*}
    \phi(\subseteq\{x,y\})   \eqqdef{} & \varphi(\subseteq\{x,y\})\vee\varphi(\subseteq\{x,y\})
    \\&
                \mid\lnot\varphi(\subseteq\{x\}) 
    \\&
                \mid  [\TC_{u,v}\phi(\{u,v\})](x,y) 
    \\&
                \mid  [\TC_{u,v}\phi(\{u,v\})](x,x) 
    \\&
                \mid \exists z_1,\dots,z_m.\psi, 
\end{align*}
where $m\geq 0$, $z_i\notin\{x,y\}$ for all $i$, and $\psi$ is a conjunction of expressions of one of these forms:
\begin{samepage}
\begin{enumerate}[label=\roman*.,ref=\roman*,align=left,labelwidth=!,labelsep=.5em,leftmargin=!]
    \item\label{item:NF_conj:1} $R(w_1,\dots w_n)$ for $R\in\Rels$,
    \item $w_1=w_2$,
    \item\label{item:NF_conj:3} $\varphi(\subseteq\{w_1,w_2\})$,
\end{enumerate}
\end{samepage}
where $w_1,w_2,\dots,w_n$ are among $x,y,z_1,\dots,z_m$. We allow for $m=0$ and in such case $\exists z_1,\dots,z_m.\psi$ is simply $\psi$. In the above grammar the non-terminal is the symbol $\phi$; the use of $\phi(\subseteq\{x,y\})$, $\phi(\subseteq\{x\})$, and $\phi(\{u,v\})$ is only to emphasise the properties of the set of free variables of $\phi$ that must hold in order to apply the grammatical rule. 
\AP
Formulas $\psi$ used in the grammar above are called ""NF-conjunctions"" and the expressions \ref{item:NF_conj:1}-\ref{item:NF_conj:3} above are called ""NF-atoms"" of $\psi$.

\AP
In analogy to the "conjunctive width" for $\UCPDLp$ "programs" we introduce the ""NF-conjunctive width"" $\intro*\cqsizenf{\psi}$ of a "NF-conjunction" $\psi$, defined as follows:
\begin{align*}
    \cqsizenf{\psi} \eqdef &\sum \set{\arity(R) : R(\bar w)  \text{ is an "NF-atom" of } \psi} + {}
    \\&
    |\set{w_1=w_2 : w_1=w_2\text{ is an "NF-atom" of } \psi }| + {}
    \\&
    \sum \set{\cqsizenf\phi : \phi\text{ is an "NF-atom" of } \psi }.
\end{align*}
We define $\intro*\Cqsizenf{\phi}$
for any $\UNTC$-formula $\phi$ in normal form to be the maximum "NF-conjunctive width" of a "NF-conjunction" therein (or 1 if it contains no "NF-conjunctions").

\color{santichange}
Disjunctions of formulas with more than two free variables can be eliminated, at the cost of a possible exponential blow-up in the length. More precisely, after renaming bound variables apart whenever needed, every forbidden disjunction (namely, a disjunction with at least three free variables) can be transformed into an equivalent one by repeatedly applying the following rewritings:
\begin{enumerate}
\item $(\psi_1\vee\psi_2)\land\psi_3 \rightsquigarrow (\psi_1\land\psi_3) \vee (\psi_2\land\psi_3)$;
\item $(\exists \bar x.\psi_1)\land\psi_2\rightsquigarrow\exists \bar x.(\psi_1\land\psi_2)$, assuming that no variable from $\bar x$ occurs free in $\psi_2$;
\item $\exists \bar x.(\psi_1\vee\psi_2)\rightsquigarrow(\exists \bar x.\psi_1)\vee(\exists \bar x.\psi_2)$; and
\item $\exists \bar x.(\exists \bar y.\psi)\rightsquigarrow\exists \bar x\bar y.\psi$.
\end{enumerate}

\noindent
More concretely, the normalization is applied bottom-up, starting from an innermost forbidden disjunction. Given a subformula $\psi$ of the form $\psi_1\vee\psi_2$ such that the set of free variables of $\psi_1$ or $\psi_2$ has cardinality greater than $2$, we first rename bound variables apart whenever needed, so that no variable capture occurs in the subsequent rewritings. We then view the smallest surrounding formula built from $\psi$ using only conjunctions and existential quantifiers, and repeatedly apply Rules~(1)--(4) so as to push conjunctions and existential quantifiers across the disjunction until the latter is exposed as a disjunction of formulas each having at most two free variables. Intuitively, Rule~(1) distributes conjunction over disjunction, Rule~(2) pulls existential quantifiers outward across conjunctions when they do not bind variables free in the other conjunct, Rule~(3) distributes existential quantification over disjunction, and Rule~(4) merges consecutive existential blocks. Repeating this procedure for every forbidden disjunction yields an equivalent formula in normal form. Observe that this process may duplicate subformulas, and therefore may cause an exponential blow-up in the size of the formula.

Indeed, the translation from $\UNTC$ to $\UCPDLp$ may incur an exponential blow-up already on formulas of $\UNFO$, that is, even without using transitive closure. A simple example is given by the family
$$
\phi_n(x)\eqdef \exists y\,\exists z_1\dots\exists z_n.\bigwedge_{i=1}^n(R_i(x,z_i)\vee S_i(y,z_i)),
$$
where $R_1,\dots,R_n,S_1,\dots,S_n\in\Rels$ are pairwise distinct binary relation symbols. Intuitively, the free variable $x$ is fixed, while $y$ and $z_1,\dots,z_n$ are existentially quantified; for each $i\leq n$, one must choose between the atom $R_i(x,z_i)$ and the atom $S_i(y,z_i)$. A corresponding $\UCPDLp$-formula can be written as
$$
\bigvee_{\sigma\in\set{R,S}^n}\tup{C_\sigma[x_s,x_t]},
$$
where, for each $\sigma\in\set{R,S}^n$, we let
$$
C_\sigma \eqdef \set{\rho_1^\sigma,\dots,\rho_n^\sigma},
\qquad
\rho_i^\sigma \eqdef
\begin{cases}
R_i(x_s,z_i) & \text{if }\sigma(i)=R,\\
S_i(x_t,z_i) & \text{if }\sigma(i)=S.
\end{cases}
$$
Here $x_s$ plays the role of the free variable $x$, $x_t$ plays the role of the existential witness $y$, and the variables $z_1,\dots,z_n$ are the internal existential variables of the corresponding "conjunctive programs". Thus each choice of one disjunct for every $i\leq n$ gives rise to one "conjunctive program" $C_\sigma[x_s,x_t]$, so that in general one obtains a disjunction of $2^n$ different "conjunctive programs". As previously mentioned, this exponential blow-up is already present in the $\UNFO$-fragment, and it was already observed in \cite[proof of Lemma 4.1]{segoufin2013unary}, where it is described as being ``reminiscent of the transformation of propositional formulas into disjunctive normal form''.

\color{black}

The next lemma isolates the normalization step. It shows that every $\UNTC$-formula with at most two free variables can be transformed into an equivalent formula in normal form, while keeping the corresponding width measure under control. This will be crucial later, when we translate normal-form formulas into $\UCPDLp$ without losing the desired complexity bounds. Notice that while the "size@@untc" of the normal form $\tilde \phi$ obtained from $\phi$ can be of exponential "size@@untc" with respect to $\untcsize\phi$, the "NF-conjunctive width" of $\tilde\phi$ remains linear. 

\begin{lemma}\AP
Let $\phi$ be a $\UNTC$-formula of at most two free variables. There is $\tilde\phi$ in normal form such that $\phi$ is equivalent to $\tilde\phi$ and $\Cqsizenf{\tilde \phi}\leq \untcsize{\phi}$.
\end{lemma}
\begin{proof}
We proceed by structural induction on $\phi$.

\proofsubcase{$\phi=R(x,y)$, for $R\in\sigma$}
We have $\tilde\phi=R(x,y)$ ($=\exists\emptyset.R(x,y)$) and $\Cqsizenf{\tilde \phi}=2 = \untcsize{R(x,y)}$.

\proofsubcase{$\phi= x=y$}
We reason as above.

\proofsubcase{$\phi=\phi_1\land\phi_2$} 
Let $\tilde\phi_1$ and $\tilde\phi_2$ be formulas in normal form equivalent to $\phi_1$ and $\phi_2$ respectively. By inductive hypothesis, $\Cqsizenf{\tilde \phi_i}\leq \untcsize{\phi_i}$ for $i=1,2$. Now, the formula $\tilde\phi=\tilde\phi_1\land \tilde\phi_2$ (=$\exists\emptyset(\tilde\phi_1\land \tilde\phi_2)$) is in normal form and satisfies that $\Cqsizenf{\tilde \phi}\leq\Cqsizenf{\tilde \phi_1}+\Cqsizenf{\tilde \phi_2}\leq\untcsize{\phi_1}+\untcsize{\phi_2}\leq\untcsize{\phi}$. The case for disjunction is analogous.

\proofsubcase{$\phi=\lnot\varphi'$ and $\phi =[\TC_{u,v}\varphi'(u,v)](x,y)$} By inductive hypothesis there is a formula $\tilde\varphi'$ in normal form such that $\varphi'$ is equivalent to $\tilde\varphi'$ and $\Cqsizenf{\tilde\varphi'}\leq \untcsize{\varphi'}$. We can then define $\tilde\phi\eqdef\lnot\tilde\varphi'$ and $\tilde\phi\eqdef[\TC_{u,v}\tilde\varphi'(u,v)](x,y)$, both of which are in normal form and satisfy the desired properties.

\proofsubcase{$\phi=\exists \bar x. \rho$} Consider the syntactic tree $T_\rho$ of $\rho$ and the set $X$ of nodes of $T_\rho$ that encode a subformula of $\rho$ with at most 2 free variables or a leaf of the form $R(x_1,\dots,x_n)$ for some $R\in\sigma$. Take the set $W$ of maximal points of $X$, that is, the set of nodes $w\in X$ such that there is no $w'\in X$ with $w'$ a strict predecessor of $w$ in $T_\rho$. Suppose $W=\{w_1,\dots,w_n\}$, let $\rho_i$ be the subformula of $\rho$ represented in $T_\rho$ by the node $w_i$. If $\rho_i$ is a formula of two free variables different from an atom $R(x_1,x_2)$ and different from an equality of the form $x_1=x_2$, let $\tilde\rho_i$ be a formula in normal form equivalent to $\rho_i$ such that $\Cqsizenf{\tilde \rho_i}\leq \untcsize{\rho_i}$ whose existence follows from the inductive hypothesis. If $\rho_i$ is an atom $R(x_1,\dots,x_n)$, define $\tilde\rho_i \eqdef \rho_i$ (and so in this case $\untcsize{\tilde\rho_i}=n$). If $\rho_i$ is an equality $x_1=x_2$, define $\tilde\rho_i \eqdef \rho_i$ (and so in this case $\untcsize{\tilde\rho_i}=1$). We use rules (1) -- (4) to convert $\rho$ into an equivalent formula of the form 
$\bigvee_{j\leq m}\exists\bar x\bar y_j. \bigwedge_{i\leq k}\rho_{r(j,i)}$ where 
$r(j,i)\neq r(j,i')$ for $i\neq i'$. 
Finally we replace $\rho_{r(j,i)}$ by $\tilde\rho_{r(j,i)}$ and obtain $\tilde\phi\eqdef\bigvee_{j\leq m}\exists\bar x\bar y_j. \bigwedge_{i\leq k}\tilde\rho_{r(j,i)}$ that verifies that (i) it is in normal form, (ii) it is equivalent to $\phi$, and (iii) for each $j\leq m$, $\bigwedge_{i\leq k}\tilde\rho_{r(j,i)}$ is an "NF-conjunction". Then
$\Cqsizenf{\tilde \phi}=\max\{\cqsizenf{\bigwedge_{i\leq k}\tilde\rho_{r(j,i)}} \colon j\leq m\}$. By the definition of "size@@untc" and the considerations above regarding $\Cqsizenf{\tilde \rho_i}$ when $\rho_i$ is not an equality or an atom, and the ones regarding $\untcsize{\tilde\rho_i}$ when $\rho_i$ is an equality or an atom, we conclude $\Cqsizenf{\tilde \phi}\leq \sum_{1\leq i\leq n}\Cqsizenf{\tilde\rho_i}\leq\untcsize{\phi}$. 
\end{proof}

\noindent
Once formulas are in normal form, the translation to $\UCPDLp$ becomes essentially syntax-directed. The key point is that formulas with one free variable will be translated into $\UCPDLp$-"formulas", while formulas with two ordered free variables will be translated into $\UCPDLp$-"programs". 

\paragraph{Translation}

We define a translation $\transUNTCECPDL$ from $\UNTC$-formulas with at most one free variable and in normal form to $\UCPDLp$-formulas, and a family of translations $\transUNTCECPDL^{(x,y)}$ indexed by an ordered pair of variables from $\UNTC$-formulas with two free variables $x$ and $y$ and in normal form to $\UCPDLp$-programs such that for any $\Rels$-structure $K$ and $u,v\in\dom{K}$:
\begin{enumerate}
\item if $\varphi$ is a $\UNTC$-sentence in normal form then $K\models\phi$ iff $K,u\models\transUNTCECPDL(\phi)$;
\item if $\varphi$ is a $\UNTC$-formula in normal form with unique free variable $x$ then $K\models\phi[x\mapsto u]$ iff $K,u\models\transUNTCECPDL(\phi)$; and
\item if $\varphi$ is a $\UNTC$-formula in normal form with exactly two free variables, $x$ and $y$ ($x\neq y$), then $K\models\phi[x\mapsto u,y\mapsto v]$ iff $K,u,v\models\transUNTCECPDL^{(x,y)}(\phi)$.
\end{enumerate}

We use $\pi^+$ as a short for $\pi\circ\pi^*$. The translations are mutually recursively defined as follows:
\begin{align*}
    \transUNTCECPDL(\lnot \phi(\subseteq\{x\})) &\eqdef \lnot\transUNTCECPDL(\phi),
    \\
    \transUNTCECPDL(\phi_1(\subseteq\{x\})\vee\phi_2(\subseteq\{x\}))   &\eqdef \transUNTCECPDL(\phi_1) \vee \transUNTCECPDL(\phi_2),
    \\
    \transUNTCECPDL^{(x,y)}(\phi_1(\{x,y\})\vee\phi_2(\{x,y\}))   &\eqdef \transUNTCECPDL^{(x,y)}(\phi_1) \cup \transUNTCECPDL^{(x,y)}(\phi_2),
    \\
    \transUNTCECPDL^{(x,y)}(\phi_1(\{x,y\})\vee\phi_2(\subseteq\{x\}))   &\eqdef \transUNTCECPDL^{(x,y)}(\phi_1) \cup \left(\transUNTCECPDL(\phi_2)?\circ\Univ\right),
    \\
    \transUNTCECPDL^{(x,y)}(\phi_1(\{x,y\})\vee\phi_2(\subseteq\{y\}))   &\eqdef \transUNTCECPDL^{(x,y)}(\phi_1) \cup \left(\Univ\circ\transUNTCECPDL(\phi_2)?\right),
    \\
    \transUNTCECPDL^{(x,y)}(\phi_1(\subseteq\{x\})\vee\phi_2(\subseteq\{y\}))   &\eqdef \left(\transUNTCECPDL(\phi_1)?\circ\Univ\right) \cup \left(\Univ\circ\transUNTCECPDL(\phi_2)?\right),
    \\
    \transUNTCECPDL^{(x,y)}([\TC_{u,v}\phi(\{u,v\})](x,y)) &\eqdef \left(\transUNTCECPDL^{(u,v)}(\phi)\right)^+,%
    \\
    \transUNTCECPDL^{(x,y)}([\TC_{u,v}\phi(\{u,v\})](y,x)) &\eqdef \left\{\left(\transUNTCECPDL^{(u,v)}(\phi)\right)^+(z_1,z_2)\right\}[y,x],%
    \\
    \transUNTCECPDL([\TC_{u,v}\phi(\{u,v\})](x,x)) &\eqdef \tup{\left\{\left(\transUNTCECPDL^{(u,v)}(\phi)\right)^+(z_1,z_2)\right\}[x,x]},
    \\
    \transUNTCECPDL^{(x,y)}((\exists  z_1,\dots,z_m.\psi)(\{x,y\}))&\eqdef C[x,y],
    \\
    \transUNTCECPDL((\exists  z_1,\dots,z_m.\psi)(\subseteq\{x\}))&\eqdef \tup{C[x,x]},
    \end{align*}
where $C$ is the smallest set such that (we identify the "NF-conjunction" $\psi$ with the set of "NF-atoms" that it contains):
\begin{enumerate}
\item $R(w_1,\dots,w_n)\in C$ if $R(w_1,\dots,w_n)\in \psi$ with $R\in\Rels$,

\item $\epsilon(w_1,w_2)\in C$ if $w_1=w_2\in\psi$,

\item $\transUNTCECPDL(\phi)?(w,w) \in C$ if $\phi(\subseteq\{w\})\in \psi$,

\item $\transUNTCECPDL^{(w_1,w_2)}(\phi)(w_1,w_2) \in C$ if $\phi(\{w_1,w_2\})\in \psi$.
\end{enumerate}
If needed, we also add to $C$ atoms $\Univ(w_1,w_2)$ necessary to satisfy that $\uGraphC{C}$ is connected. Translations $\transUNTCECPDL$ and $\transUNTCECPDL^{(x,y)}$ are computable in polynomial time, and satisfy that for any $\varphi$ in normal form: if $\phi$ has one free variable then $\Cqsizenf{\phi}=\Cqsize{\transUNTCECPDL(\phi)}$, and if $\phi$ has two free variables $x,y$ then $\Cqsizenf{\phi}=\Cqsize{\transUNTCECPDL^{(x,y)}(\phi)}$.
\end{proof}

\noindent
We now turn to the converse translation, from $\UCPDLp$ to $\UNTC$. In contrast with the previous direction, this translation is fully compositional and closely follows the semantics of the language: "formulas" are translated into $\UNTC$-formulas with one free variable, while "programs" are translated into $\UNTC$-formulas with two free variables. In particular, iteration is captured by transitive closure, and "conjunctive programs" are translated into existentially quantified conjunctions.

\begin{proposition}\AP\label{prop:ECPDL-UNTC}
    \AP
    $\UCPDLp \lleq \UNTC$ via a polynomial-time translation $\intro*\transECPDLUNTC: \UCPDLp \to \UNTC$.
\end{proposition}

\begin{proof}
For $x,y\in\Vars$ we define a translation $\transECPDLUNTC^{(x,y)}$ from $\UCPDLp$-programs to $\UNTC$-formulas with only free variables $x$ and $y$, and a translation $\transECPDLUNTC^{x}$ from $\UCPDLp$-formulas to $\UNTC$-formulas with unique free variable $x$ as follows:
    \begin{align*}
    \transECPDLUNTC^{x}(p)&\eqdef p(x),\\
    \transECPDLUNTC^{x}(\lnot\phi)&\eqdef\lnot\transECPDLUNTC^{x}(\phi),\\
    \transECPDLUNTC^{x}(\phi_1\land\phi_2)&\eqdef\transECPDLUNTC^{x}(\phi_1)\land\transECPDLUNTC^{x}(\phi_2),\\
    \transECPDLUNTC^{x}(\tup{\pi})&\eqdef\exists y.\transECPDLUNTC^{(x,y)}(\pi),\\
    \transECPDLUNTC^{(x,y)}(\epsilon)&\eqdef x=y,\\
    \transECPDLUNTC^{(x,y)}(a)&\eqdef a(x,y),\\
    \transECPDLUNTC^{(x,y)}(\bar a)&\eqdef a(y,x),\\
    \transECPDLUNTC^{(x,y)}(\pi_1\cup\pi_2)&\eqdef \transECPDLUNTC^{(x,y)}(\pi_1)\vee\transECPDLUNTC^{(x,y)}(\pi_2),\\
    \transECPDLUNTC^{(x,y)}(\pi_1\circ\pi_2)&\eqdef \exists z.\transECPDLUNTC^{(x,z)}(\pi_1)\land \transECPDLUNTC^{(z,y)}(\pi_2),\\
    \transECPDLUNTC^{(x,y)}(\phi?)&\eqdef x=y\land\transECPDLUNTC^{x}(\phi),\\
    \transECPDLUNTC^{(x,y)}(\pi^*)&\eqdef x=y\vee[\TC_{u,v}\transECPDLUNTC^{(u,v)}(\pi)](x,y),\\
    \transECPDLUNTC^{(x,y)}(\Univ)&\eqdef x=x\land y=y,\\
    \transECPDLUNTC^{(x,y)}(C[x_s,x_t])&\eqdef x=x_s\land y=x_t\land\exists x_1,\dots,x_n. \bigwedge\{ \transECPDLUNTC(\rho)\colon \rho\in C \},
    \end{align*} 
        where $p\in\Prop$, $a\in\Prog$, $\{x_1,\dots,x_n\}=\vars(C)\setminus\{x_s,x_t\}$, and in the last definition, for an "atom" $\rho$ we define its image through $\transECPDLUNTC$ as follows:
    \begin{align*}
    \transECPDLUNTC(R(z_1,\dots,z_m))&\eqdef R(z_1,\dots,z_m),\\
    \transECPDLUNTC(\pi(z,z'))&\eqdef \transECPDLUNTC^{(z,z')}(\pi),
    \end{align*}    
for an "r-atom" $R(z_1,\dots,z_m)$ and a "p-atom" $\pi(z,z')$, with $z,z',z_1,\dots,z_m\in \{x_s,x_t, x_1,\dots,x_n\}$.

It is clear that $\transECPDLUNTC^x$ and $\transECPDLUNTC^{(x,y)}$ are computable in polynomial time. We now justify the correctness of the translation. More precisely, we show by simultaneous induction on the structural complexity of formulas and programs that for every formula $\phi$, every program $\pi$ of $\UCPDLp$, every "$\Rels$-structure" $K$, and every $u,v\in\dom K$, we have
\begin{align*}
K,u\models\phi&\mbox{\quad iff\quad }K\models \transECPDLUNTC^{x}(\phi)[x\mapsto u],
\\
K,u,v\models\pi&\mbox{\quad iff\quad }K\models \transECPDLUNTC^{(x,y)}(\pi)[x\mapsto u,y\mapsto v].
\end{align*}

\noindent
The Boolean cases and the cases for $\epsilon$, $a$, $\bar a$, $\cup$, $\circ$, $?$, and $\Univ$ follow immediately from the semantics of $\UCPDLp$ and the corresponding definitions of the translation. Let us only comment on the non-immediate cases.

For $\pi^*$, recall that $(u,v)\in \dbracket{\pi^*}_K$ iff either $u=v$, or there exist $n\geq 1$ and $u_0,\dots,u_n\in\dom K$ such that $u_0=u$, $u_n=v$, and $(u_i,u_{i+1})\in \dbracket{\pi}_K$ for every $0\leq i<n$. By the induction hypothesis, this is equivalent to saying that either $u=v$, or there is a finite $\transECPDLUNTC^{(z,z')}(\pi)$-path from $u$ to $v$ in $K$, which is in turn equivalent to
$$
K \models [\TC_{z,z'}\transECPDLUNTC^{(z,z')}(\pi)](x,y)[x\mapsto u,y\mapsto v].
$$
Hence
$$
K,u,v\models \pi^* \mbox{\quad iff\quad } K\models \transECPDLUNTC^{(x,y)}(\pi^*)[x\mapsto u,y\mapsto v].
$$

\noindent
For the clause $C[x_s,x_t]$, suppose that $\vars(C)\setminus\set{x_s,x_t}=\set{x_1,\dots,x_n}$. By definition,
$
K,u,v\models C[x_s,x_t]
$
iff there exists an assignment $\alpha$ with $\alpha(x_s)=u$ and $\alpha(x_t)=v$ such that every "atom" $\rho\in C$ is satisfied by $K$ under $\alpha$. If $\rho$ is an "r-atom" $R(z_1,\dots,z_m)$, then by definition $\transECPDLUNTC(\rho)=R(z_1,\dots,z_m)$, so satisfaction is preserved trivially. If $\rho$ is a "p-atom" $\pi(z,z')$, then by the induction hypothesis,
$K,\alpha(z),\alpha(z')\models \pi$ iff $K\models \transECPDLUNTC^{(z,z')}(\pi)[\alpha],$
that is, iff $K\models \transECPDLUNTC(\rho)[\alpha]$. Therefore, $K,u,v\models C[x_s,x_t]$ iff
$$
K\models x=x_s\land y=x_t\land\exists x_1,\dots,x_n.\bigwedge\set{\transECPDLUNTC(\rho)\colon \rho\in C}[x\mapsto u,y\mapsto v],
$$
which is exactly
$$
K\models \transECPDLUNTC^{(x,y)}(C[x_s,x_t])[x\mapsto u,y\mapsto v].
$$
This completes the simultaneous induction, and therefore the correctness of the translation.\pushright{\qedsymbol}%
\let\qed\relax
\end{proof}
\color{black}

\noindent
As a consequence of \Cref{prop:ECPDL-UNTC,prop:UNTC-ECPDL} we obtain
\begin{corollary}\label{cor:UCPDLp-equiv-UNTC}
    $\UCPDLp\langsemequiv\UNTC$.
\end{corollary}

\color{black}

%% file: UNTCsat.tex
\subsection{Satisfiability of \texorpdfstring{\UNTC}{UNFO*}}

We now turn to the satisfiability problem for $\UNTC$. Thanks to the translations established above, the results of \Cref{sec:sat}, for $\UCPDLp$ can be transferred to this setting. In particular, we obtain the following complexity bound.

\begin{thm}\label{thm:UNTC-sat}
    The "satisfiability problem" for $\UNTC$ is decidable, "2ExpTime"-complete.
\end{thm}
\begin{proof}
    The lower bound follows from the "2ExpTime"-hardness of $\UNFO$ satisfiability \cite[Proposition~4.2]{segoufin2013unary}. It is worth stressing that hardness was shown to hold also on "Kripke structures" over a fixed signature.

    For the upper bound, given a $\UNTC$ formula $\phi$, we reduce it to the "satisfiability problem" for $\UCPDLp$ via the translation $\transUNTCECPDL$ of \Cref{prop:UNTC-ECPDL}.
    The $\UCPDLp$ "formula" $\phi' = \transUNTCECPDL(\phi)$ we obtain is a formula of exponential size such that $\Cqsize{\phi'} \leq \untcsize{\phi}$. Further, its "tree-width@@pdl" is bounded by $\untcsize{\phi}$ by \Cref{rk:treewidth-conjunctive-width}.
    Observe also that $\phi'$ has the same set of "atomic propositions" and "atomic programs" as $\phi$ ("ie", the same "relation names").
    
    By applying to $\phi'$ the same series of polynomial-time reductions as in the proof of \Cref{thm:sat-cpdlp} -- namely \Cref{prop:polyredux-simpleform}, \Cref{cor:treewidth-k-model-property} and \Cref{lem:redux-omega-tree-sat} -- we arrive at an instance $(\+T,\psi)$ of the "$\omega$-regular tree satisfiability" problem where: 
    \begin{enumerate}[label=(\roman*),ref=\roman*,align=left,labelwidth=!,labelsep=.5em,leftmargin=!]
        \item $\Cqsize{\psi} \leq \untcsize{\phi}$, the "tree-width@@pdl" of $\psi$ is bounded by $\untcsize{\phi}$, and the number of distinct "atomic programs" and "atomic propositions" of $\psi$ is polynomial in $\Cqsize{\psi}$.
        \item $\sizeStates{\+T}$ is polynomial and $\sizeTrans{\+T}$ is (singly) exponential. This is because the reduction of \Cref{lem:redux-omega-tree-sat} does not blow-up in view of the bounds from the previous item.
        \item $\psi \in \CPDLp$ is of exponential "size@@pdl" but $\Cqsize{\psi}$ is polynomial. 
    \end{enumerate}
    Applying the bound of \Cref{prop:omega-sat-pb-2Exp}, this yields a doubly-exponential time procedure for the satisfiability of $\phi$.
\end{proof}

%% file: conclusions.tex
\section{Conclusions}
\label{sec:conclusions}
We have introduced an expressive logic $\UCPDLp$, which captures several known formalisms and is equivalent to $\UNTC$.
This seems to be a natural and well-behaved generalization, enjoying a good balance of algorithmic properties and expressive power.

\paragraph{Finite Satisfiability} Observe that in the satisfiability problem that we study here, the "Kripke structure"  may be infinite. It is a well-known open problem, even for the case of $\loopCPDL$, whether the \emph{finite}-satisfiability problem is decidable \cite[\S 7]{DBLP:journals/jsyml/GollerLL09}.
\begin{open}
    What is the decidability status of the finite-satisfiability problem for $\loopCPDL$, $\ICPDL$ and $\UCPDLp$?
\end{open}

\diego{THIS IS NEW}
\paragraph{Model checking}
\color{diegochange}
The (finite) ""model checking"" problem for $\UCPDLp$ is the problem of, given a finite "structure" $K$, an "element@world" $w \in \dom{K}$, and a $\UCPDLp$ "formula" $\phi$, whether $K,w \models \phi$. This problem is known to be \ptime-complete for $\PDL$, $\ICPDL$, and many other variants \cite{DBLP:journals/japll/Lange06,DBLP:journals/jcss/FischerL79}. 

Mimicking what happens to "conjunctive queries", we show that if the "tree-width" of $\UCPDLp$ "formulas" is bounded, "model checking" is \ptime.
\begin{thm}\label{thm:modelchecking}
    For any class $\+G$ of connected "graphs":
    \begin{enumerate}
        \item \label{it:modelch:ptime}
        if $\+G \subseteq \Tw$ for some $k$, then the "model checking" problems for $\CPDLp(\+G)$ and for $\UCPDLp(\+G)$ are \ptime-complete;
        \item \label{it:modelch:notptime}
         otherwise, the "model checking" problem for $\CPDLp(\+G)$\sidediego{acá decía UCPDL, pero cambé por CPDL, es más fuerte} is not in \ptime, under the hypothesis that $\+G$ is recursively enumerable and $\wone \neq \fpt$. This holds even for the "positive" fragment of $\CPDLp(\+G)$.
    \end{enumerate}
\end{thm}

\begin{proof}
    \proofcase{\ref{it:modelch:ptime}.}
The procedure for showing that $\UCPDLp(\Tw)$ "model checking" is in \ptime is a classical dynamic programming algorithm. Given a "formula" $\phi$ and a finite "Kripke structure" $K$, we iteratively label all "worlds" of $K$ with "subexpressions" $\psi$ of $\phi$ based on the labeling of "subexpressions" of $\psi$. 
We will use a unary relation ("ie", a set) $U_\psi \subseteq \dom{K}$ for every "formula" $\psi \in sub(\phi)$ and a binary relation $B_\pi \subseteq \dom{K} \times \dom{K}$ for every "program" $\pi \in sub(\phi)$, which are all initialized in the empty relation.
We start by processing all "atomic propositions" $p \in \subexpr(\phi)$ and "atomic programs" $a \in \subexpr(\phi)$: we set
$U_p = p^K$ 
and $B_{a} = {\to_a}$. For the converse of "atomic programs" $\bar a \in sub(\phi)$ we set $B_{\bar a} = (\to_a)^{-1}$.
Now take any "formula" $\psi \in sub(\phi)$ such that all its "subexpressions" from $sub(\psi) \setminus \set\psi$ have already been processed. We process $\psi$ as follows:
\begin{itemize}
    \item If $\psi = \psi_1 \land \psi_2$, then set $U_\psi = U_{\psi_1} \cap U_{\psi_2}$,
    \item if $\psi = \lnot\psi'$, then set $U_{\psi} = \dom{K} \setminus U_{\psi'}$,
    \item if $\psi = \tup{\pi}$, then set $U_{\psi} = \set{w \in \dom{K} : \exists w' s.t.\ (w,w') \in B_{\pi}}$.
\end{itemize}
It is easy to see that each one of these operations is in polynomial time (even linear with the right data structure).
Now take any "program" $\pi \in sub(\phi)$ such that all its "subexpressions" from $sub(\pi) \setminus \set\pi$ have already been processed. We process $\pi$ as follows:
\begin{itemize}
    \item If $\pi=\Univ$ then set $B_{\pi}=\dom{K} \times \dom{K}$,
    \item If $\pi = \pi_1 \star \pi_2$ for $\star \in \set{\circ,\cup}$, then set $B_\pi = B_{\pi_1} \star B_{\pi_2}$, 
    \item if $\pi = (\pi')^*$, then set $B_\pi = B_{\pi'}^*$,
    \item if $\pi = C[x_s,x_t]$, we evaluate $C[x_s,x_t]$ as if it were a "conjunctive query" on the already processed relations to populate $B_\pi$. 
\end{itemize}
It is clear that the first two items can be done in quadratic time. For the last item, it is well-known that the evaluation of "conjunctive queries" of "tree-width" $\leq k$ is in polynomial time \cite[Theorem~3]{DBLP:journals/tcs/ChekuriR00} (more precisely "LOGCFL"-complete \cite[Theorem~6.12]{DBLP:journals/jacm/GottlobLS01}), via a bottom-up processing of the "tree decomposition" of the query.

All in all, this yields an algorithm which is linear in $\phi$ and polynomial in $K$, where the degree of the polynomial is $k+1$ if $\phi \in \UCPDLg{\Tw}$. 

The lower bound comes from \ptime-hardness of "model checking" for modal logic (see "eg", \cite[Proposition 5]{DBLP:journals/japll/Lange06}).

\proofcase{\ref{it:modelch:notptime}.} This follows from the fact that a similar statement is known for "conjunctive queries": the "model checking" problem  for the class of Boolean "conjunctive queries" on binary relations whose underlying graph is in $\+G$ is not in \ptime unless $\fpt = \wone$ \cite[Corollary~19]{DBLP:conf/stoc/GroheSS01}.\footnote{The fact that here we have connected "conjunctive queries" plays no role, and the result of Grohe \cite{DBLP:conf/stoc/GroheSS01} holds also for connected "conjunctive queries".} 
We have that  a Boolean "conjunctive query" $q$ holds true in a "structure@@kripke" if{f} the "formula" $\tup{C[x,x]}$ is satisfied in some of its "worlds", where
$C$ is the set of atoms of $q$, and $x$ is any variable of~$q$.

This \ptime-reduction from "model checking" of "conjunctive queries" to "model checking" of $\CPDLp(\+G)$ implies that the "model checking" of $\CPDLp(\+G)$ cannot be in \ptime, under the hypothesis that $\wone \neq \fpt$.\footnote{A similar parallel can be made for parameterized "model checking", where the parameter is the size of the "formula". In our case we would obtain that parameterized "model checking" of $\CPDLp(\+G)$ is \wone-hard.}
\end{proof}

\color{black}

\paragraph{Infinite State Model Checking}
Infinite state model checking (\ie, "model checking" for some finitely represented infinite "Kripke structures") has been studied for $\CPDL$ \cite{DBLP:conf/csl/GollerL06} and $\ICPDL$ \cite{DBLP:journals/jsyml/GollerLL09}. While we believe that the upper bound results can be extended to $\UCPDLp$, we leave this for future work.

\paragraph{Constants} Since this work intends to capture some query languages such as "CQs" and "CRPQs", it would make sense to have access to \emph{constants}, which in this context are usually modelled as \emph{nominals} (\ie, a kind of "atomic propositions" which can hold true in at most one "world"). We believe that nominals can be treated easily in $\CPDLp$ at the expense of increasing the "tree-width" in the tree-like model property.